\documentclass[11pt,a4paper]{article}
\usepackage{amsmath,amsfonts,amssymb,mathtools,amsthm,newlfont,graphicx,amscd,bbm,enumerate,galois,mathrsfs,xypic,geometry,bbm,xypic}
\usepackage{simplewick}
\usepackage{booktabs}
\usepackage{multirow}
\usepackage{bigstrut}
\usepackage{tabularx,extarrows}

\usepackage{color}

\newcommand{\lmd}{\lambda}

\newcommand{\Lmd}{\Lambda}

\newcommand{\Omg}{\Omega}

\newcommand{\sgm}{\sigma}
\newcommand{\p}{\partial}
\newcommand{\al}{\alpha}
\newcommand{\afa}{\alpha}

\newcommand{\veps}{\varepsilon}
\newcommand{\fai}{\varphi}
\newcommand{\rmi}{{\mathrm i}}
\newcommand{\rme}{{\mathrm e}}
\newcommand{\rmT}{{\mathrm T}}
\newcommand{\rmK}{{\mathrm K}}
\newcommand{\rmI}{{\mathrm I}}
\newcommand{\gfm}{M(\eta, c,e)}
\newcommand{\gfmh}{M(\hat{\eta}, c,e)}
\newcommand{\gfme}{M(\eta,c,e,E)}
\newcommand{\gfmhe}{M(\hat{\eta},c,e,E)}

\newcommand{\mcalB}{\mathcal{B}}

\newcommand{\mcalD}{\mathcal{D}}
\newcommand{\mcalF}{\mathcal{F}}

\newcommand{\mcalH}{\mathcal{H}}
\newcommand{\mcalI}{\mathcal{I}}
\newcommand{\mcalJ}{\mathcal{J}}

\newcommand{\mcalL}{\mathcal{L}}

\newcommand{\bbC}{\mathbb{C}}
\newcommand{\bbP}{\mathbb{P}}
\newcommand{\bbZ}{\mathbb{Z}}
\newcommand{\bft}{\mathbf{t}}

\newcommand{\rmP}{\mathrm{P}}
\newcommand{\bfa}{\mathbf{a}}

\newcommand{\bfr}{\boldsymbol{r}}

\newcommand*{\pp}[1]
  {\frac{\partial   }
        {\partial #1}
  }

\newcommand*{\pfrac}[2]
  {\frac{\partial #1}
        {\partial #2}
  }

\newcommand*{\pair}[2]
  {
    \left\langle
      #1,#2
    \right\rangle
  }

  \newcommand*{\Bigset}[2]
  {
   \left\{ #1 \middle| #2 \right\}
  }

\newcommand{\beq}{\begin{equation}}
\newcommand{\eeq}{\end{equation}}

\DeclareMathOperator{\res}{Res}
\DeclareMathOperator{\diag}{diag}
\DeclareMathOperator{\Vect}{Vect}

\DeclareMathOperator{\grad}{grad}

\DeclareMathOperator{\td}{d\!}

\newtheorem{thm}{Theorem}[section]

\newtheorem{rmk}[thm]{Remark}

\newtheorem{cor}[thm]{Corollary}
\newtheorem{lem}[thm]{Lemma}

\newtheorem{prop}[thm]{Proposition}
\newtheorem{defn}[thm]{Definition}
\newtheorem{ex}[thm]{Example}

\numberwithin{equation}{section}

\allowdisplaybreaks[4]

\begin{document}


\title{Legendre transformations of a class of generalized Frobenius manifolds and the associated integrable hierarchies}
\author{Si-Qi Liu, Haonan Qu, Youjin Zhang}
\date{\today}
\maketitle

\begin{abstract}

For two generalized Frobenius manifolds related by a Legendre-type transformation, we show that the associated integrable hierarchies of hydrodynamic type, which are called the Legendre-extended Principal Hierarchies, are related by a certain linear reciprocal transformation; we also show, under the semisimplicity condition, that the topological deformations of these Legendre-extended Principal Hierarchies are  related by the same linear reciprocal transformation.
\end{abstract}


\tableofcontents

\section{Introduction}

Since Dubrovin introduced the notion of Frobenius manifold in the early1990s \cite{2dTFT,Dubrovin-FM}, this geometric structure has been playing important roles in different research subjects of mathematical physics,
such as Gromov-Witten theory,
singularity theory and integrable systems, see for example
\cite{BLS, Buryak2015, normal-form, FJR, Givental-1, Givental-2, GiventalMilanov, Hertling-book, Kontsevich-Manin, LWZ1, Ruan-Tian} and references therein.
In applications of Frobenius manifold to these research subjects
some generalizations of this geometric structure arise \cite{Lorenzoni-1, Lorenzoni-2, Getzler-04, Manin-Hertling, GFM2, GFM1, Manin}. Among such generalizations we have the so called generalized Frobenius manifolds with non-flat unity \cite{GFM2, GFM1},
which satisfy all the axioms of Dubrovin's definition of Frobenius manifold but the flatness condition imposed on the unit vector fields.
In this paper, we refer to them simply as \textit{generalized Frobenius manifolds}. Examples of such generalized Frobenius manifolds can be found in \cite{Brini,r-spin, GFM2, double-Hurwitz}.

It is shown in \cite{GFM2} that one can associate with any generalized Frobenius manifold a bihamiltonian integrable hierarchy of hydrodynamic type, which is called the Principal Hierarchy of the generalized Frobenius manifold and is an analogue of the one for a usual Frobenius manifold \cite{normal-form}.
Such an integrable hierarchy  possesses a tau structure, and the tau-cover of the Principal Hierarchy admits a family of Virasoro symmetries.
The condition of linearization of actions of the Virasoro symmetries on the tau function of the Principal Hierarchy leads to the loop equation of the generalized Frobenius manifold, which is shown to have a unique solution under the assumption of semisimplicity \cite{GFM2}. The solution of the loop equation yields a quasi-Miura transformation which transforms the Principal Hierarchy of the semisimple generalized Frobenius manifold to its topological deformation.
In \cite{GFM2,GFM1,ExtendedAL} the topological deformations of the Principal Hierarchies of two particular examples of generalized Frobenius manifolds are studied, it turns out that they contain the Volterra hierarchy, the $q$-deformed KdV hierarchy and the Ablowitz-Ladik hierarchy, which are discrete integrable hierarchies well known in the theory of soliton equation.

In the present paper, we study the relationship between the integrable hierarchies associated with two generalized Frobenius manifolds which are related by a Legendre-type transformation. Such a kind of transformations for Frobenius manifolds were originally introduced by Dubrovin in \cite{2dTFT} based on symmetries of the WDVV equations of associativity.  Each Legendre-type transformation is determined by an invertible flat vector field of the Frobenius manifold that is under consideration, and it transforms the given Frobenius manifold to another one.
As it is shown in \cite{YD}, these two Frobenius manifolds
share the same monodromy data, and their Principal Hierarchies
are related by a certain linear reciprocal transformation.
Moreover, if these two Frobenius manifolds are semisimple,
the topological deformations of their Principal Hierarchies are also related by the same linear reciprocal transformation, and share the same deformed tau structure. In \cite{Strachan-Stedman}, Strachan and Stedman generalized Dubrovin's Legendre-type transformations of Frobenius manifolds. Such a generalized Legendre-type transformation is determined by an invertible vector field, which satisfies a certain set of defining conditions and is called a Legendre field. In general, a Legendre field is
non-flat with respect to the flat metric of the Frobenius manifold under consideration, and the associated Legendre-type transformation transforms the Frobenius manifold to a generalized one with non-flat unity,
moreover, one can apply such a transformation to any generalized Frobenius manifold.

In order to establish a relationship between the integrable hierarchies associated with two generalized Frobenius manifold that are related by a generalized Legendre-type transformation (which we will simply call a Legendre-type transformation), we need to extend the Principal Hierarchies of the generalized Frobenius manifolds by some flows that
are constructed from the Legendre field, and we call these extended integrable hierarchies the Legendre-extended Principal Hierarchies of the generalized Frobenius manifolds. We show that these Legendre-extended Principal Hierarchies of the generalized Frobenius manifolds are related by a linear reciprocal transformation, and they also possess tau-covers and Virasoro symmetries. We also show that the corresponding linearization conditions of the actions of the Virasoro symmetries on the tau function of the Legendre-extended Principal Hierarchies lead to the same loop equations that are derived in \cite{GFM1}. In the case when the generalized Frobenius manifolds are semisimple, we show that the topological deformations of their Legendre-extended Principal Hierarchies are also related by the same linear reciprocal transformation.

The paper is organized as follows.
In Sect.\,2, we recall the notion and main properties of Legendre-type transformations introduced in \cite{Strachan-Stedman}.
In Sect.\,3, we construct the Legendre-extended Principal Hierarchy and its tau structures for a generalized Frobenius manifold,
and establish the relationship of the Legendre-extended Principal Hierarchies of two generalized Frobenius manifolds that are related by a Legendre-type transformation.
In Sect. 4, we study the Virasoro symmetries and the topological deformations of the Legendre-extended Principal Hierarchies,
and establish the relationships between these integrable hierarchies
under the semisimplicity condition.
In Sect.\,5 and Sect.\,6, we present two important examples,
one relates the KdV hierarchy with the $q$-deformed KdV hierarchy,
while the other one relates the Toda hierarchy with the Ablowitz-Ladik hierarchy.
In Sect.\,7, we give some concluding remarks.

\section{The generalized Legendre transformations}
\subsection{Legendre fields}

Let $M$ be an $n$-dimensional smooth manifold endowed with a generalized Frobenius manifold (GFM) structure $(\eta, c, e)$, where $\eta$ is a flat metric on $M$, $c\colon TM\times TM\to TM$ is a tensor field of $(1,2)$-type, $e$ is a vector field on $M$, and the triple $(\eta, c, e)$ is required to satisfy the following conditions:
\begin{enumerate}
 \item For each $p\in M$, $(\eta, c, e)$ yields a Frobenius algebra structure on $T_pM$ with unity $e(p)$.
  \item Denote by $\nabla$ the Levi-Civita connection of $\eta$, and by $\tilde{c}$ the $(0,3)$-tensor defined by $\tilde{c}(X,Y,Z)=\langle X\cdot Y,Z\rangle$, then $\nabla \tilde{c}$ is a symmetric 4-tensor.
Here we use the short notations
  \begin{equation*}
    X\cdot Y:=c(X,Y),\quad \langle X, Y\rangle:=\eta(X,Y),\quad
    X, Y\in\Vect(M).
  \end{equation*}

\end{enumerate}
Note that the flatness condition $\nabla e = 0$ imposed on the unit vector field $e$ and the existence of an Euler vector field in Dubrovin's definition of a Frobenius manifold \cite{2dTFT} are not assumed here, and this notion of generalized Frobenius manifold
coincides with the one given in \cite{GFM1} except for the existence of an Euler vector field $E$.

In a system of flat coordinates $\{v^\afa\}_{\afa=1}^n$
of the metric $\eta$ in a neighborhood of a point $p\in M$
we have
\begin{equation*}
  \eta = \eta_{\afa\beta}\td v^\afa\otimes\td v^\beta,\qquad
  c= c^\gamma_{\afa\beta}(v)\,\p_\gamma\otimes\td v^\afa\otimes\td v^\beta,
\end{equation*}
here and in what follows summation over repeated upper and lower Greek indices with range from $1$ to $n$ is assumed,
and $\p_\gamma=\frac{\p}{\p v^\gamma}$.
Then the symmetry property of the tensor field $\nabla \tilde{c}$ implies that there exists locally a function $F(v)$ such that
\begin{equation}\label{240516-1208}
 \eta_{\al\xi} c^\xi_{\beta\gamma}(v)=\frac{\p^3 F(v)}{\p v^\afa\p v^\beta\p v^\gamma},\quad \al,\beta,\gamma=1,\dots,n,
\end{equation}
and the associativity of the Frobenius algebra is equivalent to the fact that $F(v)$ satisfies
the WDVV equations
\begin{equation}\label{WDVV}
\frac{\p^3 F(v)}{\p v^\afa \p v^\beta\p v^\lmd}
\eta^{\lmd\mu}\frac{\p^3 F(v)}{\p v^\mu \p v^\gamma\p v^\delta}
=
\frac{\p^3 F(v)}{\p v^\delta \p v^\beta\p v^\lmd}
\eta^{\lmd\mu}\frac{\p^3 F(v)}{\p v^\mu \p v^\gamma\p v^\afa}
,
\end{equation}
for $\al,\beta,\gamma, \delta=1,\dots,n$.
As it is shown in \cite{2dTFT}, the WDVV equations possess two types of symmetries. The first type of symmetries is given by \textit{Legendre-type transformations}
(see also in \cite{BihamCoh2, normal-form, An-Bn-V}),
which is generalized in \cite{Strachan-Stedman}.
Another type of symmetries consists of \textit{inversion symmetries}
(see also in \cite{Inversion-sym}),
which will not be considered in the present paper.

\begin{defn}[\cite{Strachan-Stedman}]
  Let $\gfm$ be a generalized Frobenius manifold.
  A solution $B\in\Vect(M)$ to the equation
  \begin{equation}\label{Legendre sym}
    X\cdot\nabla_YB = Y\cdot\nabla_XB,\quad \forall\,X,Y\in\Vect(M)
  \end{equation}
is called a Legendre field.
\end{defn}
Taking $Y=e$ in \eqref{Legendre sym}, we arrive at
\begin{equation}\label{basic prop Y=e}
  \nabla_XB=X\cdot\nabla_eB,\quad \forall X\in\Vect(M)
\end{equation}
for any Legendre field $B$.
In the flat coordinates $v^1,\dots, v^n$, the equation \eqref{Legendre sym} can be rewritten as
\begin{equation}\label{LegendreSym1}
  c_{\afa\gamma}^\delta\p_\beta B^\gamma
 =c_{\beta\gamma}^\delta\p_\afa B^\gamma,\quad \forall \al,\beta,\delta=1,\dots,n
\end{equation}
for a vector field $B=B^\gamma\p_\gamma$. These equations are also equivalent to
\begin{equation}\label{LegendreSym2}
  \p_\afa(B^\gamma c_{\beta\gamma}^\delta) = \p_\beta(B^\gamma c_{\afa\gamma}^\delta).
\end{equation}
From the above definition and the equations \eqref{LegendreSym2}
it follows that we have the following two examples of Legendre fields:
\begin{enumerate}
  \item If $B\in\Vect(M)$ satisfies the condition $\nabla B=0$, then $B$ is a Legendre field. We call such a Legendre field a flat Legendre field.
  \item The unit vector field $e=e^\afa\p_\afa$ is a Legendre field.
\end{enumerate}

\begin{rmk}\label{rmk:grad field}
Multiplying $e_\delta:=\eta_{\afa\delta}e^\afa$ on both sides of \eqref{LegendreSym1}, we obtain
\[
  \eta_{\afa\gamma}\p_\beta B^\gamma = \eta_{\beta\gamma}\p_\afa B^\beta,
\]
which implies that the Legendre field $B$ is locally a gradient field,
i.e., there locally exists a function $\fai$ on $M$ such that
\begin{equation}
  B=\grad_\eta\fai := \pfrac\fai{v^\afa}\eta^{\afa\beta}\pp{v^\beta}.
\end{equation}
In particular, the unity $e$ is locally a gradient field (see also in \cite{GFM1}).
\end{rmk}

For any vector field $B$ on $\gfm$, introduce a new $(0,2)$-tensor $\hat\eta=\langle\ ,\,\rangle_B$ by
\begin{equation} \label{pair-B}
  \langle X,Y\rangle_B:=\langle B\cdot X, B\cdot Y\rangle,\quad X,Y\in\Vect(M).
\end{equation}
We say $B$ is \textit{invertible}, if there exists a vector field $B^{-1}$ such that $B\cdot B^{-1}=e$.
Note that if $B$ is invertible, then $\langle\ ,\,\rangle_B$ is non-degenerate.

\begin{prop}[Generalized Legendre transformation \cite{Strachan-Stedman}]\label{zh-1}
Suppose $B=B^\afa\p_\afa$ be an invertible Legendre field on a generalized Frobenius manifold $\gfm$,  then the following statements hold true:
\begin{enumerate}
  \item $\gfmh$ is also a generalized Frobenius manifold,
  and the Levi-Civita connection $\hat\nabla$ with respect to the new metric $\hat\eta=\langle\ ,\,\rangle_B$ satisfies the relation
  \begin{equation} \label{new conenction hat nabla}
    \hat\nabla_XY=B^{-1}\cdot\nabla_X(B\cdot Y),\quad \forall X,Y\in\Vect(M).
  \end{equation}
    \item A system of flat coordinates $\{\hat v^\afa\}_{\afa=1}^n$ with respect to the new metric $\hat\eta$
  can be chosen as a solution to the equations
  \begin{equation}\label{hat v-afa1}
    \pfrac{\hat v^\afa}{v^\beta} = B^\gamma(v) c^\afa_{\beta\gamma}(v),\quad \al,\beta=1,\dots,n.
  \end{equation}
  Moreover, we have
  \begin{equation}\label{hat v-afa2}
    \hat\eta = \eta_{\afa\beta}\td \hat v^\afa\otimes \td \hat v^\beta,
    \quad e=B^\afa(v)\pp{\hat v^\afa},
    \quad
    \pp{v^\afa} = B(v)\cdot\pp{\hat v^\afa}.
  \end{equation}
  \item There locally exists a function $\hat F=\hat F(\hat v)$ on $M$ such that
  \[
    \frac{\p^3\hat F}{\p\hat v^\afa\p\hat v^\beta\p\hat v^\gamma}
 =
    \left\langle
      \pp{\hat v^\afa}\cdot\pp{\hat v^\beta} ,
      \pp{\hat v^\gamma}
    \right\rangle_B,\quad \al,\beta,\gamma=1,\dots,n,
  \]
  and such $\hat F(\hat v)$ can be chosen as a solution to the equation
  \begin{equation}
    \frac{\p^2\hat F(\hat v)}{\p\hat v^\afa\p\hat v^\beta}
   =\frac{\p^2 F(v)}{\p v^\afa\p v^\beta},\quad \al,\beta=1,\dots,n,
  \end{equation}
  where $F(v)$ is defined as in \eqref{240516-1208}.
  \item $B^{-1}$ is a Legendre field on $\gfm$,
  which transforms $\hat\eta$ back to $\eta$.
\end{enumerate}
\end{prop}

Note that the difference between $\gfm$ and $\gfmh$
lies only in their metrics $\eta$ and $\hat\eta$, and
they share the same Frobenius multiplication $c$ and the unit vector field $e$.

\subsection{Quasi-homogeneity}

We say that the generalized Frobenius manifold $\gfm$ is \textit{quasi-homogeneous},
if there exists a vector field $E=E^\afa\p_\afa$ and a constant $d\in\bbC$, called the \textit{Euler vector field} and the \textit{charge} of $M$ respectively, such that $\nabla\nabla E=0$, and
The Lie derivatives of $c$ and $\eta$ along $E$ satisfy the relations
  \begin{equation}\label{quasi-homog of c and eta}
    \mcalL_E c=c, \quad
    \mcalL_E\eta = (2-d)\eta.
  \end{equation}
In other words, $\gfme$ is a generalized Frobenius manifold of charge $d$ in the sense of \cite{GFM1}.
We assume in what follows that the Euler vector field $E$ is diagonalizable, so that we can choose a system of flat coordinates $\{v^\afa\}$ of $\eta$ such that
  \begin{equation}\label{DDE=0}
    E = \sum_{\afa=1}^{n}
      \left[
        \left(
          \frac{2-d}{2}-\mu_\afa
        \right)v^\afa+r^\afa
      \right]\pp{v^\afa},
  \end{equation}
where $\mu:=\diag(\mu_1,...,\mu_n)$ is a part of the monodromy data of $M$ at $z=0$ \cite{2dTFT, Painleve, normal-form, GFM1}, and
the constants $r^\afa\ne 0$ only if $\mu_\afa+\frac d2=1$.

\begin{defn}[\cite{BihamCoh2}]
  Suppose $\gfme$ is a generalized Frobenius manifold of charge $d$, a Legendre field $B\in\Vect(M)$ is called quasi-homogeneous if
  \begin{equation}\label{quasi-homog b}
    [E,B]=\left(\mu_B-\frac{2-d}{2}\right)B
  \end{equation}
  for a certain constant $\mu_B\in\bbC$.
\end{defn}

We note that a flat Legendre field $B=\pp{v^\afa}$ is quasi-homogeneous,
since
\[\Bigl[E,\pp{v^\afa}\Bigr]=\left(\mu_\afa-\frac{2-d}{2}\right)\pp{v^\afa},\]
and the constant $\mu_B$ coincides with $\mu_\afa$.
The unit vector field $e$ is also quasi-homogeneous,
because of the identity $[E,e]=-e$ (see \cite{GFM1}),
and $\mu_e=-\frac d2$ which
is also denoted by $\mu_0$ in \cite{GFM1}.

\begin{lem}
  Suppose an invertible vector field $B$ on $\gfme$ satisfies the relation \eqref{quasi-homog b},
  then its inverse $\hat B := B^{-1}$ satisfies
  \begin{equation}\label{b^-1 quasi-homog}
    \left[E,\hat B\right] =
    -\left(\mu_B+\frac{2+d}{2}\right)\hat B.
  \end{equation}
\end{lem}

\begin{proof}
  Using the identity $[E,e]=-e$ and \eqref{quasi-homog of c and eta}, we have
  \begin{align*}
    -e &=\, [E,e] = \mcalL_E\left(B\cdot \hat B\right) \\
    &=\,
      (\mcalL_EB)\cdot \hat B + (\mcalL_E c)(B, \hat B)
     +B\cdot(\mcalL_E \hat B) \\
    &=\,
      \left(\mu_B-\frac{2-d}{2}\right)e+e+B\cdot[E,\hat B],
  \end{align*}
from which it follows \eqref{b^-1 quasi-homog}. The lemma is proved.
\end{proof}

\begin{prop}\label{prop:linearity of new E}
  Suppose $B$ is an invertible quasi-homogeneous Legendre field on the generalized Frobenius manifold $\gfme$ of charge $d$,
then $\gfmhe$ is a generalized Frobenius manifold of charge
  \begin{equation} \label{hat d}
    \hat d = -2\mu_B.
  \end{equation}
\end{prop}

\begin{proof}
By using Lemma 2.7 of \cite{Strachan-Stedman} we have
  \[
    \mcalL_E\hat\eta = (2-\hat d)\hat\eta.
  \]
  Therefore, we only need to verify that
  $E$ has local form  \eqref{DDE=0} in a certain system of flat coordinates $\{\hat v^\afa\}$ of $\hat\eta$.
  Let $\{\hat v^\afa\}$ be the system of coordinate defined by \eqref{hat v-afa1}, then from \eqref{hat v-afa2}, \eqref{quasi-homog of c and eta} and \eqref{b^-1 quasi-homog} it follows that
  \begin{align*}
    &\,\left[E,\pp{\hat v^\afa}\right]
  =
    \mcalL_E\left(B^{-1}\cdot\pp{v^\afa}\right) \\
  =&\,
    \left[ E, B^{-1}\right]\cdot\pp{v^\afa}
   +(\mcalL_E c)\left(B^{-1},\pp{v^\afa}\right)
   +B^{-1}\cdot\left[E, \pp{v^\afa}\right] \\
  =&\,(\mu_\afa-\mu_B-1)\pp{\hat v^\afa}
    = \left(\mu_\afa-\frac{2-\hat d}{2}\right)\pp{\hat v^\afa}.
  \end{align*}
  On the other hand, in the coordinate $\hat v^1,\dots,\hat v^n$ the Euler vector field
  $E=\hat E^\afa\pp{\hat v^\afa}$ satisfies the relation
 \[\left[E,\pp{\hat v^\afa}\right]=-\pfrac{\hat E^\beta}{\hat v^\afa}\pp{\hat v^\beta},\]
  therefore we have
\[\pfrac{\hat E^\beta}{\hat v^\afa}
   =\left(\frac{2-\hat d}{2}-\mu_\afa\right)\delta^\beta_\afa,\]
   thus
  \[
    E=\sum_{\afa=1}^{n}
      \left[
        \left(\frac{2-\hat d}{2}-\hat \mu_\afa\right)\hat v^\afa + \hat r^\afa
      \right]\pp{\hat v^\afa}
  \]
  for certain constants $\hat r^\afa$, where $\hat\mu_\afa:=\mu_\afa$.
  If $\hat\mu_\afa+\frac{\hat d}{2}\neq 1$,
 we can kill $\hat r^\afa$ by a certain coordinate translation $\hat v^\afa\mapsto \hat v^\afa + \text{const}$. The proposition is proved.
\end{proof}

\begin{rmk}\label{rmk:hat d}
The relation \eqref{b^-1 quasi-homog} can be rewritten as
\begin{equation}
  [E, \hat B] = \left( \hat\mu_{\hat B}-\frac{2-\hat d}{2} \right)\hat B
\end{equation}
in terms of the new charge $\hat d$ given in \eqref{hat d}
and $\hat{\mu}_{\hat B}:= -\frac d2$.
\end{rmk}

\section{The Legendre-extended Principal Hierarchy}
Given a family of Legendre fields $\mcalB=\Bigset{B_j}{j\in \mcalJ}$, we are to show in this section that one can construct from it a hierarchy of pairwise commuting flows, and that the the Principal Hierarchy \cite{normal-form, GFM1} of a generalized Frobenius manifold can be extended by using such flows.

\subsection{The Legendre flows and the Principal Hierarchy}
\begin{defn}
 Suppose $B=B^\afa\p_\afa$ is a Legendre field
 (not necessarily invertible, quasi-homogeneous) of $M(\eta, c, e)$,
 then the Legendre flow $\pp{t^B}$ generated by $B$ is defined as
 \begin{equation}\label{Legendre flow}
   \pfrac{v^\afa}{t^B} = B^\gamma c^\afa_{\gamma\beta}v^\beta_x,\quad \al=1,\dots,n.
 \end{equation}
\end{defn}

We note that the relation \eqref{LegendreSym2} implies that there locally exists functions $\{v_B^\afa\}_{\afa=1}^n$ such that
\begin{equation}\label{v_B^afa}
B^\gamma c^\afa_{\gamma\beta}=\p_\beta v_B^\afa.
\end{equation}
If $B$ is invertible, then $v^\afa_B$ coincides with the coordinate
$\hat v^\afa$ that is defined in \eqref{hat v-afa1} up to the addition of a constant.
In terms of the notation of $v^\afa_B$, the Legendre flow \eqref{Legendre flow} can be rewritten as
\begin{equation}\label{Legendre flow2}
  \pfrac{v^\afa}{t^B} = (v^\afa_B)_x.
\end{equation}

\begin{lem}
  Let $B_1, B_2$ be any two given Legendre fields of $M(\eta, c,e)$,
  then the associated Legendre flows are commutative, i.e.,
  $\left[\pp{t^{B_1}}, \pp{t^{B_2}}\right] = 0$.
\end{lem}

\begin{proof}
  From \eqref{Legendre flow} and \eqref{Legendre flow2} it follows that
  \begin{align*}
    \pp{t^{B_1}}\left(\pp{t^{B_2}}v^\afa\right)
  &=\,
    \pp{t^{B_1}}\left(\p_x v^\afa_{B_2}\right)
  =
    \p_x\left(
      \pfrac{v_{B_2}^\afa}{t^{B_1}}
    \right)
  =
    \p_x\left(
      \pfrac{v_{B_2}^\afa}{v^\beta}\pfrac{v^\beta}{t^{B_1}}
    \right)\\
  &=\,
    \p_x\left(
      B_2^\gamma c^\afa_{\beta\gamma}B_1^\xi c^\beta_{\xi\sigma} v^\sigma_x
    \right)
  =
   \p_x\left(
     B_2\cdot (B_1\cdot v_x)
   \right)^\afa,
  \end{align*}
where $v_x:= v^\afa_x\pp{v^\afa}$.
Therefore the commutativity and associativity of the Frobenius multiplication $c$ imply that
$\left[\pp{t^{B_1}}, \pp{t^{B_2}}\right]v^\afa = 0$, the lemma is proved.
\end{proof}

As a consequence, a family of Legendre fields $\Bigset{B_j}{j\in\mcalJ}$
yields a hierarchy of mutually commutative flows $\Bigset{\pp{t^{B_j}}}{j\in\mcalJ}$.
An important example of such integrable hierarchies is the Principal Hierarchy (see details in \cite{GFM1})
\begin{equation}\label{PH}
  \pfrac{v^\afa}{t^{i,p}}
 =\eta^{\afa\beta}\p_x
  \left(
  \pfrac{\theta_{i,p+1}}{v^\beta}
  \right) ,\quad (i,p)\in\mcalI,
\end{equation}
of a generalized Frobenius manifold
$M(\eta, c,e,E)$,
where the index set
\begin{equation}\label{zh-5}
  \mcalI :=
    \Big(
      \{1,2,...,n\}\times\bbZ_{\geq 0}
    \Big)\cup
    \Big(
      \{0\}\times\bbZ
    \Big).
\end{equation}
The functions $\{\theta_{i,p}\}_{(i,p)\in\mcalI}$ on $M$ satisfy the relations
\begin{align}
  &\theta_{\afa,0}=v_\afa:=\eta_{\afa\beta}v^\beta,\quad
  \grad_\eta\theta_{0,0}=e,\quad \al=1,\dots,n,\\
  \label{def of theta afa p}
 & \theta_{\afa,p}=\sum_{k=0}^{p}(-1)^k
    \pair{\grad_\eta\theta_{0,k+1}}{\grad_\eta\theta_{\afa,p-k}},
    \quad \al=1,\dots,n,\,p\ge 0,\\
\label{d-recursion}
 &\p_\afa\p_\beta\theta_{i,p+1}=c^\gamma_{\afa\beta}\p_\gamma\theta_{i,p},\quad \al, \beta=1,\dots,n,\, (i,p)\in\mcalI.
\end{align}
These functions also satisfy the following quasi-homogeneous conditions
\begin{align}
  \mcalL_E\theta_{\afa,p}
&=\,
  \left(
    p+\mu_\afa + \frac{2-d}{2}
  \right)\theta_{\afa,p}
 +\sum_{s=1}^{p}\theta_{\veps,p-s}(R_s)^\veps_\afa
 +(-1)^pr^\veps_{p+1}\eta_{\veps\afa},  \label{qhomog for theta-1}\\
\mcalL_E\theta_{0,q} &=\,
  \left(q+1-d\right)\theta_{0,q}
 +\sum_{s=1}^{q}
    \theta_{\veps, q-s}r^\veps_s
 +a_{00;q+1}                      \label{qhomog for theta-2}
\end{align}
for $1\leq\afa\leq n$, $p\geq 0$ and $q\in\bbZ$,
where $d$ is the charge of the generalized Frobenius manifold, and
\begin{equation}\label{mu,R,r,c}
\Bigset{(R_s)^\afa_\beta,\, r_s^\afa,\, a_{00; q}}{1\leq\afa,\beta\leq n,\, s\geq 1,\, q\in\bbZ}
\end{equation}
is a certain family of constants which satisfy the relations
\begin{align}
&R_s^{\rmT}\eta = (-1)^{s+1}\eta R_s,\quad
    [\mu, R_s] = sR_s,\label{Rk-1}\\
&r_s^\afa \neq 0\quad\text{only if}\ \mu_\afa + \frac d2 = s, \label{Rk-3}\\
&a_{00; q}\neq 0 \quad\text{only if}\ q=d,\,\text{and}\, q\,\text{is an odd integer}, \label{Rk-4}
\end{align}
and $r_1^\afa$ coincide with the coefficients $r^\afa$ that appear in the expression \eqref{DDE=0} of the Euler vector field $E$.
We also note that the flow $\pp{t^{0,0}}$ defined in \eqref{PH}  is
given by the translation along the spatial variable $x$, i.e.,
\[
  \pfrac{v^\afa}{t^{0,0}} = v^\afa_x,
\]
so in what follows we will identify the time variable $t^{0,0}$ with the spatial variable $x$, i.e.,
\begin{equation}\label{t00=x}
  t^{0,0}=x.
\end{equation}

Introduce the gradient fields
\begin{equation} \label{xi=grad theta}
  \xi_{i,p}:=\grad_\eta\theta_{i,p}
  = \eta^{\afa\beta}\pfrac{\theta_{i,p}}{v^\beta}
    \pp{v^\afa}, \quad (i,p)\in\mcalI,
\end{equation}
then we have
\begin{equation}\label{xi-initial}
  \xi_{\afa,0}=\pp{v^\afa},\quad \xi_{0,0}=e.
\end{equation}
Note that the recursion relations \eqref{d-recursion} imply that
\begin{equation}\label{d-recur-2}
  \nabla_X\xi_{i,p}=X\cdot\xi_{i,p-1}
\end{equation}
for all $(i,p)\in\mcalI$ and $X\in\Vect(M)$,
and \eqref{qhomog for theta-1}, \eqref{qhomog for theta-2} imply that
\begin{align}
  [E,\xi_{\afa,p}] &=
    \left(
      p+\mu_\afa-\frac {2-d}2
    \right)\xi_{\afa,p}
 +\sum_{s=1}^{p}
    (R_s)_\afa^\veps\xi_{\veps, p-s},
\label{qhomog for xi-1}\\
  [E,\xi_{0,q}] &= (q-1)\xi_{0,q} + \sum_{s=1}^{q}r^\veps_s\xi_{\veps, q-s}
\label{qhomog for xi-2}
\end{align}
for all $1\leq\afa\leq n$, $p\geq 0$ and $q\in\bbZ$.
The functions $\theta_{\al,p}$ that are constructed in \cite{GFM1} also satisfy the following normalization conditions:
\begin{equation}\label{normalization}
  \left\langle
    \xi_\afa(-z), \xi_\beta(z)
  \right\rangle = \eta_{\afa\beta},\quad \al,\beta=1,\dots,n,
\end{equation}
where
\[
  \xi_\afa(z):=\sum_{p\geq 0}\xi_{\afa,p}z^p,\qquad \afa=1,2,...,n.
\]

\begin{prop}
The vector fields $\xi_{i,p}$ with $(i,p)\in\mcalI$ are Legendre fields of the generalized Frobenius manifold $\gfme$.
\end{prop}

\begin{proof}
For each $X,Y\in\Vect(M)$, from \eqref{d-recur-2} it follows that
\[
  X\cdot\nabla_Y\xi_{i,p} = X\cdot Y\cdot \xi_{i,p-1}
 =Y\cdot\nabla_X\xi_{i,p},
\]
here $\xi_{\afa,-1}:=0$ for $1\leq \afa\leq n$.
Hence $B:=\xi_{i,p}$ satisfies \eqref{Legendre sym}, the proposition is proved.
\end{proof}

As a Legendre field, $\xi_{i,p}$ may not be quasi-homogeneous in general,
because the constants $(R_s)_\afa^\veps,\, r_s^\veps$ in
\eqref{qhomog for xi-1}, \eqref{qhomog for xi-2} may not vanish.

It is easy to see that the flow $\pp{t^{i,p}}$ in the Principal Hierarchy \eqref{PH}
coincides with the Legendre flow $\pp{t^{\xi_{i,p}}}$ that is defined in \eqref{Legendre flow} via the vector field $\xi_{i,p}$.
In other words, the Principal Hierarchy \eqref{PH}
is a hierarchy of Legendre flows generated by $\Bigset{\xi_{i,p}}{(i,p)\in\mcalI}$.

\begin{rmk}
 The $n\times n$ matrices
  \begin{equation}\label{eta mu Rs}
    \eta:=(\eta_{\afa\beta}),\quad
    \mu:=\diag(\mu_1,\mu_2,...,\mu_n),\quad
    R_s:= \left((R_s)_\afa^\beta\right),\quad s\geq 1
  \end{equation}
form the monodromy data \cite{2dTFT, Painleve, normal-form, GFM1}
  of the generalized Frobenius manifold $\gfme$ at $z=0$.
Let us introduce the $n\times 1$ column matrices
  \begin{equation}
    \bfr_s :=
      \begin{cases}
    (r_s^1, r_s^2,...,r_s^n)^\rmT, & s>0, \\
    \boldsymbol 0, & s\leq 0,
    \end{cases}
  \end{equation}
  and the $(n+2)\times(n+2)$ matrices
  \begin{equation}\label{n+2 eta mu R}
  \widetilde{\eta}:=
    \begin{pmatrix}
        &   & 1 \\
        & \eta &  \\
      1 &   &
    \end{pmatrix},\quad
    \widetilde\mu:=
    \begin{pmatrix}
      \mu_0 &  &  \\
       & \mu &  \\
       &  & -\mu_0
    \end{pmatrix},\quad
  \widetilde{R}_s:=
    \begin{pmatrix}
      0 & \boldsymbol 0 & 0 \\
      \bfr_s & R_s & \boldsymbol 0 \\
      a_{00;s} & \bfr_s^\dag & 0
    \end{pmatrix},
  \end{equation}
here we denote $R_s=0$ when $s\le 0$ and
\[
  \mu_0:=-\frac d2,\quad
  \bfr_s^\dag:=(-1)^{s+1}\bfr^\rmT_s\eta.
\]
Then the relations \eqref{Rk-1}--\eqref{Rk-4} can be rewritten as
\[
    \widetilde R_s^{\rmT}\widetilde\eta = (-1)^{s+1}\widetilde\eta \widetilde R_s,\quad
    [\widetilde\mu, \widetilde R_s] = s\widetilde R_s,\quad s\in\mathbb{Z},
\]
and the matrices $\widetilde\eta,\widetilde\mu, \widetilde R_s\,(s\ge 1)$ form the monodromy data of
an $(n+2)$-dimensional Frobenius manifold $\widetilde{M}\cong\bbC\times M\times\bbC$ that is associated with $M$ at $z=0$.
The explicit construction of the Frobenius manifold structure on $\widetilde{M}$ can be found in \cite{GFM1},
see also in \cite{n+2 dim}.
\end{rmk}

\subsection{The Legendre-extended Principal Hierarchy}
Given a Legendre field of a generalized Frobenius manifold, we are to  construct in this subsection a family of Legendre fields and extend the Principal Hierarchy by adding to it the associated Legendre flows.
To this end, we first present
the following lemma which is a generalized version of Proposition 2.6 given in \cite{Strachan-Stedman}.

\begin{lem}
  Suppose $B$ is a Legendre field on a generalized Frobenius manifold $\gfm$, then we have
  \begin{enumerate}
    \item $\nabla_eB$ is also a Legendre field.
    \item There locally exists a Legendre field $B_+$ satisfying the relation
    \begin{equation}\label{def B+}
      \nabla_XB_+=X\cdot B,\quad \forall X\in\Vect(M).
    \end{equation}
     \end{enumerate}
\end{lem}

\begin{proof}
We note that the $(1,3)$-tensor $\nabla c$ has the symmetry property
\[
  (\nabla c)(X_1,X_2,X_3)
  = (\nabla c)(X_{\sigma (1)},X_{\sigma (2)},X_{\sigma (3)}),\quad
  \forall X_i\in\Vect(M),\, \forall \sigma\in S_3,
\]
where $ (\nabla c)(X_1,X_2,X_3)=\nabla_{X_1}c(X_2,X_3)$. In particular,
\begin{align}
   & (\nabla_e c)(X,Y)=\,
      (\nabla_X c)(e,Y)
    =\nabla_X(e\cdot Y) - (\nabla_Xe)\cdot Y - e\cdot\nabla_XY\notag\\
  =&
    \nabla_XY-(\nabla_X e)\cdot Y - \nabla_XY
  =
    -(\nabla_Xe)\cdot Y \label{nabla e bullet}
  \end{align}
holds true for all $X,Y\in\Vect(M)$.
Thus, for a given Legendre field $B$, by using the flatness of the Levi-Civita connection
$\nabla$ of $\eta$
and by using \eqref{basic prop Y=e}, \eqref{nabla e bullet}, we obtain
\begin{align*}
&\nabla_Y\nabla_eB
=
    \nabla_e\nabla_YB + \nabla_{[Y,e]}B \\
=&
    \nabla_e(Y\cdot\nabla_eB) + [Y,e]\cdot\nabla_eB
\\=&
  \nabla_eY\cdot\nabla_eB
 +(\nabla_e c)(Y, \nabla_eB)
 +Y\cdot\nabla_e\nabla_eB
 +[Y,e]\cdot\nabla_eB  \\
=&
  \left(
    \nabla_eY-\nabla_Ye+[Y,e]
  \right)\cdot\nabla_eB
 +Y\cdot\nabla_e\nabla_eB \\
=&
  Y\cdot \nabla_e\nabla_eB,\quad \forall Y\in\Vect(M).
\end{align*}
Therefore we have
\[
  X\cdot\nabla_Y(\nabla_eB) = X\cdot Y\cdot \nabla_e\nabla_eB
 = Y\cdot\nabla_X(\nabla_eB),
\]
which implies that $\nabla_eB$ is a Legendre field.

Suppose $B_+=v_B^\al\p_\al$ is a solution to the system of equations \eqref{v_B^afa},
then we have
\[X\cdot\nabla_e B_+=X\cdot(e\cdot B)=X\cdot B=\nabla_X B_+,\]
so $B_+$ is also a Legendre field.
The lemma is proved.
\end{proof}

Taking $X=e$ in \eqref{def B+}, we obtain
\[
  \nabla_eB_+=B
\]
for the solution $B_+$ to \eqref{def B+}.
We also note that such a $B_+$ is not unique, since for an arbitrary flat vector field $X$ (i.e. $\nabla X=0$,
or locally $X=X^\afa\p_\afa$ for constants $X^\afa$), $B_++X$ is also a solution to \eqref{def B+}.

Let $B$ be a Legendre field of $\gfm$,
then from the above lemma we know
the existence of a family of Legendre fields $\{\xi_{B,q}\}_{q\in\bbZ}$ such that
\begin{align}
 & \xi_{B,0} =\, B, \label{xiB0=B}\\
 & \nabla_X\xi_{B,q+1} =\, X\cdot \xi_{B,q}  \label{B+ recursion}
\end{align}
for all $q\in\bbZ$ and $X\in\Vect(M)$. From Remark \ref{rmk:grad field}
we also know that
there exists a family of functions $\{\theta_{B,q}\}_{q\in\bbZ}$ such that
\begin{equation}\label{xiBq-grad}
  \xi_{B,q} = \grad_\eta\theta_{B,q},\quad q\in\bbZ.
\end{equation}
Then in the flat coordinates $v^1,\dots, v^n$ the relation \eqref{B+ recursion} can be rewritten in the form
\begin{equation}
  \p_\afa\p_\beta\theta_{B,q+1} = c^\gamma_{\afa\beta}\p_\gamma\theta_{B,q},\quad \al, \beta=1,\dots,n,\, q\in\mathbb{Z}.
\end{equation}

Now let us consider the quasi-homogeneous case. Suppose $B$ be a quasi-homogeneous Legendre field of the generalized Frobenius manifold $\gfme$ of charge $d$,
and $\mu_B$ be the parameter that is given by \eqref{quasi-homog b},
then the above-mentioned vector fields $\{\xi_{B,q}\}_{q\in\bbZ}$ can be specified by the following proposition.
\begin{prop}
Let $B$ be a quasi-homogeneous Legendre field of  the generalized Frobenius manifold $\gfme$, then
there exists a family of Legendre fields $\{\xi_{B,q}\}_{q\in\bbZ}$ satisfying the relations \eqref{xiB0=B}, \eqref{B+ recursion},
and the following additional equations:
\begin{equation}\label{qhomog-xiBq}
  [E,\xi_{B,q}] =
  \left(
    q+\mu_B-\frac {2-d}2
  \right)\xi_{B,q}
 +\sum_{s=1}^{q}
   \xi_{\afa, q-s}r_{B;s}^\afa,
\end{equation}
where $\Bigset{r_{B;s}^\afa}{s\geq 1,\, 1\leq \afa\leq n}$ is a set of constants satisfying the condition
\begin{equation}\label{kill rBs}
  r_{B;s}^\afa\neq 0 \quad\textrm{only if}\quad \mu_\afa-\mu_B=s.
\end{equation}
\end{prop}

\begin{proof}
  We first consider special cases.
  If $B=\pp{v^\afa}$, then $\xi_{B,q},\, r_{B;s}^\beta$
  coincide with $\xi_{\afa,q},\, (R_s)^\beta_\afa$ in \eqref{qhomog for theta-1}--\eqref{qhomog for xi-2},
  here $\xi_{\afa,q}:=0$ if $q<0$;
  in the cases when $B=e$, the corresponding $\xi_{B,q},\,r_{B;s}^\afa$
  coincide with $\xi_{0,q},\, r_s^\afa$.

  For a generic quasi-homogeneous Legendre field $B$,
  the proof of the existence of $\xi_{B,q},\, r_{B;s}^\afa$
  is similar to that of  the case $B=e$, the details of which can be found in \cite{GFM1} and we omit here. The proposition is proved.
\end{proof}

\begin{rmk}
The equation \eqref{qhomog-xiBq} implies that the functions $\theta_{B,q}$ that are defined by \eqref{xiBq-grad}
satisfy the equations
  \begin{equation}\label{qhomog-thetaBq}
    \mcalL_E\theta_{B,q} =
    \left(q+\mu_B+\frac{2-d}{2}\right)\theta_{B,q}
   +\sum_{s=1}^{q}
     \theta_{\afa,q-s}r^\afa_{B;s}
   +a_{10;q+1},
  \end{equation}
where $a_{10;q}$ are some constants.
We can adjust the functions $\theta_{B,q}$, if needed,
by adding to them certain constant terms so that the constants $a_{10;q}$ satisfy the condition
\begin{equation}\label{a01-kill}
  a_{10;q}\neq 0 \quad\text{only if}\quad \mu_B-\frac d2+q = 0.
\end{equation}
\end{rmk}

We introduce the $n\times 1$-matrices
\begin{equation}\label{bfrBs}
  \bfr_{B;s}:= (r_{B;s}^1, r_{B;s}^2,...,r_{B;s}^n)^\rmT,\quad s\ge 1.
\end{equation}

\begin{defn}
  Let $B$ be a quasi-homogeneous Legendre field of a generalized Frobenius manifold $\gfme$
  of charge $d$.
  \begin{enumerate}
  \item
  A family of vector fields $\Xi_B:=\Bigset{\xi_{\afa,p},\,\xi_{0,q},\,\xi_{B,q}}{1\leq\afa\leq n,\, p\geq 0,\, q\in\bbZ}$
  satisfying the relations \eqref{xi-initial}--\eqref{normalization} and \eqref{xiB0=B}--\eqref{qhomog-xiBq}
  is called a \emph{Legendre-extended calibration} of $(M, B)$.
\item The family of constants $\left\{\mu, \mu_0:=-\frac d2, \mu_B, R_s, \bfr_s, \bfr_{B;s}\mid s\ge 1\right\}$
that are given in \eqref{quasi-homog b},
  \eqref{qhomog for xi-1}--\eqref{qhomog for xi-2},
  \eqref{eta mu Rs},
  \eqref{qhomog-xiBq}, \eqref{bfrBs}
  satisfying the relations
  \eqref{Rk-1}--\eqref{Rk-3} and \eqref{kill rBs}
  is called the \emph{basic data} for the Legendre-extended calibration
  $\Xi_B$.
  \end{enumerate}
\end{defn}

\begin{defn}
  Let $B$ is a quasi-homogeneous Legendre field of a generalized Frobenius manifold $\gfme$
  of charge $d$,
  and $\Xi_{B}$ is a Legendre-extended calibration of $(M,B)$,
  then the \emph{Legendre-extended Principal Hierarchy} of $\gfme$
  with respect to $\Xi_B$ is given by the following family of evolutionary PDEs:
  \begin{align} \label{extended B-PH-1}
  \pfrac{v}{t^{\afa,p}} = \xi_{\afa,p}\cdot v_x,\quad
  \pfrac{v}{t^{0,q}} = \xi_{0,q}\cdot v_x,\quad
  \pfrac{v}{t^{B,q}} = \xi_{B,q}\cdot v_x
\end{align}
for $1\leq\afa\leq n$, $p\geq 0$ and $q\in\bbZ$, here we use the notations
\[v:=v^\afa\pp{v^\afa},\quad v_x:=v^\afa_x\pp{v^\afa}.\]
\end{defn}

The above hierarchy can be written in the form
  \begin{align}\label{extended B-PH-2}
    \pfrac{v^\afa}{t^{\beta,p}} =\eta^{\afa\gamma}\p_x\p_{\gamma}\theta_{\beta,p+1}, \quad
    \pfrac{v^\afa}{t^{0,q}} =\eta^{\afa\beta}\p_x\p_{\beta}\theta_{0,q+1},  \quad
    \pfrac{v^\afa}{t^{B,q}} =\eta^{\afa\beta}\p_x\p_{\beta}\theta_{B,q+1}
  \end{align}
for $1\leq\afa\leq n$, $p\geq 0$ and $q\in\bbZ$,
where the functions $\theta_{i,p}, \theta_{B,q}$ associated with $\Xi_B$
are defined as in \eqref{d-recursion}--\eqref{qhomog for theta-2}
and \eqref{xiBq-grad}--\eqref{qhomog-thetaBq}.

Let us introduce the index set $\mcalI_B$ as follows
\begin{equation}\label{index set IB}
  \mcalI_B = \Big(\{1,\dots,n\}\times\bbZ_{\geq 0}\Big)\cup\Big(\{0, B\}\times\bbZ\Big),
\end{equation}
then the Legendre-extended Principal Hierarchy \eqref{extended B-PH-2} can be rewritten as
\[
  \pfrac{v^\afa}{t^{i ,p}} =\eta^{\afa\gamma}\p_x\p_{\gamma}\theta_{i,p+1},\qquad
  1\leq\afa\leq n,\,
  (i,p)\in\mcalI_B.
\]

In the above convention, we have
\begin{equation}\label{Nabla e-xi}
  \nabla_e\xi_{i,p} = \xi_{i,p-1},\quad \text{for all}\, (i,p)\in\mcalI_B,
\end{equation}
here $\xi_{\afa,p}:=0$ if $\afa=1,2,\dots n$ and $p<0$.
Moreover, the formulae \eqref{qhomog for xi-1}--\eqref{qhomog for xi-2} and \eqref{qhomog-xiBq} can be written uniformly in the form
\begin{equation}\label{LE-xi}
  [E,\xi_{i,p}] =
  \left(
    p+\mu_i-\frac{2-d}{2}
  \right)\xi_{i,p}
 +\sum_{s=1}^{p}
  \left(\widetilde{R}_{B;s}\right)_i^\afa\xi_{\afa,p-s},
\end{equation}
where
\begin{equation}\label{240610-2125}
  \left(\widetilde{R}_{B;s}\right)_i^\afa :=
  \begin{cases}
    (R_s)_\beta^\afa & \text{if}\,\, i=\beta\in\{1,2,...,n\}, \\
    r_s^\afa & \text{if}\,\, i=0, \\
    r_{B;s}^\afa & \text{if}\,\, i=B.
  \end{cases}
\end{equation}
Here and in what follows, Greek letters $\afa,\beta,\gamma,\dots$ are always assumed to be in $\{1,2,\dots,n\}$,
and the summation $\sum_{s=1}^{p}$ vanishes when $p\leq 0$.
We remark that $\left(\widetilde{R}_{B;s}\right)_i^\afa$ are entries of a certain
matrix $\widetilde{R}_{B;s}$ of size $(n+4)\times(n+4)$ which will be defined later, see Remark \ref{rmk:240610-notation}.

\begin{rmk}
As for a Frobenius manifold $M$ with flat-unity \cite{2dTFT,normal-form},
there is a bihamiltonian structure on the jet space $J^\infty(M)$
given by the compatible Hamiltonian operators $\mathcal{P}_a=(\mathcal{P}_a^{\al\beta})$, $a=1,2$ with
\begin{equation}\label{Bihamiltonian structure}
\mathcal{P}_1^{\al\beta}=\eta^{\afa\beta}\p_x,\quad
  \mathcal{P}_2^{\al\beta} = g^{\afa\beta}\p_x+\Gamma^{\afa\beta}_\gamma v^\gamma_x,
\end{equation}
where the contravariant metric $g^{\afa\beta}=E^\veps c_\veps^{\afa\beta}$ is called the intersection form of $M$,
and $\Gamma^{\afa\beta}_\gamma=\left(\frac12-\mu_\beta\right)c^{\afa\beta}_\gamma$
are the contravariant coefficients of the Levi-Civita connection of $(g^{\afa\beta})$,
see details in \cite{2dTFT,normal-form} or \cite{GFM1}.
For an invertible quasi-homogeneous Legendre field $B$ on $M$,
the Legendre-extended Principal Hierarchy \eqref{extended B-PH-2}
is a bihamiltonian hierarchy with respect to this bihamiltonian structure,
i.e., its flows are Hamiltonian systems of the form
\[\frac{\p v^\al}{\p t^{j,q}}=\mathcal{P}_1^{\al\gamma}\frac{\delta H_{j,q}}{\delta v^\gamma},\quad \al=1,\dots,n,\, (j,q)\in\mcalI_B,\]
and they satisfy the bihamiltonian recursion relations
\begin{align*}
\mathcal{P}_2^{\al\gamma}\frac{\delta H_{i,q-1}}{\delta v^\gamma}&=
\left(q+\frac12+\mu_i\right)\mathcal{P}_1^{\al\gamma}\frac{\delta H_{i,q}}{\delta v^\gamma}
+\sum_{s=1}^q (\widetilde R_{B;s})^\veps_i \mathcal{P}_1^{\al\gamma}\frac{\delta H_{\veps,q-s}}{\delta v^\gamma}
\end{align*}
for $\al=1,\dots, n$ and $(i,q)\in\mcalI_B$,
where the Hamiltonians $H_{i,q}$ for all $(i,q)\in\mcalI_B$ are given by
\[H_{i,q}=\int \theta_{i,q+1}(v(x)) \td x.\]
\end{rmk}

\subsection{Linear reciprocal transformations}
Let $B$ be an invertible quasi-homogeneous Legendre field of a generalized Frobenius manifold
$\gfme$ of charge $d$, which will also be denoted by $M$ in what follows.
From Propositions \ref{zh-1}, \ref{prop:linearity of new E}
we know that the Legendre transformation that is induced by
$B$ transforms $M$ to a new generalized Frobenius manifold $\gfmhe$
of charge $\hat d = -2\mu_B$ which will also be denoted by $\hat M$,
and
\[\hat B:=B^{-1}\]
is an invertible quasi-homogeneous Legendre field of $\hat M$ that transforms
$\hat M$ back to $M$.
We emphasize that $M$ and $\hat M$ share the same multiplication $c$, unity $e$ and Euler field $E$.

Now we study the relationship between the
Legendre-extended calibrations of $(M, B)$ and $(\hat M, \hat B)$.
We first fix systems of flat coordinates $\{v^\afa\}$ and $\{\hat v^\afa\}$ of the metrics $\eta$ and $\hat\eta$ respectively,
which are related via \eqref{hat v-afa1}--\eqref{hat v-afa2}.

\begin{thm}
Let $\Xi_B=\Bigset{\xi_{i,p}}{(i,p)\in\mcalI_B}$
be a Legendre-extended calibration of $(M, B)$
with basic data
$\left\{\mu,\mu_0=-\frac d2, \mu_B, R_s,\bfr_s,\bfr_{B;s}\mid s\ge 1\right\}$, then
\[
  \hat\Xi_{\hat B}=\Bigset{\hat\xi_{i,p}}
  {(i,p)\in\mcalI_{\hat B}}
\]
forms a Legendre-extended calibration of $(\hat M, \hat B)$, where
\begin{equation} \label{rel1}
  \hat\xi_{\afa,p} = \hat B\cdot\xi_{\afa,p}, \quad
  \hat\xi_{0,q} = \hat B\cdot\xi_{B,q}, \quad
  \hat\xi_{\hat B,q} = \hat B\cdot\xi_{0,q}.
\end{equation}
Moreover, the basic data $\left\{\hat\mu, \hat\mu_0=-\frac{\hat d}{2}, \hat\mu_B,
\hat R_s,\hat \bfr_s, \hat\bfr_{\hat B;s}\mid s\ge 1\right\}$ for $\hat\Xi_{\hat B}$ satisfy
the relations
\begin{align}\label{rel4}
\hat\mu &= \mu,\quad \hat\mu_0 = \mu_B,\quad
  \hat\mu_{\hat B} = \mu_0,
  \\
\hat R_s&=R_s, \quad
\hat\bfr_s = \bfr_{B_s},\quad \hat\bfr_{\hat B;s}=\bfr_s.\label{rel5}
\end{align}
\end{thm}

\begin{proof}
Firstly, the initial conditions \eqref{xi-initial}, \eqref{xiB0=B} for $\hat\Xi_{\hat B}$ hold true, because
\begin{align*}
  \hat\xi_{\afa,0} &=\, \hat B\cdot\xi_{\afa,0} = B^{-1}\cdot\pp{v^\afa} = \pp{\hat v^\afa}, \\
  \hat\xi_{0,0} &=\, \hat B\cdot\xi_{B,0} = B^{-1}\cdot B = e, \\
  \hat\xi_{\hat B,0} &=\, \hat B\cdot\xi_{0,0} = \hat B\cdot e = \hat B.
\end{align*}
Denote by $\hat\nabla$ the Levi-Civita connection for the new metric $\hat\eta$,
from \eqref{new conenction hat nabla} we obtain
\begin{align*}
  \hat\nabla_X\hat\xi_{\afa,p+1}
&=\,
  \hat B\cdot\nabla_X
  \left(
    B\cdot\hat B\cdot\xi_{\afa,p+1}
  \right)
=
  \hat B\cdot\nabla_X\xi_{\afa,p+1}
=
  \hat B\cdot (X\cdot\xi_{\afa,p}) \\
&=\,
  X\cdot\left(\hat B\cdot \xi_{\afa,p}\right)
=
  X\cdot\hat\xi_{\afa,p}
\end{align*}
for $X\in\Vect(M)$,
therefore the recursion relations \eqref{d-recur-2} for $\hat\xi_{\afa,p}$ hold true.
Similarly, we also have
\[
  \hat\nabla_X\hat\xi_{0,q+1} = X\cdot\hat\xi_{0,q},\quad
  \hat\nabla_X\hat\xi_{B,q+1} = X\cdot\hat\xi_{B,q}
\]
for $q\in\bbZ$ and $X\in\Vect(M)$.

The normalization condition \eqref{normalization} for $\hat\xi_{\al,p}$ follows from the Definition \ref{pair-B} of the metric $\hat\eta$ and the relation \eqref{hat v-afa2}. Indeed, we have
\begin{align*}
&\,
  \pair{\hat\xi_\afa(-z)}{\hat\xi_\beta(z)}_B
=
  \pair{B\cdot\hat\xi_{\afa}(-z)}{B\cdot\hat\xi_{\beta}(z)} \\
=&\,
   \pair{B\cdot\hat B\cdot\xi_\afa(-z)}
   {B\cdot\hat B\cdot\xi_\beta(z)}
=
  \pair{\xi_\afa(-z)}{\xi_\beta(z)} = \eta_{\afa\beta}=\hat\eta_{\afa\beta}.
\end{align*}

Now, from \eqref{quasi-homog of c and eta} and \eqref{b^-1 quasi-homog} we arrive at
\begin{align}
  [E,\hat\xi_{\afa,p}]
&=
  \mcalL_E\left(\hat B\cdot\xi_{\afa,p}\right) \notag\\
&=
  [E,\hat B]\cdot\xi_{\afa,p}
 +(\mcalL_E c)(\hat B, \xi_{\afa,p})
 +\hat B\cdot [E,\xi_{\afa,p}]\notag \\
&=
  -\left(
    \mu_B+\frac{2+d}{2}
  \right)\hat B\cdot \xi_{\afa,p}
 +\hat B\cdot \xi_{\afa,p} \notag\\
&\quad
  +\hat B\cdot
   \left(
     \left(
       p+\mu_\afa-\frac{2-d}{2}
     \right)\xi_{\afa,p}
    +\sum_{s=1}^{p}
       (R_s)_\afa^\gamma \xi_{\gamma,p-s}
   \right) \notag\\
&=
  \left(
    p+\mu_\afa-\frac{2-\hat d}{2}
  \right)\hat\xi_{\afa,p}
 +\sum_{s=1}^{p}
  (R_s)_\afa^\gamma\hat\xi_{\gamma,p-s}.\label{zh-2}
\end{align}
In a similar way, we can verify that
\begin{align}
  [E,\hat\xi_{0,q}] &=\, (q-1)\hat\xi_{0,q}
+
  \sum_{s=1}^{q}r_{B;s}^\afa\hat\xi_{\afa,q-s}, \label{zh-3}\\
  [E,\hat\xi_{B,q}]
&=
  \left(
    q+\hat\mu_{\hat B}-\frac{2-\hat d}{2}
  \right)\hat\xi_{B,q}
 +\sum_{s=1}^{q}r^\afa_s\hat\xi_{\afa,q-s},\label{zh-4}
\end{align}
therefore $\hat\Xi_{\hat B}$ is a Legendre-extended calibration of $(\hat M,\hat B)$.
The relations \eqref{rel4}--\eqref{rel5} for the basic data of $\hat\Xi_{\hat B}$ follow from
Proposition \ref{prop:linearity of new E}, Remark \ref{rmk:hat d} and
\eqref{zh-2}--\eqref{zh-4}.
The theorem is proved.
\end{proof}

\begin{rmk}
The relations \eqref{rel1} can be rewritten as
\begin{equation}\label{240701-1433}
  \pfrac{\hat\theta_{\afa,p}}{\hat v^\beta} = \pfrac{\theta_{\afa,p}}{v^\beta},\quad
  \pfrac{\hat\theta_{0,q}}{\hat v^\afa} = \pfrac{\theta_{B,q}}{v^\afa},\quad
  \pfrac{\hat\theta_{\hat B,q}}{\hat v^\afa} = \pfrac{\theta_{0,q}}{v^\afa}
\end{equation}
via the densities $\theta_{\afa,p}, \theta_{0,q}$ and $\theta_{B,q}$
(and their $\hat M$-analogues)
defined as in \eqref{d-recursion}--\eqref{qhomog for theta-2}.
\end{rmk}

\begin{thm}\label{thm:linear reciprocal}
Let $\Xi_B$ and $\hat\Xi_{\hat B}$ be the Legendre-extended calibrations of $(M, B)$ and $(\hat M, \hat B)$.
Then the Legendre-extended Principal Hierarchies \eqref{extended B-PH-1} of $\Xi_B$ and $\hat\Xi_{\hat B}$
are related by the linear reciprocal transformation
\begin{equation}\label{reciprocal}
  \hat t^{\afa,p} = t^{\afa,p},\qquad
  \hat t^{0,q} = t^{B,q},\qquad
  \hat t^{\hat B,q} = t^{0,q}
\end{equation}
for $1\leq\afa\leq n$, $p\geq 0$, and $q\in\bbZ$.
\end{thm}
\begin{proof}
Suppose $v=v^\afa\pp{v^\afa}$ satisfies the Legendre-extended Principal Hierarchy \eqref{extended B-PH-1} of $\Xi_B$,
we need to verify that $\hat v:=\hat v^\afa\pp{\hat v^\afa}$
satisfies the corresponding Legendre-extended Principal Hierarchy of $\hat\Xi_{\hat B}$
under the identification \eqref{reciprocal} of time variables.

Note that the new spatial variable $\hat x$ is identified with the new time variable $\hat t^{0,0}$, i.e.
\[
  \hat x = \hat t^{0,0},
\]
which is analogous to \eqref{t00=x}.
For a vector field $X\in\Vect(M)$, we denote $X^\afa:=\pair{\td v^\afa}{X}$. Then from \eqref{hat v-afa1} we obtain
\begin{align*}
  \pfrac{\hat v^\afa}{\hat x}
&:=
  \pfrac{\hat v^\afa}{\hat t^{0,0}}
=
  \pfrac{\hat v^\afa}{ t^{B,0}}
=
  \pfrac{\hat v^\afa}{v^\beta}\pfrac{v^\beta}{t^{B,0}}
=
  B^\veps c^\afa_{\veps\beta}
  B^\lmd c^\beta_{\lmd\gamma}v^\gamma_x
=
  (B^2\cdot v_x)^\afa.
\end{align*}
It follows from \eqref{hat v-afa2} that
\begin{align}
  \hat v_{\hat x}&:= \pfrac{\hat v^\afa}{\hat x}\pp{\hat v^\afa}
=
  (B^2\cdot v_x)^\afa
  \left(
    \hat B\cdot\pp{v^\afa}
  \right) \notag\\
&\phantom:=
  \hat B\cdot \left(
    (B^2\cdot v_x)^\afa\pp{v^\afa}
  \right)
=
  \hat B\cdot (B^2\cdot v_x) = B\cdot v_x.  \label{240701-1622}
\end{align}
Therefore, by \eqref{rel1} and the above equations we arrive at
\begin{align*}
\pfrac{\hat v}{\hat t^{\afa,p}}
&=
  \pfrac{\hat v^\gamma}{\hat t^{\afa,p}}\pp{\hat v^\gamma}
=
  \pfrac{\hat v^\gamma}{t^{\afa,p}}\pp{\hat v^\gamma}
=
  \pfrac{\hat v^\gamma}{v^\beta}
  \pfrac{v^\beta}{t^{\afa,p}}
  \pp{\hat v^\gamma} \\
&=
  B^\veps c^\gamma_{\veps\beta}
  (\xi_{\afa,p}\cdot v_x)^\beta
  \left(
    \hat B\cdot\pp{v^\gamma}
  \right)
=
  \hat B\cdot
  \left(
    B\cdot\xi_{\afa,p}\cdot v_x
  \right)^\gamma \pp{v^\gamma} \\
&=
  \xi_{\afa,p}\cdot v_x = \xi_{\afa,p}\cdot \hat B\cdot\hat v_{\hat x}
=
  \hat\xi_{\afa,p}\cdot\hat v_{\hat x}.
\end{align*}
In a similar way, we can verify that
\[
  \pfrac{\hat v}{\hat t^{0,q}}
=
  \hat\xi_{0,q}\cdot\hat v_{\hat x},
\quad
  \pfrac{\hat v}{\hat t^{\hat B,q}}
=
  \hat\xi_{\hat B,q}\cdot\hat v_{\hat x},\quad q\in\mathbb{Z}.
\]
The theorem is proved.
\end{proof}

\subsection{Tau structures of the Legendre-extended Principal Hierarchy}

In this section, we introduce the tau structure and tau-cover of the Legendre-extended Principal Hierarchy
\eqref{extended B-PH-1} of $(M, B)$.

Suppose $B_1, B_2$ be two arbitrary Legendre fields of $M$,
then by using \eqref{LegendreSym1} we have the identities
  \[
    \p^\delta\left(B_1^\afa c^\gamma_{\afa\beta} B_2^\beta\right)
  = \p^\gamma\left(B_1^\afa c^\delta_{\afa\beta} B_2^\beta\right),
  \]
which imply that $B_1\cdot B_2$ is a gradient field, i.e.
there locally exists a function $\Omg_{B_1;B_2}$ on $M$ such that
\[B_1\cdot B_2 = \grad_\eta\Omg_{B_1;B_2}.\]
In particular, for a Legendre-extended calibration
\[\Xi_B:=\Bigset{\xi_{i,p}}{(i,p)\in\mcalI_B}\]
of $(M,B)$, since all vector fields in $\Xi_B$ are Legendre,
there locally exists a family of 2-point functions
$\{\Omg_{i,p;j,q}\}_{(i,p), (j,q)\in\mcalI_B}$ on $M$ such that
\begin{equation}\label{dOmg-1}
  \xi_{i,p}\cdot\xi_{j,q} = \grad_\eta\Omg_{i,p;j,q}.
\end{equation}
From the definition of the Legendre-extended Principal Hierarchy we know that the above identities are equivalent to the following
\emph{tau-symmetry} condition \cite{normal-form}:
\begin{equation}\label{dOmg-2}
  \pfrac{\theta_{i,p}}{t^{j,q}} = \p_x\Omg_{i,p;j,q} = \pfrac{\theta_{j,q}}{t^{i,p}}.
\end{equation}

For any vector field $X\in\Vect(M)$, from \eqref{quasi-homog of c and eta}, \eqref{qhomog for xi-1}, \eqref{qhomog for xi-2}, \eqref{qhomog-xiBq} and \eqref{dOmg-1} we obtain
\begin{align*}
 &
   \left(\td\,(\mcalL_E\Omg_{i,p;j,q})\right)(X)
=
  \left(\mcalL_E(\td\Omg_{i,p;j,q})\right)(X)\\
&\quad=
  \mcalL_E ((\td\Omg_{i,p;j,q})X) - (\td\Omg_{i,p;j,q})[E,X]\\
&\quad=
  \mcalL_E\pair{\grad_\eta\Omg_{i,p;j,q}}{X}
 -\pair{\grad_\eta\Omg_{i,p;j,q}}{[E,X]}\\
&\quad=
  \pair{\mcalL_E(\xi_{i,p}\cdot\xi_{j,q})}{X}
 +(2-d)\pair{\xi_{i,p}\cdot\xi_{j,q}}{X}
\\&\quad=
  (p+q+\mu_i+\mu_j+1)
  \pair{\xi_{i,p}\cdot\xi_{j,q}}{X} \\
&\qquad
 +\sum_{s=1}^{p}
  \left(\widetilde{R}_{B;s}\right)_i^\afa
  \pair{\xi_{\afa,p-s}\cdot\xi_{j,q}}{X}
 +\sum_{s=1}^{q}
  \left(\widetilde{R}_{B;s}\right)_j^\afa
  \pair{\xi_{i,p}\cdot\xi_{\afa,q-s}}{X}
\\&\quad=
  \td\left(
    (p+q+\mu_i+\mu_j+1)\Omg_{i,p;j,q}
   +\sum_{s=1}^{p}
  \left(\widetilde{R}_{B;s}\right)_i^\afa\Omg_{\afa,p-s;j,q}
  \right.\\
  &\qquad\left.+\sum_{s=1}^{q}
  \left(\widetilde{R}_{B;s}\right)_j^\afa\Omg_{i,p;\afa,q-s}
  \right)(X),
\end{align*}
which imply that there exists a family of constants $\{C_{i,p;j,q}\}_{(i,p),(j,q)\in\mcalI_B}$ such that
\begin{equation}\label{LE Omg C}
\begin{split}
  &\mcalL_E\Omg_{i,p;j,q}
=
  (p+q+\mu_i+\mu_j+1)
  \Omg_{i,p;j,q} \\
&\quad
 +\sum_{s=1}^{p}
  \left(\widetilde{R}_{B;s}\right)_i^\afa\Omg_{\afa,p-s;j,q}
  +\sum_{s=1}^{q}
  \left(\widetilde{R}_{B;s}\right)_j^\afa\Omg_{i,p;\afa,q-s}
 +C_{i,p;j,q}.
\end{split}
\end{equation}
Note that the 2-point functions $\Omg_{i,p;j,q}$ satisfying \eqref{dOmg-1} are not unique,
they have the freedom of adding constants.
Now we specify the 2-point functions $\Omg_{i,p;j,q}$
to regularize the constants $C_{i,p;j,q}$ in the right hand side of \eqref{LE Omg C} by the following steps:

\begin{description}
  \item [\textbf{Step1.}] For $q\geq 0$ and $i\in\{0,B; 1,2,...,n\}$, define
    \begin{align}\label{240609-2127-1}
      \Omg_{i,p;\beta,q}&:=
        \sum_{k=0}^{q} (-1)^k\pair{\xi_{i,p+k+1}}{\xi_{\beta,q-k}},\\
     \Omg_{\beta,q;0,p}&:=\Omg_{0,p;\beta,q},\quad
      \Omg_{\beta,q;B,p}:=\Omg_{B,p;\beta,q}.
    \end{align}
  \item [\textbf{Step2.}] Fix a family of functions
  $\{\theta_{B,p}\}_{p\in\bbZ}$ satisfying \eqref{qhomog-thetaBq} and \eqref{a01-kill}, then define
    \begin{align}
      \Omg_{0,p;B,q}&:=
        \begin{cases}
          \sum\limits_{k=0}^{p-1}
            (-1)^k
            \pair{\xi_{0,p-k}}{\xi_{B,q+1+k}}
         +(-1)^p\theta_{B,p+q}, &  p\geq 0,
        \\
          \sum\limits_{k=0}^{-p-1}
            (-1)^k
            \pair{\xi_{0,p+1+k}}{\xi_{B,q-k}}
           +(-1)^p\theta_{B,p+q},  & p<0,
        \end{cases} \\
      \Omg_{B,q;0,p}&:=\Omg_{0,p;B,q}.
    \end{align}
Moreover, let $a_{10;p}$ be the constants that are given by \eqref{qhomog-thetaBq} and \eqref{a01-kill}, then we denote
    \begin{equation}\label{a01}
      a_{01;p}:=(-1)^{p+1}a_{10;p},\quad p\in\bbZ.
    \end{equation}

  \item [\textbf{Step3.}] Fix a family of functions $\{\theta_{0,p}\}_{p\in\bbZ}$
  satisfying \eqref{qhomog for theta-2} and \eqref{Rk-4}.
  Note that when $p$ is an odd integer, such $\theta_{0,p}$ is chosen as
  \[
    \theta_{0,p}=
    \begin{cases}
      \frac12\sum\limits_{k=0}^{p-1}(-1)^k
        \pair{\xi_{0,p-k}}{\xi_{0,k+1}}, &\text{if $p>0$ is odd}, \\
      \frac12\sum\limits_{k=0}^{-p-1}(-1)^k
        \pair{\xi_{0,p+1+k}}{\xi_{0,-k}}, &\text{if $p<0$ is odd},
    \end{cases}
  \]
  (see details in \cite{GFM1}).
  Then define
  \begin{equation}
    \Omg_{0,p;0,q}:=
      \begin{cases}
        \sum\limits_{k=0}^{p-1}
          (-1)^k\pair{\xi_{0,p-k}}{\xi_{0,q+1+k}}+(-1)^p\theta_{0,p+q}, &\text{if}\,\, p\geq 0, \\
        \sum\limits_{k=0}^{-p-1}
          (-1)^k
          \pair{\xi_{0,p+1+k}}{\xi_{0,q-k}} + (-1)^p\theta_{0,p+q}, &\text{if}\,\, p<0.
      \end{cases}
  \end{equation}
  \item [\textbf{Step4.}] In the case of $i=j=B$, \eqref{LE Omg C} imply that
  \begin{equation}\label{a11}
    \mcalL_E\Omg_{B,0;B,p} = (p+2\mu_B+1)\Omg_{B,p;B,0}
    +\sum_{s=1}^{p}
      r_{B;s}^\afa\Omg_{\afa,p-s;B,0} +a_{11;p+1}
  \end{equation}
  for some constants $\{a_{11;p}\}_{p\in\bbZ}$.
  We can appropriately choose $\{\Omg_{B,0;B,p}\}_{p\in\bbZ}$ such that
  \begin{equation}\label{240609-2028}
    a_{11;p}\neq 0 \quad\text{only if}\quad p+2\mu_B=0 \,\,\text{and $p$ is an odd integer}.
  \end{equation}
We do this in a way that is similar to the construction of $\{\theta_{0,p}\}_{p\in\bbZ}$ given in Step 3, i.e., we define
  \begin{equation*}
    \Omg_{B,0;B,p}:=
      \begin{cases}
        \frac12\sum\limits_{k=0}^{p-1}
          (-1)^p
          \pair{\xi_{B,p-k}}{\xi_{B,k+1}}, & \text{if $p>0$ is odd},\\
        \frac12\sum\limits_{k=0}^{-p-1}
          (-1)^p
          \pair{\xi_{B,p+1+k}}{\xi_{B,-k}}, & \text{if $p<0$ is odd},
      \end{cases}
  \end{equation*}
  and when $p$ is even,
  we adjust $\Omg_{B,0;B,p}$ by adding a constant to ensure that \eqref{240609-2028} holds true.
  Fix such a family of $\{\Omg_{B,0;B,p}\}_{p\in\bbZ}$, we define
  \begin{equation}\label{240609-2127-2}
    \Omg_{B,p;B,q}:=
    \begin{cases}
      \sum\limits_{k=0}^{p-1}(-1)^k
        \pair{\xi_{B,p-k}}{\xi_{B,q+1+k}}
       +(-1)^p\Omg_{B,0;B,p+q}, & \text{if $p\geq 0$},
    \\
        \sum\limits_{k=0}^{-p-1}(-1)^k
        \pair{\xi_{B,p+1+k}}{\xi_{B,q-k}}
       +(-1)^p\Omg_{B,0;B,p+q}, & \text{if $p < 0$},
    \end{cases}
  \end{equation}
  for all $p,q\in\bbZ$.
\end{description}

We can verify via a straightforward calculation
that the functions $\Omg_{i,p;j,q}$ defined in the above-mentioned four steps
satisfy the relations \eqref{dOmg-1}--\eqref{dOmg-2}, and they also have the following properties
for all $(i,p), (j,q)\in\mcalI_B$:
\begin{equation}\label{240701-1058}
  \Omg_{i,p;j,q} = \Omg_{j,q;i,p}, \qquad
  \Omg_{0,0;i,p} = \theta_{i,p}, \qquad
  v^\afa = \eta^{\afa\beta}\Omg_{0,0;\beta,0}.
\end{equation}
We call this family of functions $\{\Omg_{i,p;j,q}\}_{(i,p),(j,q)\in\mcalI_B}$
a \textit{tau structure} of the Legendre-extend Principal Hierarchy \eqref{extended B-PH-2}.
For each $s\in\bbZ$, we introduce the $2\times 2$ matrix
\begin{equation}\label{def of A_s}
  A_s:=\begin{pmatrix}
         a_{00;s} & a_{01;s} \\
         a_{10;s} & a_{11;s}
       \end{pmatrix},
\end{equation}
where the entries of the above matrices are defined as in
\eqref{qhomog for theta-2}, \eqref{Rk-4}, \eqref{qhomog-thetaBq}, \eqref{a01-kill},
\eqref{a01}, \eqref{a11} and \eqref{240609-2028}, and we call the matrix
\begin{equation}
  A:=\sum_{s\in\bbZ}A_s
\end{equation}
the \textit{extra data} of the tau structure $\{\Omg_{i,p;j,q}\}$. From the definition of the constants $a_{ij;s}$ we know that
the above sum have only finite number of non-zero terms.

\begin{defn}\label{def:complete data}
  Let $\{\Omg_{i,p;j,q}\}$ be a tau structure for the hierarchy \eqref{extended B-PH-2},
  and $\{\mu, \mu_0:=-\frac d2, \mu_B, R, \bfr, \bfr_B\}$ be the basic data.
We introduce the following blocked $(n+4)\times(n+4)$ matrices for $s\in\bbZ$:
  \begin{align}\label{complete data-1}
    &\widetilde{\mu}_B
  :=
    \begin{pmatrix}
      \nu &   &   \\
        & \mu &   \\
        &   & -\nu
    \end{pmatrix},\quad
  \widetilde{R}_{B;s}:=
    \begin{pmatrix}
      0 & 0 & 0 \\
      T_s & R_s & 0 \\
      A_s & T_s^\dag & 0
    \end{pmatrix},
\quad    \widetilde{R}_B:=\sum_{s\in\bbZ}\widetilde{R}_{B;s},
  \end{align}
where the $2\times 2$ blocks $\{A_s\}_{s\in\bbZ}$ are defined in \eqref{def of A_s},
\begin{align}
 & 
  \nu:=\begin{pmatrix}
       \mu_0 &  \\
         & \mu_B
     \end{pmatrix},\quad
  T_s:=\begin{cases}
         (\bfr_s, \bfr_{B;s})_{n\times 2} & \text{if $s>0$}, \\
         \boldsymbol 0 _{n\times 2} & \text{if $s\leq 0$},
       \end{cases}\\
&
  T_s^\dag := (-1)^{s+1}\left(T_s\right)^{\rmT}\eta,
\end{align}
and the $n\times n$ blocks $R_s:=0$ when $s\leq 0$.
We call $\{\widetilde{\mu}_B, \widetilde{R}_B\}$ the complete data
of the tau structure $\{\Omg_{i,p;j,q}\}$.
\end{defn}

Note that all the entries of complete data $\{\widetilde{\mu}_B, \widetilde{R}_B\}$
are determined by the basis data  $\{\mu, \mu_0:=-\frac d2, \mu_B, R, \bfr, \bfr_B\}$ and the extra data $A$,
and the identities \eqref{Rk-1}--\eqref{Rk-4}, \eqref{kill rBs}, \eqref{a01-kill}
and \eqref{240609-2028} can be uniformly written as
\begin{align}
 &\widetilde{\eta}_B\widetilde{\mu}_B
 +\widetilde{\mu}_B\widetilde{\eta}_B=0,
\\&
  (\widetilde{R}_{B;s})^\rmT\widetilde{\eta}_B
 =(-1)^{s+1}\widetilde{\eta}_B\widetilde{R}_{B;s},
\\&
  [\widetilde{\mu}_B, \widetilde{R}_{B;s}] = s\widetilde{R}_{B;s}
\end{align}
for all $s\in\bbZ$, where the matrix $\tilde{\eta}_B$ is defined by
\[\widetilde{\eta}_B=
    \begin{pmatrix}
        &   & I_2 \\
        & \eta &  \\
      I_2 &  &
    \end{pmatrix},\qquad 
I_2=\begin{pmatrix}
         1 &   \\
           & 1
       \end{pmatrix}. 
\]

\begin{rmk}\label{rmk:240610-notation}
Let us index the rows and columns of the above $(n+4)\times(n+4)$ matrices by
  \[
    (0,B; 1,2,...,n ; 0', B').
  \]
Then the left hand side of \eqref{240610-2125} is indeed the $(\afa,i)$-entry of $\widetilde{R}_{B;s}$
  for $1\leq\afa\leq n$ and $i\in\{0,B; 1,2,...,n\}$.
    We also note that the matrices
    $\widetilde{\eta}, \widetilde{\mu}, \widetilde{R}_s$ defined in \eqref{n+2 eta mu R}
    are submatrices of $\widetilde{\eta}_B, \widetilde{\mu}_B$, $\widetilde{R}_{B;s}$
    obtained by deleting the $B$-th and $B'$-th rows and columns.
\end{rmk}

\begin{prop}The tau structure $\{\Omg_{i,p;j,q}\}$ defined in
\eqref{240609-2127-1}--\eqref{240609-2127-2} satisfy the relations
\begin{align}
  \mcalL_e\Omg_{i,p;j,q} =&\,
  \Omg_{i,p-1;j,q} + \Omg_{i,p;j,q-1} + \delta_{p0}\delta_{q0}(\widetilde{\eta}_B)_{ij}, \label{p0Omg}\\
 \mcalL_E\Omg_{i,p;j,q} =&\,
    \left(
      p+q+\mu_i+\mu_j+1
    \right)\Omg_{i,p;j,q}
    + \sum_{s=1}^{p}
       \left(\widetilde{R}_{B;s}\right)_i^\afa\Omg_{\afa,p-s;j,q}\notag \\
  &\,  + \sum_{s=1}^{q}
       \left(\widetilde{R}_{B;s}\right)_j^\afa\Omg_{i,p;\afa,q-s}
    + (-1)^p \left(\widetilde{\eta}_B\widetilde{R}_{B;p+q+1}\right)_{ij}  \label{p1Omg}
\end{align}
for all $(i,p), (j,q)\in\mcalI_B$.
\end{prop}

\begin{proof}
The proposition can be proved via a straightforward calculation
by using the definitions \eqref{240609-2127-1}--\eqref{240609-2127-2},
and the relations \eqref{quasi-homog of c and eta}, \eqref{def of theta afa p},
\eqref{Nabla e-xi}--\eqref{LE-xi}. We omit the details here.
\end{proof}

 We can generalize the concept of \textit{tau-cover} of the Principal Hierarchy \eqref{PH}
(see details in \cite{GFM1}, or \cite{normal-form})  to the Legendre-extended Principal Hierarchy.

\begin{defn}
Let $\{\Omg_{i,p;j,q}\}$ be a tau structure for the Legendre-extend Principal Hierarchy \eqref{extended B-PH-2},
then the following system of PDEs
\begin{equation}\label{tau cover}
  \pfrac{f}{t^{i,p}}=f_{i,p},\quad
  \pfrac{f_{i,p}}{t^{j,q}}=\Omg_{i,p;j,q},\quad
  \pfrac{v^\afa}{t^{i,p}} = \eta^{\afa\beta}\p_x\Omg_{\beta,0;i,p}
\end{equation}
for unknown functions $\{f, f_{i,p}, v^\afa\}$ is called the Legendre-extended tau-cover
of the Legendre-extend Principal Hierarchy.
Here the double-indexes $(i,p), (j,q)\in\mcalI_B$.
\end{defn}

We note that \textit{tau-cover} of the Principal Hierarchy \eqref{PH}
are obtained from \eqref{tau cover} by
restricting the indexes $i,j$ to the subset $\{0;1,2...,n\}$. We also note that from \eqref{240701-1058} it follows that
\begin{equation}\label{zh-09}
v^\al=\eta^{\al\beta}\frac{\p^2 f}{\p x\p t^{\beta,0}},\quad \al=1,\dots,n.
\end{equation}

\begin{thm}\label{thm:tau correspendence}
Let $B$ be an invertible quasi-homogeneous Legendre field on $M$,
$\Xi_B$ and $\hat\Xi_{\hat B}$ be families of Legendre-extended calibrations of $(M,B)$ and
$(\hat M, \hat B)$ related by \eqref{rel1}, and
$\{\Omg_{i,p;j,q}\}$ be a tau structure of the Legendre-extended Principal Hierarchy with respect to $\Xi_B$
with complete data $\{\widetilde{\mu}_B, \widetilde{R}_B\}$, then the following statements hold true:
\begin{enumerate}
  \item The functions
  \begin{equation} \label{240701-1059}
    \hat\Omg_{i,p;j,q} := \Omg_{\sgm(i),p; \sgm(j),q}
  \end{equation}
  form a tau structure of the Legendre-extended Principal Hierarchy with respect to $\hat\Xi_{\hat B}$,
  where the bijection
  \[\sgm\colon\{0,\hat B;1,2,...,n\}\to\{0, B;1,2,...,n\}\] is defined by
  \begin{equation}\label{involution sgm}
    \sgm(\afa) = \afa,\qquad \sgm(0)= B,\qquad \sgm(\hat B)=0
  \end{equation}
  for $\afa\in\{1,2,...,n\}$.
  \item The complete data $\{\hat{\widetilde{\mu}}_{\hat B}, \hat{\widetilde{R}}_{\hat B}\}$, $\{\widetilde{\mu}_B, \widetilde{R}_B\}$
  of the tau structures $\{\hat\Omg_{i,p;j,q}\}$ and $\{\Omg_{i,p;j,q}\}$ respectively
  are related by the formulae \eqref{rel4}--\eqref{rel5} and
  \begin{equation}
    \begin{pmatrix}
      \hat a_{00;s} & \hat a_{01;s} \\
      \hat a_{10;s} & \hat a_{11;s}
    \end{pmatrix}
  =
    \begin{pmatrix}
       a_{11;s} &  a_{10;s} \\
       a_{01;s} &  a_{00;s}
    \end{pmatrix}.
  \end{equation}

\item Suppose $f$, $f_{i,p}$ and $v^\afa$
satisfy the Legendre-extended tau-cover \eqref{tau cover},
then under the time variable identification \eqref{reciprocal}, the functions
\begin{equation}\label{zh-08}
  \hat f := f,\quad \hat f_{i,p}:=f_{\sgm(i),p},\quad
  \hat v^\afa :=
  \hat{\eta}^{\afa\beta}
  \frac{\p^2\hat f}{\p\hat x\p\hat t^{\beta,0}}
\end{equation}
satisfy the Legendre-extended tau-cover of the Legendre-extended Principal Hierarchy with respect to $\hat\Xi_{\hat B}$
and the tau structure $\{\hat\Omg_{i,p;j,q}\}$. We note that $\hat\eta^{\al\beta}=\eta^{\al\beta}$ due to \eqref{hat v-afa2}.
\end{enumerate}
\end{thm}

\begin{proof}
  By equation \eqref{pair-B} for the new metric $\hat\eta$, it is clear that
  \[
    \grad_{\hat\eta}\fai = \hat B^2\cdot\grad_\eta \fai
  \]
  for $\fai\in C^\infty(M)$,
  therefore from \eqref{rel1} and \eqref{dOmg-1} it follows that
  \begin{align*}
    \grad_{\hat\eta}\hat\Omg_{i,p;j,q}
  &=\,
    \hat B^2\cdot\grad_{\eta}\Omg_{\sgm(i),p;\sgm(j),q}
  =
    (\hat B\cdot\xi_{\sgm(i),p})\cdot (\hat B\cdot\xi_{\sgm(j),q})
  =
    \hat\xi_{i,p}\cdot\hat\xi_{j,q},
  \end{align*}
i.e., the functions $\hat\Omg_{i,p;j,q}$ indeed satisfy the $\hat M$-analogue of \eqref{dOmg-1}.
The remainder of this theorem can be verified through simple and straightforward calculations, so we omit the details here. The theorem is proved.
\end{proof}

\section{Topological deformations}

\subsection{Virasoro operators}

In this subsection we are to construct a collection of Virasoro operators
in terms of the complete data \eqref{complete data-1}
for the tau structure $\{\Omg_{i,p;j,q}\}$ of $(M,B)$
constructed by \eqref{240609-2127-1}--\eqref{240609-2127-2}.

Given an $N$-dimensional diagonal matrix $\mu=\diag(\mu_1,...,\mu_N)$,
we introduce the following subspaces of $\bbC^{N\times N}$ for each $s\in\bbZ$:
\begin{align}
  V_{\mu;s} &:=\, \Bigset{R\in\bbC^{N\times N}}{[\mu,R]=sR}, \\
  V_{\mu} &:=\,
    \bigoplus_{s\in\bbZ}V_{\mu;s}.   \label{Vmu-oplus}
\end{align}
Then we have the canonical projections
\begin{equation}
     \pi_s\colon V_{\mu} \to\, V_{\mu;s},\quad
    R \mapsto\, R_{[s]},\quad s\in\bbZ
\end{equation}
with respect to the decomposition \eqref{Vmu-oplus}.

For each $R\in V_\mu$ and integer $m\geq -1$, we denote
\begin{equation} \label{Pm-mu-R}
  \rmP_m(\mu, R):=
    \begin{cases}
      \rme^{R\p_x}
        \prod\limits_{j=0}^{m}
          \left(
            x+\mu+j-\frac12
          \right)
      \Big|_{x=0},& \text{if $m\geq 0$,}
    \\
      1,  & \text{if $m=-1$}.
    \end{cases}
\end{equation}
It is clear that $\rmP_m(\mu, R)$ also belongs to $V_\mu$.

\begin{ex}Let $M$ be an $n$-dimensional generalized Frobenius manifold.
\begin{enumerate}
  \item We take $N=n$, and $\mu$, $R$ to be the monodromy data of $M$ given  in \eqref{eta mu Rs}. Then
  we have $R\in V_\mu$, and $R_{[s]}=R_s$ for $s\in\bbZ$.
  We also note that $R_{[s]}=0$ for $s\leq 0$.

  \item We take $N=n+2$, and $\mu$, $R$ to be the monodromy data $\widetilde{\mu}, \widetilde{R}$ of $\widetilde{M}$ given  in
  \eqref{n+2 eta mu R}. Then we have $\widetilde{R}\in V_{\widetilde{\mu}}$,
  and $\widetilde{R}_{[s]}=\widetilde{R}_{s}$ for $s\in\bbZ$.
  Note that $\widetilde{R}_{s\leq 0}$ does not necessarily vanish,
  it can possess a non-zero entry $a_{00;s}$.

  \item Let $N=n+4$, and $B$ be an invertible quasi-homogeneous Legendre field on $M$. We take $\mu$, $R$ to be the matrices
 $\widetilde{\mu}_B, \widetilde{R}_B$ given in Definition \ref{def:complete data}. Then $\widetilde{R}_B\in V_{\widetilde{\mu}_B}$, and
  $(\widetilde{R}_B)_{[s]}=\widetilde{R}_{B;s}$ for $s\in\bbZ$.
  Note that $\widetilde{R}_{B;s\leq 0}$ does not necessarily vanish,
  it can possess a non-zero block $A_{s}$.
\end{enumerate}
\end{ex}

Let us introduce the $(n+4)$-dimensional operator-valued column vectors
\[\widetilde{\bfa}_p^{(B)}=\left(\widetilde{\bfa}_p^{(B),i}\right)_{i\in\{0,B;1,...,n;0',B'\}}\]
for $p\in\bbZ$ such that
\begin{equation} \label{Heisenberg operator}
  \widetilde{\bfa}_p^{(B),i}
:=
  \begin{cases}
    (-1)^{p+1}t^{0,-p-1}, & i=0,\,p\in\bbZ,\\
    (-1)^{p+1}t^{B,-p-1}, & i=B,\,p\in\bbZ,\\
    \eta^{i\beta}\pp{t^{\beta,p}}, & 1\leq i\leq n,\, p\geq 0,\\
    (-1)^{p+1}t^{i,-p-1}, & 1\leq i\leq n,\, p < 0,\\
    \pp{t^{0,p}}, & i=0',\, p\in\bbZ, \\
    \pp{t^{B,p}}, & i=B',\, p\in\bbZ.
  \end{cases}
\end{equation}
Then the operators $\left\{\widetilde{\bfa}^{(B),i}_p\mid p\in\bbZ,\, i=0,B;1,...,n;0',B'\right\}$
together with the identity $\boldsymbol 1$ form
a realization of the Heisenberg algebra $\widetilde{\mcalH}^{(B)}$
which satisfy the commutation relations
\[
  \left[
    \boldsymbol 1,
    \widetilde{\bfa}^{(B),i}_p
  \right] = 0,
\quad
  \left[
    \widetilde{\bfa}^{(B),i}_p,
    \widetilde{\bfa}^{(B),j}_q
  \right]
  =
  (-1)^p(\widetilde{\eta}_B)^{ij}\delta_{p+q+1,0}\cdot\boldsymbol 1.
\]

A family of $(n+2)$-dimensional operator-valued column vectors $\widetilde{\bfa}_p$ for $p\in\bbZ$
were constructed in \cite{GFM1} by deleting the $B$-th and $B'$-th entries of the above-mentioned $\widetilde{\bfa}_p^{(B)}$,
and a family of linear differential operators $\{L_m\}_{m\geq -1}$ were introduced as
\begin{align}
  L_m(\bft,\pp{\bft})
:=&\,
  \frac12\sum_{p,q\in\bbZ}
    (-1)^{p+1}
    :
        \widetilde{\bfa}_q^\rmT
      \widetilde{\eta}
      \left(
        \rmP_m(\widetilde{\mu}-p, \widetilde{R})
      \right)_{[m-1-q-p]}
      \widetilde{\bfa}_p
    :
\notag\\
&
  +
    \frac14\delta_{m,0}
      \mathrm{tr}\left(\frac14-\mu^2\right)\cdot\boldsymbol 1. \label{Vir-1}
\end{align}
They are called the Virasoro operators of the generalized Frobenius manifold $M$.
In the same way, we define the Legendre-extended Virasoro operators
$\{L^{(B)}_m\}_{m\geq -1}$ of $(M,B)$ by
\begin{align}
  L_m^{(B)}(\bft,\pp{\bft})
:=&\,
  \frac12\sum_{p,q\in\bbZ}
    (-1)^{p+1}
    :
      \left(
        \widetilde{\bfa}_q^{(B)}
      \right)^\rmT
      \widetilde{\eta}_B
      \left(
        \rmP_m(\widetilde{\mu}_B-p, \widetilde{R}_B)
      \right)_{[m-1-q-p]}
      \widetilde{\bfa}_p^{(B)}
    :
\notag\\
&
  +
    \frac14\delta_{m,0}
      \mathrm{tr}\left(\frac14-\mu^2\right)\cdot\boldsymbol 1  \label{Vir-B1}
\end{align}
in terms of the complete data $\widetilde{\mu}_B$, $\widetilde{R}_B$.
Here the normal ordering $:\ :$ in \eqref{Vir-1}--\eqref{Vir-B1}
  means to put the differential operator terms $\pp{t^{i,p}}$ on the right.
By using the method of \cite{FM-Virasoro}, we can verify that these operators also satisfy the commutation relations
\begin{equation}
  \left[
    L_m^{(B)}, L_k^{(B)}
  \right]
=
  (m-k)L_{m+k}^{(B)},\quad m,k\geq -1.
\end{equation}

From the above definition it follows that the Legendre-extended Virasoro operators $L^{(B)}_m$
can be obtained by adding certain terms involving $t^{B,p}$ or $\pp{t^{B,p}}$ to the operators $L_m$,
and the above operators can be represented in the form
\begin{align}
  L_m &=
    a_{m}^{IJ}\frac{\p^2}{\p t^I\p t^J}
   +b_{m;I}^J t^I\pp{t^J}
   +c_{m;IJ} t^It^J + \frac14\delta_{m,0}
      \mathrm{tr}\left(\frac14-\mu^2\right)\cdot\boldsymbol 1,  \label{Lm-VirCoef}
\\
  L_m^{(B)} &=
    a_{m}^{(B);IJ}\frac{\p^2}{\p t^I\p t^J}
   +b_{m;I}^{(B);J} t^I\pp{t^J}
   +c_{m;IJ}^{(B)} t^It^J + \frac14\delta_{m,0}
      \mathrm{tr}\left(\frac14-\mu^2\right)\cdot\boldsymbol 1,  \label{L(B)m-VirCoef}
\end{align}
here the Virasoro coefficients $a_{m}^{IJ}, a_{m}^{(B);IJ}$,
$b_{m;I}^J, b_{m;I}^{(B);J}$ and $c_{m;IJ}, c_{m;IJ}^{(B)}$ satisfy the relations
\[
  a_{m}^{IJ}=a_{m}^{JI},\quad
  a_{m}^{(B);IJ} = a_{m}^{(B);JI}, \quad
  c_{m;IJ} = c_{m;JI}, \quad
  c_{m;IJ}^{(B)} = c_{m;JI}^{(B)},
\]
the uppercase Latin indices $I,J,\dots$ are assumed to range over the index set
$\mcalI_B$, see \eqref{index set IB},
and we assume Einstein summation convention w.r.t. these indices.
We remark that the Virasoro coefficients also possess the properties
\begin{equation}\label{Vir coef rel 0}
  a_m^{B,p;J}
 =b_{m;B,p}^J
 =b_{m;J}^{B,p}
 =c_{m;B,p;J} =0,
\end{equation}
since the operators $L_m$ do not contain terms involving $t^{B,p}$ or $\pp{t^{B,p}}$.
It is clear that
\begin{equation} \label{Vir coef rel 1}
  a_m^{i,p;j,q} = a_m^{(B);i,p;j,q},\quad
  b_{m;i,p}^{j,q} = b_{m;i,p}^{(B);j,q},\quad
  c_{m;i,p;j,q} = c^{(B)}_{m;i,p;j,q}
\quad
\text{if}\ i,j\neq B.
\end{equation}
Moreover, a more involved observation from the definitions
\eqref{Pm-mu-R}--\eqref{Vir-B1} shows that
\begin{equation}\label{Vir coef rel 2}
  a^{(B);IJ}_m = a^{IJ}_m,
\end{equation}
for all $I,J$ (even if $I=(B,p)$ or $J=(B,p)$), and
\begin{align}
  b^{(B);B,p}_{m;i,q} \neq 0, &\quad \text{only if \,$i=B$},  \label{Vir coef rel 3}\\
  b^{(B);i,q}_{m;B,p} \neq 0, &\quad \text{only if \,$i\in\{1,2,\dots,n\}$}  \label{Vir coef rel 4}
\end{align}
for $m\geq -1$, and $p,q\in\bbZ$.

\begin{rmk}
  Let $(\hat M,\hat B)$ be the generalized Legendre transformation of $(M,B)$
  with respect to the invertible quasi-homogeneous Legendre field $B$.
  From the relationships between the complete data given in Theorem \ref{thm:tau correspendence},
  it is clear that the Legendre-extended Virasoro operators of
  $(M,B)$ and $(\hat M,\hat B)$ are related by
  \begin{equation}\label{Viraosro relationship}
     L^{(B)}_m =
  \hat L^{(\hat B)}_m\Big|_{\hat t^{i,p}\mapsto t^{\sgm(i),p},\, \pp{\hat t^{i,p}}\mapsto\pp{t^{\sgm(i),p}}},
  \end{equation}
  here the bijection $\sgm$ is defined in \eqref{involution sgm}.
  Moreover, the above formulae are equivalent to the relationships
  \begin{align}
    \hat a^{(\hat B);i,p;j,q}_m
  =
     a^{(B);\sgm(i),p;\sgm(j),q}_m ,
  \quad
    \hat b^{(\hat B);j,q}_{i,p}
  =
    b^{(B);\sgm(j),q}_{\sgm(i),p},
  \quad
    \hat c^{(\hat B)}_{m;i,p;j,q}
  =
   c^{(B)}_{m;\sgm(i),p;\sgm(j),q}  \label{relation of Vira coef}
  \end{align}
  between the Virasoro coefficients.
\end{rmk}

At the end of this subsection,
we provide the explicit formulae of $L^{(B)}_m$ for $m=-1$, $0,1,2$,
which have the same expressions as the common case given in \cite{FM-Virasoro}
(see also in \cite{GFM1}):
\begin{align}
  L^{(B)}_{-1}=&\,
  \sum_{p\in\bbZ}
      t^{i,p}\pp{t^{i,p-1}}
  +\frac12\eta_{\afa\beta}t^{\afa,0}t^{\beta,0}.
\label{explicit L-1}
\\
  L^{(B)}_0=&\,
    \sum_{p\in\bbZ}
      \left(
        p+\mu_i+\frac12
      \right)
      t^{i,p}\pp{t^{i,p}}
    +\frac12 \sum_{p,q\in\bbZ}
        (-1)^p
        \left(
          \widetilde{\eta}_B\widetilde{R}_{B;p+q+1}
        \right)_{ij} t^{i,p}t^{j,q} \notag \\
&
  +\frac14\mathrm{tr}
     \left(
       \frac14-\mu^2
     \right)\cdot\boldsymbol 1.
\label{explicit L-0}
\\
  L^{(B)}_1=&\,
    \frac12\left(\frac14-\mu^2\right)_\veps^\afa
    \eta^{\veps\beta}
    \frac{\p^2}{\p t^{\afa,0}\p t^{\beta,0}}\notag\\
  &
    +\sum_{p\in\bbZ}
       \left(
         p+\mu_i+\frac12
       \right)
       \left(
         p+\mu_i+\frac32
       \right)
       t^{i,p}\pp{t^{i,p+1}}
\notag\\
  &+2\sum_{p\in\bbZ,\, s\ge 1}
        (p+\mu_i+1)
        \left(\widetilde R_{B;s}\right)_i^\afa
        t^{i,p}\pp{t^{\afa,p+1-s}}
\notag\\
  &
    +\sum_{p\in\bbZ,\, s\geq 2}
      \left(
        (\widetilde{R}^2_B)_{[s]}
      \right)_i^\afa
      t^{i,p}\pp{t^{\afa,p+1-s}}
\notag\\
  &
    +\sum_{p,q\in\bbZ}
       (-1)^{p+1}
       (p+\mu_i+1)
       \left(
         \widetilde{\eta}_B\widetilde{R}_{B;p+q+2}
       \right)_{ij}
      t^{i,p}t^{j,q}
\notag\\
  &
    +\frac12\sum_{p,q\in\bbZ}
       (-1)^p
       \left(
         \widetilde{\eta}_B
        (
          \widetilde{R}_B^2
        )_{[p+q+2]}
       \right)_{ij}
    t^{i,p}t^{j,q}.
\\
  L^{(B)}_2=&\,
    \frac12
    \left(-3\mu_\afa^2+3\mu_\afa+\frac14\right)
    (R_1)^\afa_\gamma\eta^{\gamma\beta}
    \frac{\p^2}{\p t^{\afa,0}\p t^{\beta,0}}
\notag\\
  &
    +
      \left(
        \frac12-\mu_\afa
      \right)
      \left(
        \frac32-\mu_\afa
      \right)
      \eta^{\afa\beta}
      \left(
        \frac12-\mu_\beta
      \right)
      \frac{\p^2}{\p t^{\afa,0}\p t^{\beta,1}}
\notag\\
  &
    +\sum_{p\in\bbZ}
       \left(
         p+\mu_i+\frac12
       \right)
       \left(
         p+\mu_i+\frac32
       \right)
       \left(
         p+\mu_i+\frac52
       \right)
       t^{i,p}\pp{t^{i,p+2}}
\notag\\
  &
    +\sum_{p\in\bbZ,\,s\ge 1}
         \left[
          3\left(
            p+\mu_i+\frac32
          \right)^2 - 1
        \right]
        (\widetilde R_{B;s})^\afa_i
        t^{i,p}\pp{t^{\afa,p+2-s}}
\notag\\
  &
   +\sum_{p\in\bbZ,\, s\geq 2}
     3\left(p+\mu_i+\frac 32\right)
      \left(
        (\widetilde{R}^2_B)_{[s]}
      \right)_i^\afa
      t^{i,p}\pp{t^{\afa,p+2-s}}
\notag\\
  &
   +\sum_{p\in\bbZ,\, s\geq 3}
      \left(
        (\widetilde{R}_B^3)_{[s]}
      \right)_i^\afa
      t^{i,p}\pp{t^{\afa,p+2-s}}
\notag\\
  &
    +\frac12
     \sum_{p,q\in\bbZ}(-1)^p
       \Bigg[
         \widetilde{\eta}_B(\widetilde R^3_B)_{[p+q+3]}
        -3\left(p+\mu_i+\frac32\right)
          \widetilde{\eta}_B(\widetilde R^2_B)_{[p+q+3]}
\notag\\
  &\quad
    +\left(
       3\left(p+\mu_i+\frac 32\right)^2-1
     \right)\widetilde{\eta}_B(\widetilde R_B)_{[p+q+3]}
    \Bigg]_{ij}
  t^{i,p}t^{j,q}.
       \label{explicit L2}
\end{align}

\subsection{Virasoro symmetries}

For a generalized Frobenius manifold $\gfme$,
a family of Virasoro symmetries $\{\frac{\p}{\p s_m}\}_{m\geq -1}$
of the tau-cover of the Principal Hierarchy \eqref{PH} is given by
\begin{align}
  \pfrac{f}{s_m}&=
    a_m^{IJ}f_If_J+b_{m;J}^It^Jf_I+c_{m;IJ}t^It^J, \label{vir sym 1}\\
  \pfrac{f_I}{s_m}&=
    \pp{t^I}\pfrac{f}{s_m}=(2a_m^{KL}\Omg_{IK}+b^L_{m;I})f_L
    +(b^K_{m;L}\Omg_{IK}+2 c_{m;IL})t^L, \label{vir sym 2}\\
  \pfrac{v^\afa}{s_m}
&=
  \eta^{\afa\beta}\p_x\pp{t^{\beta,0}}\pfrac{f}{s_m}, \label{vir sym 3}
\end{align}
details of which can be found in \cite{GFM2,GFM1} (see also in \cite{normal-form}),
here the tau structure $\{\Omg_{IJ}\}$ is constructed in
\eqref{240609-2127-1}--\eqref{240609-2127-2},
and the Virasoro coefficients $a_m^{IJ}, b_{m;J}^I$ and $c_{m;IJ}$
are given as in \eqref{Lm-VirCoef}.
We emphasize that although the tau-cover of the Principal Hierarchy of $M$
does not involve the Legendre-extended flows $\pp{t^{B,p}}$,
the uppercase Latin indices $I,J,K,L$ in \eqref{vir sym 1}--\eqref{vir sym 3}
are still assumed to range over the larger index set $\mcalI_B$ \eqref{index set IB}
for a fixed invertible quasi-homogeneous Legendre field $B$.

From \eqref{vir sym 1}--\eqref{vir sym 3} it follows that
\begin{align*}
  \pfrac{v^\afa}{s_m}
=&\,
  \eta^{\afa\beta}
  \Big(
    2a_m^{KL}(\p_x\Omg_{\beta,0;K})f_L
   +(2a_m^{KL}\Omg_{\beta,0;K}+b_{m;\beta,0}^L)\theta_L
\\
&\,
   +b_{m;L}^K(\p_x\Omg_{\beta,0;K})t^L
   +(b_{m;0,0}^K\Omg_{\beta,0;K}
      +2c_{m;\beta,0;0,0}
     )
   \Big),
\end{align*}
where we recall that $\theta_L=\Omg_{0,0;L}$ and $x=t^{0,0}$.
On the other hand, it was proved in \cite{GFM1} that \eqref{vir sym 3} can be alternatively represented in the form
\[
  \pfrac{v^\afa}{s_m} = (E^{m+1})^\afa
+
  \text{terms involving time variables $t^I$ or 1-point functions $f_I$},
\]
where $E$ is the Euler vector field of $M$. Therefore we obtain
\begin{align}
  (E^{m+1})^\afa &=
\eta^{\afa\beta}
  \Big(
    (2a_m^{KL}\Omg_{\beta,0;K}+b_{m;\beta,0}^L)\theta_L
   + (b_{m;0,0}^K\Omg_{\beta,0;K}
      +2c_{m;\beta,0;0,0}
     )
  \Big), \label{E^m+1}
\\
  \pfrac{v^\afa}{s_m} &=\, (E^{m+1})^\afa
+
  \eta^{\afa\beta}
  (\p_x\Omg_{\beta,0;K})
  \left(
    2a_m^{KL}f_L
   +b^K_{m;L}t^L
  \right)
\notag \\
&=
  (E^{m+1})^\afa
+
  \pfrac{v^\afa}{t^K}
  \left(
    2a_m^{KL}f_L
   +b^K_{m;L}t^L
  \right).  \label{Vir-sym of v^afa}
\end{align}
It was also proved in \cite{GFM1} that the relations
\[
  \left[
    \pp{s_m},\pp{t^I}
  \right]f
 =\left[
    \pp{s_m},\pp{t^I}
  \right]f_J
 =\left[
    \pp{s_m},\pp{t^I}
  \right]v^\afa
 =0
\]
hold true for 
$m\geq -1$, $1\leq\afa\leq n$.
However, $\left[
    \pp{s_m},\pp{t^{B,p}}
  \right]f_{J}$ and $\left[
    \pp{s_m},\pp{t^{I}}
  \right]f_{B,p}$ do not necessarily vanish.
Now we are to construct a family of so called \textit{Legendre-extended Virasoro symmetries}
$\left\{\pp{s^{(B)}_m}\right\}_{m\geq -1}$ of the Legendre-extended tau-cover \eqref{tau cover} which
commute with all the flows $\pp{t^I}$ of \eqref{extended B-PH-2}.
\begin{defn}
The family of Legendre-extended Virasoro symmetries $\left\{\pp{s^{(B)}_m}\right\}_{m\geq -1}$
of the Legendre-extended tau-cover \eqref{tau cover} are defined by
\begin{align}
  \pfrac{f}{s_m^{(B)}}&=
    a_m^{(B);IJ}f_If_J+b_{m;J}^{(B);I}t^Jf_I+c_{m;IJ}^{(B)}t^It^J, \label{B-vir sym 1}\\
  \pfrac{f_I}{s_m^{(B)}}&=
    \pp{t^I}\pfrac{f}{s_m^{(B)}}=(2a_m^{(B);KL}\Omg_{IK}+b^{(B);L}_{m;I})f_L
    +(b^{(B);K}_{m;L}\Omg_{IK}+2 c^{(B)}_{m;IL})t^L, \label{B-vir sym 2}\\
  \pfrac{v^\afa}{s_m^{(B)}}
&=
  \eta^{\afa\beta}\p_x\pp{t^{\beta,0}}\pfrac{f}{s_m},  \label{B-vir sym 3}
\end{align}
here the Legendre-extended Virasoro coefficients
$a_m^{(B);IJ}, b_{m;J}^{(B);I}$ and $c_{m;IJ}^{(B)}$ are given by \eqref{L(B)m-VirCoef}.
\end{defn}

From \eqref{E^m+1} and the relations \eqref{Vir coef rel 1}--\eqref{Vir coef rel 4}
between the two types of Virasoro coefficients, equation
\eqref{B-vir sym 3} can be reformulated as
\begin{equation}\label{v^afa sBm}
  \pfrac{v^\afa}{s_m^{(B)}} = (E^{m+1})^\afa
+
  \pfrac{v^\afa}{t^K}
  \left(
    2a_m^{(B);KL}f_L
   +b^{(B);K}_{m;L}t^L
  \right),
\end{equation}
which is analogous to \eqref{Vir-sym of v^afa}.

\begin{thm}
The above-mentioned flows $\left\{\pp{s^{(B)}_m}\right\}_{m\geq -1}$ satisfy
the commutation relations
\begin{equation}
\Big[\pp{s^{(B)}_m}, \pp{t^I}\Big]=0
\end{equation}
for $m\geq-1$ and $I\in\{0,B;1,2,...,n\}\times\bbZ$.
\end{thm}

\begin{proof}
The proof is similar to the one for the Virasoro symmetries of the tau-cover of the usual Principal Hierarchy of a generalized Frobenius manifold given in  \cite{GFM1}.
From definitions \eqref{B-vir sym 1}--\eqref{B-vir sym 3} we see that
$\left[\pp{s^{(B)}_m}, \pp{t^I}\right]f=0$ hold true trivially;
and if $\left[\pp{s^{(B)}_m}, \pp{t^I}\right]f_J=0$ hold true,
then it is easy to verify the validity of $\left[\pp{s^{(B)}_m}, \pp{t^I}\right]v^\afa=0$.
Therefore we only need to verify the identities
\begin{equation}\label{240616-1649}
  \Big[\pp{s^{(B)}_m}, \pp{t^I}\Big]f_J=0,\quad I,J\in\{0,B;1,2,...,n\}\times\bbZ.
\end{equation}
By using \eqref{tau cover},
\eqref{B-vir sym 1}--\eqref{B-vir sym 3} and \eqref{v^afa sBm}, we obtain
\begin{align*}
  &\left(
    \pp{s_m^{[B]}}\circ\pp{t^I}
  \right)f_J
=
  \pp{s_m^{(B)}}\Omg_{IJ}
=
  (\p_\afa\Omg_{IJ})\pfrac{v^\afa}{s_m^{(B)}} \\
&\qquad =
  (\p_\afa\Omg_{IJ})
  \left(
    (E^{m+1})^\afa
   +\pfrac{v^\afa}{t^K}
    (2a_m^{(B);KL}f_L+b_{m;L}^{(B);K}t^L)
  \right) \\
&\qquad =
  \mcalL_{E^{m+1}}\Omg_{IJ}
 +\pfrac{\Omg_{IJ}}{t^K}
   \left(2a_m^{(B);KL}f_L+b_{m;L}^{(B);K}t^L\right),
\\
 & \left(
    \pp{t^I}\circ\pp{s_m^{[B]}}
  \right)f_J
=
  \pp{t^I}
  \left(
    (2a_m^{(B);KL}\Omg_{JK}+b^{(B);L}_{m;J})f_L
    +(b^{(B);K}_{m;L}\Omg_{JK}+2 c^{(B)}_{m;JL})t^L
  \right) \\
&\qquad=
  2a_{m}^{(B);KL}
  \pfrac{\Omg_{JK}}{t^I}f_L
 +\left(2a_m^{(B);KL}\Omg_{JK}+b_{m;J}^{(B);L}\right)\Omg_{IL} \\
&\qquad\quad
  +b_{m;L}^{(B);K}
   \pfrac{\Omg_{JK}}{t^I}t^L
  +b_{m;I}^{(B);K}\Omg_{JK}
  +2c_{m;IJ}^{(B)}.
\end{align*}
We also note that $\pfrac{\Omg_{IJ}}{t^K}=\frac{\p^3f}{\p t^I\p t^J\p t^K}
=\pfrac{\Omg_{JK}}{t^I}$, therefore \eqref{240616-1649} is equivalent to
\begin{equation}\label{pEOmg}
  \mcalL_{E^{m+1}}\Omg_{IJ}
=
  2a_{m}^{(B);KL}\Omg_{IK}\Omg_{LJ}
 +\left(
   b^{(B);K}_{m;I}\Omg_{KJ}
  +b^{(B);K}_{m;J}\Omg_{IK}
  \right)
 +2c^{(B)}_{m;IJ}.
\end{equation}
It suffices to verify the validity of the above identities for $m=-1,0,1,2$.
The identities for $m=-1,0$ have already been proved in \eqref{p0Omg}--\eqref{p1Omg};
for $m=1,2$, we remark that the way to calculate $\mcalL_{E^2}\Omg_{0,p;0,q}$ and
$\mcalL_{E^3}\Omg_{0,p;0,q}$ in the Appendix B of \cite{GFM1}
can be applied to calculate $\mcalL_{E^2}\Omg_{IJ}, \mcalL_{E^3}\Omg_{IJ}$
for $I,J\in\mcalI_B$,
and we omit the details here.
The theorem is proved.
\end{proof}

\subsection{Linearization of the Legendre-extended Virasoro symmetries}

Let $\gfme$ be a generalized Frobenius manifold.
Consider a \textit{quasi-Miura transformation} \cite{normal-form,GFM2,GFM1}
for the Principal Hierarchy \eqref{PH} of the form
\begin{equation}\label{quasi-Miura transf}
  v^\afa\mapsto w^\afa=v^\afa+\veps^2\eta^{\afa\gamma}\p_x\p_{t^{\gamma,0}}\Delta\mcalF,
\end{equation}
where
\begin{equation}\label{zh-06}
\Delta\mcalF=\sum_{k\geq 1}\veps^{k-2}\mcalF^{[k]},\quad
\mcalF^{[k]}=\mcalF^{[k]}\bigl(v,v_x,v_{xx},...,v^{(N_k)}\bigr)
\end{equation}
and $N_k$ are certain non-negative integer numbers.
By representing the Principal Hierarchy in terms of the unknown functions $w^1,\dots,w^n$,
we get a deformation of this integrable hierarchy.

We say that the transformation \eqref{quasi-Miura transf}
linearizes the Virasoro symmetries \eqref{vir sym 1}--\eqref{vir sym 3}
if the actions of $\left\{\pp{s_m}\right\}$ on the \textit{tau function}
\begin{equation}\label{tau=exp F}
\tau= \exp(\mcalF)=\exp({\veps^{-2}f+\Delta\mcalF})
\end{equation}
of the deformed integrable hierarchy are given by
\begin{equation}\label{linearization condition}
  \frac{\p\tau}{\p s_m}=L_m\Bigl(\veps^{-1}\bft,\veps\pp{\bft}\Bigr)\tau,\quad m\geq -1,
\end{equation}
here the Virasoro operators $L_m$ are defined in \eqref{Vir-1}.
We can represent the right hand side of \eqref{linearization condition} in the form
\begin{align}
&L_m\Bigl(\veps^{-1}\bft,\veps\pp{\bft}\Bigr)\tau
=
  \left[
    \mcalD_m\Delta\mcalF+\veps^{-2}\pfrac{f}{s_m}
    +a_m^{IJ}\Omg_{IJ}\right. \notag\\
  &\qquad  +\veps^2 a_m^{IJ}
  \left.
      \left(
        \frac{\p^2\Delta\mcalF}{\p t^I\p t^J}
       +\pfrac{\Delta\mcalF}{t^I}\pfrac{\Delta\mcalF}{t^J}
      \right)
      +\frac{\delta_{m,0}}{4}\mathrm{tr}
          \left(\frac 14-\mu^2\right)
  \right]\tau,  \label{230831-1644}
\end{align}
where the operators $\mcalD_m$ for $m\geq -1$ are defined by
\begin{equation}\label{def of mcalD m}
  \mcalD_m= 2a_m^{IJ}f_I\pp{t^J}+b^J_{m;I}t^I\pp{t^J}.
\end{equation}
As it is shown in \cite{GFM1} (see also in \cite{normal-form}),
for any functions $\Delta\mcalF$ of the form \eqref{zh-06}
we have the following identities:
\begin{align}\label{240617-1905-1}
&\mcalD_m\Delta\mcalF = K_m \Delta\mcalF+\pfrac{\Delta\mcalF}{s_m},\\
&\sum_{m\geq -1}\frac{1}{\lmd^{m+2}}
\left[
a_m^{IJ}\Omg_{IJ} + \veps^2 a_m^{IJ}
\left(
        \frac{\p^2\Delta\mcalF}{\p t^I\p t^J}
       +\pfrac{\Delta\mcalF}{t^I}\pfrac{\Delta\mcalF}{t^J}
      \right)
\right]
\notag\\
&\qquad=
   \frac{\veps^2}{2}
 \sum_{k,\ell\geq 0}
 \left(
   \pfrac{\Delta\mcalF}{v^{\afa,k}}
   \pfrac{\Delta\mcalF}{v^{\beta,\ell}}
  +\frac{\p^2\Delta\mcalF}{\p v^{\afa,k}\p v^{\beta,\ell}}
 \right)
 (\p_x^{k+1}\p^\afa p_\sigma)
 G^{\sigma\rho}
 (\p_x^{\ell+1}\p^\beta p_\rho)
\notag\\
&\qquad\quad
 +\frac{\veps^2}{2}
  \sum_{k\geq 0}
  \pfrac{\Delta\mcalF}{v^{\afa,k}}
  \p_x^{k+1}
  \left(
    \grad_\eta\left(\pfrac{p_\sigma}{\lmd}\right)
    \cdot
    \grad_\eta\left(\pfrac{p_\rho}{\lmd}\right)
    \cdot v_x
  \right)^\afa
  G^{\sigma\rho}\notag\\
&\qquad\quad+
 \frac12 G^{\afa\beta}
 \pfrac{p_\afa}{\lmd}*\pfrac{p_\beta}{\lmd},  \label{240617-1905-2}
\end{align}
where the vector fields
$K_m=\sum\limits_{s\geq 0}K_m^{\gamma,s}\pp{v^{\gamma,s}}$ on the jet space $J^\infty(M)$
are defined by the following generating series:
\[
\sum_{m\geq -1}\frac{K_m^{\gamma,s}}{\lmd^{m+2}}
=
        (s+1)\p_x^s\left(\frac{1}{E-\lmd e}\right)^\gamma
      +
    \sum_{k=1}^{s}
      k{s+1\choose k+1}
        \left(
          \p_x^{s-k}\p^\gamma p_\afa
        \right)
        G^{\afa\beta}
        \left(
          \pp{\lmd}
          \p_x^k p_\beta
        \right),
\]
$p_1,\dots,p_n$ are any given basis of \textit{periods} of $M$,
$G=(G^{\afa\beta})$ is the associated Gram matrix,
and the \textit{star product} $*$ is defined in \cite{normal-form, GFM1}.
Hence the linearization condition \eqref{linearization condition}
is equivalent to the following \textit{loop equation} \cite{GFM2,GFM1}
\begin{align}
&
  \sum_{s\geq 0}
    \pfrac{\Delta\mcalF}{v^{\gamma,s}}
    (s+1)\p_x^s\left(\frac{1}{E-\lmd e}\right)^\gamma \notag \\
&\qquad+
      \sum_{s\geq 0}
    \pfrac{\Delta\mcalF}{v^{\gamma,s}}
    \sum_{k=1}^{s}
      k{s+1\choose k+1}
        \left(
          \p_x^{s-k}\p^\gamma p_\afa
        \right)
        G^{\afa\beta}
        \left(
          \pp{\lmd}
          \p_x^k p_\beta
        \right)
\notag\\
&\quad=
 \frac{\veps^2}{2}
 \sum_{k,\ell\geq 0}
 \left(
   \pfrac{\Delta\mcalF}{v^{\afa,k}}
   \pfrac{\Delta\mcalF}{v^{\beta,\ell}}
  +\frac{\p^2\Delta\mcalF}{\p v^{\afa,k}\p v^{\beta,\ell}}
 \right)
 (\p_x^{k+1}\p^\afa p_\sigma)
 G^{\sigma\rho}
 (\p_x^{\ell+1}\p^\beta p_\rho)
\notag\\
&\qquad
 +\frac{\veps^2}{2}
  \sum_{k\geq 0}
  \pfrac{\Delta\mcalF}{v^{\afa,k}}
  \p_x^{k+1}
  \left(
    \left(\grad_\eta\pfrac{p_\sigma}{\lmd}\right)
    \cdot
    \left(\grad_\eta\pfrac{p_\rho}{\lmd}\right)
    \cdot v_x
  \right)^\afa
  G^{\sigma\rho}
\notag
\\&\qquad
+
 \frac12 G^{\afa\beta}
 \pfrac{p_\afa}{\lmd}*\pfrac{p_\beta}{\lmd}-
 \frac{1}{4\lmd^2}
 \mathrm{tr}
 \left(
   \frac14-\mu^2
 \right)\label{loop equation-2308}
\end{align}
of the generalized Frobenius manifold $M$ with the unknown function $\Delta\mcalF$.
The above equation is required to hold true identically with respect to the formal parameter $\lmd$.
It was proved in \cite{GFM2} that if the generalized Frobenius manifold $M$ is semisimple,
then there exists a unique (up to adding constant terms) solution $\Delta\mcalF$ to the loop equation \eqref{loop equation-2308},
which is of the form
\begin{equation}\label{Delta F}
  \Delta\mcalF =
  \sum_{g\geq 1}
    \veps^{2g-2}
    \mcalF_g\bigl(v,v_x,...,v^{(3g-2)}\bigr).
\end{equation}
Here the functions $\mcalF_g$ are called the genus $g$ free energies of $M$.
The quasi-Miura transformation \eqref{quasi-Miura transf} induced by the solution
$\Delta\mcalF$ to \eqref{loop equation-2308} yields a deformation of the Principal Hierarchy \eqref{PH}
which is called the \textit{topological deformation} of the Principal Hierarchy of the generalized Frobenius manifold $M$.

Now let $B$ be an invertible quasi-homogeneous Legendre field of $M$,
we can also impose the linearization condition
\begin{equation}\label{B-linearization condition}
  \frac{\p\tau}{\p s^{(B)}_m}
  =L^{(B)}_m\Bigl(\veps^{-1}\bft,\veps\pp{\bft}\Bigr)\tau,\quad m\geq -1
\end{equation}
of the Legendre-extended Virasoro symmetries \eqref{B-vir sym 1}--\eqref{B-vir sym 3} on the tau function \eqref{tau=exp F},
which are analogous to \eqref{linearization condition}.
Here $L_m^{(B)}$ and $\pp{s_m^{(B)}}$ are introduced in \eqref{Vir-B1} and
\eqref{B-vir sym 1}--\eqref{B-vir sym 3} respectively.
\begin{thm}\label{Thm:B-loop eqn}
  The Legendre-extended analogue of the linearization condition \eqref{B-linearization condition}
is equivalent to the loop equation \eqref{loop equation-2308}.
\end{thm}
\begin{proof}
By straightforward calculation, we obtain the following analogues
of \eqref{230831-1644}--\eqref{def of mcalD m}:
\begin{align}
&\mcalD^{(B)}_m:=2a_m^{(B);IJ}f_I\pp{t^J}+b^{(B);J}_{m;I}t^I\pp{t^J},
  \label{def of mcalD^B m}\\
&L_m^{(B)}\Bigl(\veps^{-1}\bft,\veps\pp{\bft}\Bigr)\tau
=
  \Bigl[
    \mcalD_m^{(B)}\Delta\mcalF+\veps^{-2}\pfrac{f}{s_m^{(B)}} +a_m^{(B);IJ}\Omg_{IJ} \notag \\
    &\quad
    +\veps^2 a_m^{(B);IJ}
      \left(
        \frac{\p^2\Delta\mcalF}{\p t^I\p t^J}
       +\pfrac{\Delta\mcalF}{t^I}\pfrac{\Delta\mcalF}{t^J}
      \right)
      +\frac{\delta_{m,0}}{4}\mathrm{tr}
          \left(\frac 14-\mu^2\right)
  \Bigr]\tau.
  \end{align}
To prove this theorem, it suffices to verify that the identities
\begin{equation}\label{240617-2002}
  \left(
    \pp{s^{(B)}_m} - \pp{s_m}
  \right)\Delta\mcalF
=
  \left(
    \mcalD_m^{(B)} - \mcalD_m
  \right)\Delta\mcalF,\quad m\ge -1
\end{equation}
hold true for all functions $\Delta\mcalF$ on the jet space $J^\infty(M)$,
because of the relations \eqref{Vir coef rel 2} and
\eqref{240617-1905-1}--\eqref{240617-1905-2}. Denote
\[
  \Delta b_{m;J}^I:= b_{m;J}^{(B);I} - b_{m;J}^I,
\]
then from \eqref{Vir coef rel 2}, \eqref{Vir-sym of v^afa} and \eqref{v^afa sBm} we arrive at
\begin{align*}
  \left(
    \pp{s^{(B)}_m} - \pp{s_m}
  \right)\Delta\mcalF
&=
  \sum_{s\geq 0}
    \pfrac{\Delta\mcalF}{v^{\afa,s}}
    \p_x^s
      \left(
        \pfrac{v^\afa}{s^{(B)}_m}
       -\pfrac{v^\afa}{s_m}
      \right)
=
  \sum_{s\geq 0}
    \pfrac{\Delta\mcalF}{v^{\afa,s}}
    \p_x^s
      \left(
        \Delta b_{m;J}^I t^J\pfrac{v^\afa}{t^I}
      \right),
\\
    \left(
    \mcalD_m^{(B)} - \mcalD_m
  \right)\Delta\mcalF
&=
  \Delta b_{m;J}^It^J\pfrac{\Delta\mcalF}{t^I}
=
  \Delta b_{m;J}^It^J
  \sum_{s\geq 0}
    \pfrac{\Delta\mcalF}{v^{\afa,s}}
    \p_x^s\left(\pfrac{v^\afa}{t^I}\right).
\end{align*}
Note that the spatial variable $x$ is always identified with the time variable $t^{0,0}$,
and the relations \eqref{Vir coef rel 0}--\eqref{Vir coef rel 4} imply
\[
  \Delta b_{m;0,0}^J=0,
\]
therefore we arrive at \eqref{240617-2002}. The theorem is proved.
\end{proof}
We remark that the loop equation \eqref{loop equation-2308} depends only on the generalized Frobenius manifold $M$,
and is independent of the choice of the Legendre field $B$.

\subsection{Relationship between the deformed integrable hierarchies}
\label{subsection:full genera}

Let $\gfme$ be a generalized Frobenius manifold,
and $B$ be an invertible quasi-homogeneous Legendre field on $M$.
Denote $(\hat M, \hat B)$ the generalized Legendre transformation of $(M, B)$,
where $\hat B:= B^{-1}$.
By using the relationship between the Legendre-extended tau-covers of
$(M, B)$ and $(\hat M,\hat B)$ that is given in Theorem \ref{thm:tau correspendence},
we apply the following identifications
\begin{equation}
  \hat f = f,\qquad
  \hat f_{i,p} = f_{\sgm(i),p},\quad
  \hat t^{i,p} = t^{\sgm(i),p},\quad
  \pp{\hat t^{i,p}} = \pp{t^{\sgm(i),p}}
\end{equation}
for their Legendre-extended tau-covers,
where $\sgm$ is defined in \eqref{involution sgm}.
Then the relationship \eqref{Viraosro relationship}
between the Legendre-extended Virasoro operators reads
\begin{equation}\label{zh-11}
  \hat L_m^{(\hat B)}
  \left(
    \veps^{-1}\hat\bft, \veps\pp{\hat\bft}
  \right)
=
   L_m^{(B)}
  \left(
    \veps^{-1}\bft, \veps\pp{\bft}
  \right),\quad m\ge -1,
\end{equation}
moreover, from \eqref{relation of Vira coef} and the definitions of the Legendre-extended Virasoro symmetries
it is clear that
\begin{equation}\label{zh-12}
 \pp{\hat s^{(\hat B)}_m} = \pp{s^{(B)}_m},\quad m\ge -1.
\end{equation}

The quasi-Miura transformation \eqref{quasi-Miura transf}, \eqref{zh-06}
for the Principal Hierarchy \eqref{PH} of $M$ yields a quasi-Miura transformation
\begin{equation}\label{zh-07}
  \hat{v}^\afa\mapsto \hat{w}^\afa=\hat{v}^\afa+\veps^2\hat{\eta}^{\afa\gamma}\p_{\hat{x}}\p_{\hat{t}^{\gamma,0}}\Delta\hat{\mcalF},
\end{equation}
for the Principal Hierarchy of $\hat M$,
where $\hat{v}^1,\dots, \hat{v}^n$ are related with $v^1,\dots,v^n$ via \eqref{zh-09}, \eqref{zh-08}, and
\begin{align}
&\Delta\hat\mcalF=\sum_{k\geq 1}\veps^{k-2}\hat\mcalF^{[k]}\bigl(\hat{v},\hat{v}_{\hat{x}},\hat{v}_{\hat{x}\hat{x}},...,\hat{v}^{(N_k)}\bigr),\\
&\hat\mcalF^{[k]}\bigl(\hat{v},\hat{v}_{\hat{x}},\hat{v}_{\hat{x}\hat{x}},...,\hat{v}^{(N_k)}\bigr)=\mcalF^{[k]}\bigl(v,v_x,v_{xx},...,v^{(N_k)}\bigr).\label{zh-10}
\end{align}
We note that the $x$-derivatives of $v^\al$ are related with the $\hat{x}$-derivatives of $\hat{v}^\al$ via $\frac{\p}{\p\hat{x}}=\frac{\p}{\p t^{B,0}}$. For example, from \eqref{240701-1622} we have
\[v_x=\hat{B}\cdot \hat{v}_{\hat x}.\]
From \eqref{zh-11}, \eqref{zh-12} it follows that the quasi-Miura transformation
\eqref{quasi-Miura transf}, \eqref{zh-06} linearizes the Legendre-extended Virasoro symmetries of $(M, B)$ if and only if the quasi-Miura transformation
\eqref{zh-07}--\eqref{zh-10}  linearizes the Legendre-extended Virasoro symmetries of $(\hat M, \hat{B})$, i.e.,
\begin{equation*}
  \pfrac{\tau}{s_m^{(B)}}=L_m^{(B)}
  \left(
    \veps^{-1}\bft, \veps\pp{\bft}
  \right)\tau
\quad
\text{if and only if}\quad
  \pfrac{\tau}{\hat s_m^{(\hat B)}}= \hat L_m^{(\hat B)}
  \left(
    \veps^{-1}\hat\bft, \veps\pp{\hat\bft}
  \right)\tau
\end{equation*}
for $m\geq -1$,
where the tau function is given by
\[\tau:=\exp(\veps^{-2}f+\Delta\mcalF)
=\exp(\veps^{-2}\hat f+\Delta\hat\mcalF).\]
From Theorem \ref{Thm:B-loop eqn} we know that
the above linearization conditions are equivalent to the loop equations
\eqref{loop equation-2308} of $M$ and $\hat M$ respectively,
therefore $\Delta\mcalF$ satisfies the loop equation of $M$
if and only if $\Delta\hat\mcalF$ satisfies the loop equation of $\hat M$.

If the generalized Frobenius manifold $M$ is semisimple,
then $\hat M$ is also semisimple because $M$ and
$\hat M$ share the same multiplication.
So in this case, from \cite{GFM2} we know that the loop equation \eqref{loop equation-2308} of $M$
and the one of $\hat M$ have unique solutions
\[ \Delta\mcalF =
  \sum_{g\geq 1}
    \veps^{2g-2}
    \mcalF_g\bigl(v; v_x,\dots,v^{(2g-3)}\bigr),
    \quad
    \Delta\hat\mcalF =
  \sum_{g\geq 1}
    \veps^{2g-2}
    \hat\mcalF_g\bigl(\hat{v}; \hat{v}_{\hat x},\dots,\hat{v}^{(2g-3)}\bigr)\]
up to the addition of constant terms, and we have the following theorem.

\begin{thm}\label{thm: sol to loop eqn}
Let $(\hat M,\hat B)$ be the generalized Legendre transformation of $(M, B)$,
and $M$ is semisimple, then the solutions $\Delta\mcalF$, $\Delta\hat\mcalF$ of their loop equations satisfy the relations
\begin{equation}
  \mcalF_g(v,v_x,...,v^{(3g-2)}) =
  \hat\mcalF_g(\hat v,\hat v_{\hat x},...,\hat v^{(3g-2)}) + \text{const.}
\end{equation}
for $g\geq 1$.  The topological deformations of the Legendre-extended Principal Hierarchies of $(M, B)$ and $(\hat M,\hat B)$ that are obtained by the quasi-Miura transformations \eqref{quasi-Miura transf} and \eqref{zh-07} share the same tau function, and they are related by the linear reciprocal transformation \eqref{reciprocal}.
\end{thm}

\section{Example: the KdV and the $q$-deformed KdV hierarchies}
\subsection{Generalized Frobenius manifolds and their complete data}
Consider the well-known 1-dimensional Frobenius manifold $M$ with potential
\begin{equation}\label{KdV FM potential}
  F = \frac {v^3}6,
\end{equation}
here $v$ is the flat coordinate with respect to the flat metric $\eta = \td v\otimes\td v$.
The unit vector field $e$ and Euler field $E$ of this Frobenius manifold are
\begin{equation}
  e=\p_v,\quad E=v\p_v.
\end{equation}
Moreover, this Frobenius manifold has charge $d=0$, and monodromy data $\mu=R=0$.

Introduce an invertible quasi-homogeneous Legendre field
\begin{equation}
  B = \frac{1}{2v^{\frac12}}\p_v
\end{equation}
on $M$ with $\mu_B=-\frac12$.
The Legendre field $B$ transforms $M$ to a new generalized Frobenius manifold $\hat M$ with potential
\begin{equation}\label{v^4/12}
  \hat F = \frac{\hat v^4}{12}
\end{equation}
and new metric $\hat\eta = \td\hat v\otimes\td\hat v$, where the new coordinate $\hat v$ reads
\begin{equation} \label{KdV-Vol coordinates}
  \hat v = v^{\frac12}.
\end{equation}
Under the new coordinate $\hat v$, the unit vector field $e$ and Euler field $E$ have the form
\begin{equation}
  e = \frac{1}{2\hat v}\p_{\hat v}, \qquad
  E=\frac12\hat v\p_{\hat v},
\end{equation}
and the charge $\hat d$ of $\hat M$ is $\hat d = -2\mu_B=1$.
Moreover, the inversion $\hat B:=B^{-1}$ of $B$ is
\begin{equation}
  \hat B = \p_{\hat v},
\end{equation}
in particular, $\hat B$ is flat with respect to the new metric $\hat\eta$.

The Legendre-extended calibrations $\Xi_B$, $\hat\Xi_{\hat B}$
of $(M, B)$ and $(\hat M,\hat B)$ are given by
\begin{equation*}
\begin{array}{c|c}
\begin{split}
  \xi_{0,p}&=\,\xi_{1,p}=\frac{v^p}{p!}\p_v, \\
  \xi_{B,p}&=\, \frac{2^{p-1}}{(2p-1)!!}v^{p-\frac12}\p_v, \\
  \xi_{0,-q}&=\, 0, \\
  \xi_{B,-q}&=\, (-1)^q\frac{(2q-1)!!}{2^{q+1}}v^{-q-\frac12}\p_v,
\end{split}
&
\begin{split}
  \hat\xi_{\hat B,p}&=\,\hat\xi_{1,p} = \frac{\hat v^{2p}}{p!}\p_{\hat v}, \\
  \hat\xi_{0,p} &=\, \frac{2^{p-1}}{(2p-1)!!}\hat v^{2p-1}\p_{\hat v}, \\
  \hat\xi_{\hat B,-q} &=\, 0, \\
  \hat\xi_{0,-q} &=\, (-1)^q\frac{(2q-1)!!}{2^{q+1}}\hat v^{-2q-1}\p_{\hat v}
\end{split}
\end{array}
\end{equation*}
for $p\geq 0$ and $q>0$.
Note that the above formulae can be uniformly written as
\begin{align}
  \xi_{0,p}=\xi_{1,p}=\frac{v^p}{\Gamma(p+1)}\p_v,\qquad
  \xi_{B,p}=\frac12\frac{\Gamma(\frac12)}{\Gamma(p+\frac12)}v^{p-\frac12}\p_v, \\
  \hat\xi_{\hat B,p}
 =\hat\xi_{1,p} =  \frac{\hat v^{2p}}{\Gamma(p+1)}\p_{\hat v},\qquad
  \hat\xi_{0,p}
 =\frac12\frac{\Gamma(\frac12)}{\Gamma(p+\frac12)}\hat v^{2p-1}\p_{\hat v},
\end{align}
for $p\in\bbZ$, and these vector fields are related by \eqref{rel1}.
From \eqref{extended B-PH-1} we know that
the corresponding Legendre-extended Principal Hierarchies are given by
\begin{align}
  \pfrac{v}{t^{0,p}}&=\pfrac v{t^{1,p}}=\frac{v^p}{\Gamma(p+1)}v_x,\qquad
  \pfrac v{t^{B,p}}=\frac12\frac{\Gamma(\frac12)}{\Gamma(p+\frac12)}v^{p-\frac12}v_x, \label{KdV-BPH-1}\\
  \pfrac{\hat v}{\hat t^{\hat B,p}}
 &=\pfrac{\hat v}{\hat t^{1,p}} =  \frac{2\hat v^{2p+1}}{\Gamma(p+1)}\hat v_{\hat x},\qquad
  \pfrac{\hat v}{\hat t^{0,p}}
 =\frac{\Gamma(\frac12)}{\Gamma(p+\frac12)}\hat v^{2p}\hat v_{\hat x},  \label{KdV-BPH-2}
\end{align}
here $p\in\bbZ$.
Note that the flows $\left\{\pp{t^{1,p}}\right\}_{p\geq 0}$ form the well-known \textit{dispersionless KdV hierarchy},
which is extended by a family of additional flows $\left\{\pp{t^{B,p}}\right\}_{p\in\bbZ}$.
On the other hand, the Legendre-extended Principal Hierarchy \eqref{KdV-BPH-2}
coincides with the dispersionless limit of the \textit{extended $q$-deformed KdV hierarchy},
see details in \cite{GFM2}. From \eqref{KdV-Vol coordinates} and
the identifications $x=t^{0,0}$, $\hat x=\hat t^{0,0}$,
we obtain the relationship
\begin{equation}
  \hat v_{\hat x} = \frac{1}{4v}v_x
\end{equation}
between the coordinates $v_x$ and $\hat v_{\hat x}$ on the jet space $J^\infty(M)$,
and then it can be verified directly that the dispersionless extended KdV hierarchy \eqref{KdV-BPH-1}
and the dispersionless extended $q$-deformed KdV hierarchy \eqref{KdV-BPH-2}
are related by the linear reciprocal transformation
\begin{equation}\label{KdV reciproocal}
  \hat t^{1,p} = t^{1,p},\qquad \hat t^{0,q}=t^{B,q},\qquad
  \hat t^{\hat B,q} = t^{0,q}
\end{equation}
for $p\geq 0$ and $q\in\bbZ$, which confirms the assertion of Theorem \ref{thm:linear reciprocal}.

The tau structures $\{\Omg_{i,p;j,q}\}, \{\hat\Omg_{i,p;j,q}\}$
of the integrable hierarchies \eqref{KdV-BPH-1}--\eqref{KdV-BPH-2} are given by
\begin{align*}
  \Omg_{1,p;1,q} &=\,
  \Omg_{1,p;0,q} = \Omg_{0,p;1,q} = \Omg_{0,p;0,q}
=
  \begin{cases}
    \frac{1}{\Gamma(p+1)\Gamma(q+1)}\frac{v^{p+q+1}}{p+q+1}, & \text{if $p,q\geq 0$}, \\
    0, & \text{if $p<0$ or $q<0$},
  \end{cases} \\
\Omg_{B,p;0,q} &=\, \Omg_{B,p;1,q}
=
  \frac{\Gamma(\frac12)}{\Gamma(p+\frac12)\Gamma(q+1)}
  \frac{v^{p+q+\frac12}}{2p+2q+1}, \\
\Omg_{B,p;B,q} &=\,
  \begin{cases}
    \frac14 \frac{\Gamma^2(\frac12)}{\Gamma(p+\frac12)\Gamma(q+\frac12)}
    \frac{v^{p+q}}{p+q}, & \text{if $p+q\neq 0$}, \\
    \frac{(-1)^p}{4}\log v, & \text{if $p+q=0$},
  \end{cases}
\\
  \hat\Omg_{1,p;1,q} &=\,
  \hat\Omg_{1,p;\hat B,q} = \hat\Omg_{\hat B,p;1,q} = \hat\Omg_{\hat B,p;\hat B,q}
=
  \begin{cases}
    \frac{1}{\Gamma(p+1)\Gamma(q+1)}\frac{\hat v^{2p+2q+2}}{p+q+1}, & \text{if $p,q\geq 0$}, \\
    0, & \text{if $p<0$ or $q<0$},
  \end{cases} \\
\hat\Omg_{0,p;\hat B,q} &=\, \hat\Omg_{0,p;1,q}
=
  \frac{\Gamma(\frac12)}{\Gamma(p+\frac12)\Gamma(q+1)}
  \frac{\hat v^{2p+2q+1}}{2p+2q+1}, \\
\hat\Omg_{0,p;0,q} &=\,
  \begin{cases}
    \frac14 \frac{\Gamma^2(\frac12)}{\Gamma(p+\frac12)\Gamma(q+\frac12)}
    \frac{\hat v^{2p+2q}}{p+q}, & \text{if $p+q\neq 0$}, \\
    \frac{(-1)^p}{2}\log \hat v, & \text{if $p+q=0$},
  \end{cases}
\end{align*}
for $p,q\in\bbZ$, and the associated complete data $\widetilde{\mu}_B, \widetilde{R}_B$
and $\hat{\widetilde{\mu}}_{\hat B},
\hat{\widetilde{R}}_{\hat B}$ are given by
\begin{equation*}
\begin{split}
  \widetilde{\mu}_B
= \left(
  \begin{array}{cc|c|cc}
    0 &   &   &   &   \\
      & -\frac12 &   &   &   \\
    \hline
      &   & 0 &   &   \\
    \hline
      &   &   & 0 &   \\
      &   &   &   &  \frac12
  \end{array}
  \right),\quad
\widetilde{R}_B=\widetilde{R}_{B;1} =
\left(
  \begin{array}{cc|c|cc}
      &   &   & 0 & 0  \\
      &   &   & 0 & 0  \\
    \hline
      &   & 0 &   &   \\
    \hline
     0& 0 &   &   &   \\
     0& \frac14 &   &   &
  \end{array}
  \right),
\\
  \hat{\widetilde{\mu}}_{\hat B}
= \left(
  \begin{array}{cc|c|cc}
    -\frac12 &   &   &   &   \\
      & 0 &   &   &   \\
    \hline
      &   & 0 &   &   \\
    \hline
      &   &   & \frac12 &   \\
      &   &   &   &  0
  \end{array}
  \right),\quad
\hat{\widetilde{R}}_{\hat B}=\hat{\widetilde{R}}_{\hat B;1} =
\left(
  \begin{array}{cc|c|cc}
      &   &   & 0 & 0  \\
      &   &   & 0 & 0  \\
    \hline
      &   & 0 &   &   \\
    \hline
     \frac14& 0 &   &   &   \\
     0& 0 &   &   &
  \end{array}
  \right).
\end{split}
\end{equation*}
We also note that
\[
  \widetilde{\eta}_B = \hat{\widetilde{\eta}}_{\hat B}
=
  \left(
  \begin{array}{cc|c|cc}
      &   &   & 1 & 0  \\
      &   &   & 0 & 1  \\
    \hline
      &   & 1 &   &   \\
    \hline
     1& 0 &   &   &   \\
     0& 1 &   &   &
  \end{array}
  \right).
\]
The Legendre-extended Virasoro operators $\left\{L_m^{(B)}\right\}_{m\geq -1}$
of $(M, B)$ given by the above complete data and \eqref{Vir-B1} have the expressions
\begin{align}
  L^{(B)}_{-1}
=&\,
  \sum_{p\in\bbZ}\left(t^{B,p+1}\pp{t^{B,p}}+t^{0,p+1}\pp{t^{0,p}}\right)
 +\sum_{p\geq 0}t^{1,p+1}\pp{t^{1,p}}
 +\frac12 t^{1,0}t^{1,0}, \\
L_0^{(B)}
=&\,
  \sum_{p\in\bbZ}pt^{B,p}\pp{t^{B,p}}
  +\sum_{p\in\bbZ}
  \left(
    p+\frac12
  \right)t^{0,p}\pp{t^{0,p}}
  +\sum_{p\geq 0}
  \left(
    p+\frac12
  \right)t^{1,p}\pp{t^{1,p}}
\notag\\
&
+\frac18\sum_{p\in\bbZ}
(-1)^p t^{B,p}t^{B,-p}+\frac{1}{16},\\
L_m^{(B)}
=&\,
\frac12\sum_{p=0}^{m-1}
\frac{(2p+1)!!(2m-2p-1)!!}{2^{m+1}}
\frac{\p^2}{\p t^{1,p}\p t^{1,m-1-p}}\notag\\
&
+
\sum_{p\geq 1}
\frac{(p+m)!}{(p-1)!}
\left(
  t^{B,p}\pp{t^{B,p+m}}
  +(-1)^{m+1}
  t^{B,-p-m}
  \pp{t^{B,-p}}
\right)\notag\\
&
+
  \sum_{p\in\bbZ}
    \frac{\Gamma(p+m+\frac 32)}{\Gamma(p+\frac12)} t^{0,p}\pp{t^{0,p+m}}
+\sum_{p\geq 0}
  \frac{\Gamma(p+m+\frac 32)}{\Gamma(p+\frac12)}
  t^{1,p}
  \pp{t^{1,p+m}}
\notag\\
&
+
  \frac14\sum_{p\geq 1}
  (-1)^{p+m}
  \frac{(p+m)!}{(p-1)!}
  \left(
      \sum_{k=p}^{p+m}
      \frac{1}{k}
  \right)
  t^{B,p}t^{B,-m-p}\notag\\
&
+
  \frac18\sum_{p=0}^{m}
  (-1)^m
  p!(m-p)!
  t^{B,-p}t^{B,p-m},\quad m\geq 1,
\end{align}
which are the Virasoro operators of the (dispersionless) extended KdV hierarchy \eqref{KdV-BPH-1}.
On the other hand,
the Legendre-extended Virasoro operators
$\bigl\{\hat L_m^{(\hat B)}\bigr\}_{m\geq -1}$ of $(\hat M,\hat B)$
are obtained from $\bigl\{L_m^{(B)}\bigr\}_{m\geq -1}$ by
the replacement $t^{0,p}\mapsto \hat t^{\hat B,p}$,
$t^{1,p}\mapsto \hat t^{1,p}$ and $t^{B,p}\mapsto \hat t^{0,p}$.
These operators are the Virasoro operators of the (dispersionless) extended $q$-deformed KdV hierarchy \eqref{KdV-BPH-2},
which are given in \cite{GFM1}.

\subsection{Solutions to their loop equations}

As it is shown in \cite{normal-form},
the loop equation \eqref{loop equation-2308} of the Frobenius manifold
$M$ with potential \eqref{KdV FM potential} reads
\begin{align}
&
  \sum_{s\geq 0}
    \pfrac{\Delta\mcalF}{v^{(s)}}
    \left[
      \p_x^s\frac{1}{v-\lmd}
     +\sum_{k=1}^{s}
        {s\choose k}
        \left(\p_x^{k-1}\frac{1}{\sqrt{v-\lmd}}\right)
        \left(\p_x^{s+1-k}\frac{1}{\sqrt{v-\lmd}}\right)
    \right]
\notag \\
&\quad=
  \frac{\veps^2}{2}
  \sum_{k,\ell\geq 0}
    \left(
      \frac{\p^2\Delta\mcalF}{\p v^{(k)}\p v^{(\ell)}}
     +\pfrac{\Delta\mcalF}{v^{(k)}}
      \pfrac{\Delta\mcalF}{v^{(\ell)}}
    \right)
    \left(
      \p_x^{k+1}\frac{1}{\sqrt{v-\lmd}}
    \right)
    \left(
      \p_x^{\ell+1}\frac{1}{\sqrt{v-\lmd}}
    \right)
\notag \\
&\qquad
  -\frac{\veps^2}{16}
   \sum_{k\geq 0}
     \pfrac{\Delta\mcalF}{v^{(k)}}
     \left(
       \p_x^{k+2}
       \frac{1}{(v-\lmd)^2}
     \right)
   -\frac{1}{16}\frac{1}{(v-\lmd)^2},
\end{align}
where $v^{(s)}:=\p_x^sv$ for $s\geq 0$, and the unknown function has the form
\[\Delta\mcalF=\sum_{g\geq 1}\veps^{2g-2}\mcalF_g(v,v_x,...,v^{(3g-2)}).\]
The above loop equation is required to hold true identically with respect to the
parameter $\lmd$. The first three free energies of are given by
\begin{align*}
  \mcalF_1 &=\, \frac{1}{24}\log v_x, \qquad
  \mcalF_2 = \frac{v^{(4)}}{1152 v_x^2} - \frac{7v^{(3)}v_{xx}}{1920 v_x^3}
               +\frac{v_{xx}^3}{360 v_x^4}, \\
  \mcalF_3 &=\,
  \frac{v^{(7)}}{82944 v_x^3}
  -\frac{7 v^{(6)} v_{xx}}{46080 v_x^4}
  -\frac{53 v^{(5)} v^{(3)}}{161280 v_x^4}
  +\frac{353 v^{(5)} v_{xx}^2}{322560 v_x^5}
  -\frac{103 \left(v^{(4)}\right)^2}{483840 v_x^4}
  +\frac{1273 v^{(4)}v^{(3)} v_{xx}}{322560 v_x^5} \\
&\quad
  -\frac{83 v^{(4)}v_{xx}^3}{15120 v_x^6}
  +\frac{59 \left(v^{(3)}\right)^3}{64512 v_x^5}
  -\frac{83 \left(v^{(3)}\right)^2 v_{xx}^2}{7168 v_x^6}
  +\frac{59 v^{(3)} v_{xx}^4}{3024 v_x^7}
  -\frac{5 v_{xx}^6}{648 v_x^8}.
\end{align*}
On the other hand, as it is shown in \cite{GFM2,GFM1},
the loop equation of the generalized Frobenius manifold $\hat M$ with potential \eqref{v^4/12} has the form
 \begin{align}
&
   \sum_{s\geq 0}
   \pfrac{\Delta\hat\mcalF}{\hat v^{(s)}}
   \p_{\hat x}^s\frac{1}{2\hat v(\hat v^2-\lmd)}
+
  \sum_{s\geq 1}
   \pfrac{\Delta\hat\mcalF}{\hat v^{(s)}}
   s\p_{\hat x}^s
   \left[
     \frac{1}{2\lmd}
     \left(
       \frac{1}{\sqrt{\hat v^2-\lmd}}-\frac 1{\hat v}
     \right)
   \right]  \notag \\
&\qquad
  +\sum_{s\geq 1}
   \pfrac{\Delta\hat\mcalF}{\hat v^{(s)}}
   \sum_{k=1}^{s}
   {s\choose k}
   \left[
     \p_{\hat x}^{k-1}
     \frac{1}{2\lmd}
     \left(
       \frac{\hat v}{\sqrt{\hat v^2-\lmd}}-1
     \right)
   \right]
   \left(
     \p_{\hat x}^{s+1-k}
     \frac{1}{\sqrt{\hat v^2-\lmd}}
   \right)  \notag \\
&\quad=
  \frac12\veps^2
  \sum_{k,\ell\geq 0}
  \left(
    \pfrac{\Delta\hat\mcalF}{\hat v^{(k)}}
    \pfrac{\Delta\hat\mcalF}{\hat v^{(\ell)}}
   +\frac{\p^2\Delta\hat\mcalF}{\p \hat v^{(k)}\p \hat v^{(\ell)}}
  \right)
  \left(
    \p_{\hat x}^{k+1}
    \frac{1}{\sqrt{\hat v^2-\lmd}}
  \right)
  \left(
    \p_{\hat x}^{\ell+1}
    \frac{1}{\sqrt{\hat v^2-\lmd}}
  \right)  \notag \\
&\qquad+
   \frac12\veps^2
     \sum_{k\geq 0}
     \pfrac{\Delta\hat\mcalF}{\hat v^{(k)}}
     \p_{\hat x}^{k+1}
     \frac{\hat v^2\hat v_{\hat x}}{(\hat v^2-\lmd)^3}
 -
   \frac{1}{16}\frac{1}{(\hat v^2-\lmd)^2},\label{Loop eqn for v^4/12}
\end{align}
where $\hat v^{(s)}=\p_{\hat x}^s\hat v$, and the unknown function has the form
\[\Delta\hat\mcalF=\sum_{g\geq 1}\veps^{2g-2}\hat\mcalF_g(\hat v,\hat v_{\hat x},...,\hat v^{(3g-2)}).\]
The first three free energies are given by
\begin{align*}
\hat\mcalF_1 &=\, \frac{1}{24}\log \hat v_{\hat x}+\frac{1}{12}\log \hat v, \\
\hat\mcalF_2 &=\,
 \frac{ \hat v^{(4)} \hat v}{576 \hat v_{\hat x}^2}
-\frac{7 \hat v^{(3)} \hat v_{\hat x\hat x}\hat v}{960 \hat v_{\hat x}^3}
+\frac{37 \hat v^{(3)}}{2880 \hat v_{\hat x}}
+\frac{\hat v_{\hat x\hat x}^3 \hat v}{180 \hat v_{\hat x}^4}
-\frac{11 \hat v_{\hat x\hat x}^2}{960 \hat v_{\hat x}^2}
+\frac{\hat v_{\hat x\hat x}}{120 \hat v}
-\frac{\hat v_{\hat x}^2}{120 \hat v^2},
\\
\hat\mcalF_3 &=\,
  \frac{    \hat{v}^{(7)}\hat{v}^2 }{20736 \hat{v}_{\hat{x}}^3}
 -\frac{7   \hat{v}^{(6)}\hat{v}_{\hat{x}\hat{x}}\hat{v}^2}{11520 \hat{v}_{\hat{x}}^4}
 +\frac{91  \hat{v}^{(6)}\hat{v} }{103680 \hat{v}_{\hat{x}}^2}
 -\frac{53  \hat{v}^{(5)}\hat{v}^{(3)}\hat{v}^2}{40320 \hat{v}_{\hat{x}}^4}
 +\frac{353 \hat{v}^{(5)}\hat{v}_{\hat{x}\hat{x}}^2 \hat{v}^2}{80640 \hat{v}_{\hat{x}}^5} \notag \\
&\quad
 -\frac{419 \hat{v}^{(5)}\hat{v}_{\hat{x}\hat{x}}\hat{v} }{60480 \hat{v}_{\hat{x}}^3}
 +\frac{913 \hat{v}^{(5)}}{241920 \hat{v}}
 -\frac{103 (\hat{v}^{(4)})^2 \hat{v}^2 }{120960 \hat{v}_{\hat{x}}^4}
 +\frac{1273  \hat{v}^{(4)} \hat{v}^{(3)}\hat{v}_{\hat{x}\hat{x}}\hat{v}^2 }{80640 \hat{v}_{\hat{x}}^5}\notag\\
 &\quad
 -\frac{9169 \hat{v}^{(4)} \hat{v}^{(3)} \hat{v} }{725760 \hat{v}_{\hat{x}}^3}
 -\frac{83  \hat{v}^{(4)} \hat{v}_{\hat{x}\hat{x}}^3 \hat{v}^2}{3780 \hat{v}_{\hat{x}}^6}
 +\frac{545 \hat{v}^{(4)} \hat{v}_{\hat{x}\hat{x}}^2 \hat{v}}{16128 \hat{v}_{\hat{x}}^4}
 -\frac{3727\hat{v}^{(4)}\hat{v}_{\hat{x}\hat{x}}}{241920 \hat{v}_{\hat{x}}^2}
 +\frac{\hat{v}^{(4)}}{1512 \hat{v}}\notag\\
 &\quad
 +\frac{59 (\hat{v}^{(3)})^3 \hat{v}^2 }{16128 \hat{v}_{\hat{x}}^5}
 -\frac{83  (\hat{v}^{(3)})^2 \hat{v}_{\hat{x}\hat{x}}^2 \hat{v}^2 }{1792 \hat{v}_{\hat{x}}^6}
 +\frac{97  (\hat{v}^{(3)})^2\hat{v}_{\hat{x}\hat{x}}\hat{v}}{2016 \hat{v}_{\hat{x}}^4}
 -\frac{1669 (\hat{v}^{(3)})^2}{145152 \hat{v}_{\hat{x}}^2}\notag\\
 &\quad
 +\frac{59 \hat{v}^{(3)}  \hat{v}_{\hat{x}\hat{x}}^4 \hat{v}^2}{756 \hat{v}_{\hat{x}}^7}
 -\frac{5555  \hat{v}^{(3)}\hat{v}_{\hat{x}\hat{x}}^3\hat{v}  }{48384 \hat{v}_{\hat{x}}^5}
 +\frac{325  \hat{v}^{(3)} \hat{v}_{\hat{x}\hat{x}}^2}{6912 \hat{v}_{\hat{x}}^3}
 -\frac{\hat{v}^{(3)}\hat{v}_{\hat{x}}}{378 \hat{v}^2}
 -\frac{5 \hat{v}_{\hat{x}\hat{x}}^6\hat{v}^2 }{162 \hat{v}_{\hat{x}}^8}\notag\\
 &\quad
 +\frac{13 \hat{v}_{\hat{x}\hat{x}}^5\hat{v} }{252 \hat{v}_{\hat{x}}^6}
 -\frac{193 \hat{v}_{\hat{x}\hat{x}}^4}{8064 \hat{v}_{\hat{x}}^4}
 -\frac{\hat{v}_{\hat{x}\hat{x}}^2}{504 \hat{v}^2}
 +\frac{\hat{v}_{\hat{x}\hat{x}}\hat{v}_{\hat{x}}^2 }{126 \hat{v}^3}
 -\frac{\hat{v}_{\hat{x}}^4}{252 \hat{v}^4}.
\end{align*}
From the flows
\[\frac{\p v}{\p t^{0,0}}=v_x,\quad \pfrac{\hat v}{\hat t^{\hat B,0}}=2\hat v\hat v_{\hat x}\]
and \eqref{KdV-Vol coordinates}, \eqref{KdV reciproocal},
we obtain the following relationships between the coordinates
$\left\{v,v_x,v_{xx},...,v^{(4)}\right\}$ and $\left\{\hat v, \hat v_{\hat x}, \hat v_{\hat x\hat x},...,\hat v^{(4)}\right\}$
on the jet space $J^\infty(M)$:
\begin{align*}
  &v=\,\hat v^2, \quad v_x = 4\hat v^2 \hat v_{\hat x},
\quad
  v_{xx} = 8 \hat v^3 \hat v_{\hat x\hat x}+24 \hat v^2 \hat v_{\hat x}^2, \\
  &v^{(3)} =\, 16 \hat v^4 \hat v^{(3)}
     +192 \hat v^3 \hat v_{\hat x\hat x}\hat v_{\hat x}
     +192 \hat v^2 \hat v_{\hat x}^3, \\
 & v^{(4)} =\,
    32\hat v^5\hat v^{(4)} +640\hat v^4\hat v^{(3)} \hat v_{\hat x}
   +480\hat v^2\hat v_{\hat x\hat x}^2
   +3840\hat v^3\hat v_{\hat x\hat x}\hat v_{\hat x}^2
   +1920\hat v^2\hat v_{\hat x}^4.
\end{align*}
By using these relations, we can verify that
\[
  \mcalF_1 = \hat\mcalF_1 + \text{const.},\qquad
  \mcalF_2 = \hat\mcalF_2.
\]
The relation $\mcalF_3 = \hat\mcalF_3$ can also be verified in the same way.

\subsection{Linear reciprocal transformation between hierarchies}

The solutions $\Delta\mcalF$, $\Delta\hat\mcalF$ to the loop equations of
  $M$ and $\hat M$ yield the topological deformations of the
  Legendre-extended Principal Hierarchies.
  The topological deformations of the flows $\left\{\pp{t^{1,p}}\right\}_{p\geq 0}$
  form the KdV hierarchy which controls the 2D topological gravity as we learn from Witten's conjecture and its proof by Kontsevich
  \cite{Kontsevich,Witten},
  and the topological deformations of the Legendre flows $\left\{\pp{t^{B,p}}\right\}_{p\in\bbZ}$
  extend the KdV hierarchy.
  Precisely speaking, if $v$ satisfies the Legendre-extended Principal Hierarchy \eqref{KdV-BPH-1} of $(M,B)$,
  then
  \begin{align}
      U&=\, \veps^2\frac{\p^2\log\tau}{\p x\p t^{1,0}}
      =
        v+\veps^2(\Delta\mcalF)_{xx}
      \notag \\
    &=\,
      v+
      \left(
        \frac{v^{(3)}}{24 v_x}
       -\frac{v_{xx}^2}{24 v_x^2}
      \right)\veps^2
    +
      \left(
        \frac{v^{(6)}}{1152 v_x^2}
       -\frac{41 v^{(5)} v_{xx}}{5760 v_x^3}
       -\frac{73 v^{(4)} v^{(3)}}{5760 v_x^3} \right. \notag \\
    &\qquad
      \left.
       +\frac{17 v^{(4)} v_{xx}^2}{480 v_x^4}
       +\frac{19 \left(v^{(3)}\right)^2 v_{xx}}{384 v_x^4}
       -\frac{35 v^{(3)} v_{xx}^3}{288 v_x^5}
       +\frac{v_{xx}^5}{18 v_x^6}
      \right)\veps^4 + O(\veps^6)  \label{full-genera KdV U}
  \end{align}
  satisfies the following well-known KdV hierarchy
  \begin{equation}\label{KdV-hierarchy}
    \veps\pfrac{L}{t^{1,p}}
  =
    \frac{2^{p+\frac12}}{(2p+1)!!}
    \left[
      \left(L^{p+\frac12}\right)_+ , L
    \right],\qquad p\geq 0,
  \end{equation}
  and topological deformations of the Legendre flows $\left\{\pp{t^{B,p}}\right\}_{p\in\bbZ}$,
  here the Lax operator
  \begin{equation}\label{Lax KdV}
    L = \frac{\veps^2}{2}\p_x^2 + U.
  \end{equation}
  Note that the first few flows of the KdV hierarchy \eqref{KdV-hierarchy} reads
  \begin{align}
    \pfrac{U}{t^{1,0}} &=\, U_x, \notag\\
    \pfrac{U}{t^{1,1}} &=\, UU_x + \frac{\veps^2}{12}U^{(3)}, \\
    \pfrac{U}{t^{1,2}} &=\,
      \frac12 U^2U_x + \frac{\veps^2}{12}
      \left(
        UU^{(3)}+2U_{xx}U_x
      \right)
    +
      \frac{\veps^4}{240}U^{(5)}, \notag
  \end{align}
and it can be verified directly that $U$ satisfies
\begin{align}
  \pfrac{U}{t^{B,0}}
&=\,
  \frac{1}{2} U^{-\frac{1}{2}} U_x
+
  \left(
    -\frac{U^{(3)}}{48U^{\frac{3}{2}}}
    +\frac{U_{xx} U_x}{16U^{\frac{5}{2}}}
    -\frac{5U_x^3}{128U^{\frac{7}{2}}}
  \right)\veps^2 \notag\\
&\quad +
\left(
  \frac{{U^{(5)}}}{640U^{\frac{5}{2}}}
 -\frac{3{U^{(4)}} U_x}{256U^{\frac{7}{2}} }
 -\frac{5{U^{(3)}}U_{xx}}{256U^{\frac{7}{2}}}
 +\frac{161U^{(3)}U_x^2}{3072 U^{\frac{9}{2}}}\right.\notag\\
&\qquad\left.
 +\frac{217 U_{xx}^2 U_x}{3072U^{\frac{9}{2}}}
 -\frac{21U_{xx}U_x^3}{128U^{\frac{11}{2}}}
 +\frac{1155 U_x^5}{16384U^{\frac{13}{2}}}
\right)\veps^4 + O(\veps^6),  \label{full-genera Legendre t^B,0 flow}
\end{align}
which is the topological deformation of the Legendre flow $\pfrac{v}{t^{B,0}}=\frac12 v^{-\frac12}v_x$.

On the other hand, it is proved in \cite{GFM2} that the topological deformations of
  the Legendre-extended Principal Hierarchy \eqref{KdV-BPH-2} of $\hat M$ coincide, up to a Miura-type transformation, with the extended $q$-deformed KdV hierarchy.
More precisely, if $\hat v$ satisfies the Legendre-extended Principal Hierarchy \eqref{KdV-BPH-2} of $\hat M$, then
  \begin{align}
    \hat U&=\,
      \veps\frac{\Lmd-\Lmd^{-1}}{2\sqrt{2}\,\rmi}
      \pfrac{\log\hat\tau}{\hat t^{1,0}}
    =
      \frac{\Lmd-\Lmd^{-1}}{2\sqrt{2}\,\rmi\,\veps\p_x}
      \left(
        \hat v + \frac{\p^2\Delta\hat{\mcalF}}{\p\hat x\p\hat t^{1,0}}
      \right) \\
  &=\,
    \hat v +
    \left(
      \frac{\hat v \hat v^{(3)}}{12 \hat v_{\hat x}}
     -\frac{\hat v \hat v_{\hat x\hat x}^2}{12 \hat v_{\hat x}^2}
      \right)\veps^2
     +
     \left(
     \frac{\hat{v}^2 \hat{v}^{(6)}}{288 \hat{v}_{\hat x}^2}
    -\frac{41 \hat{v}^2 \hat{v}^{(5)} \hat{v}_{\hat x\hat x}}{1440 \hat{v}_{\hat x}^3}
    +\frac{3 \hat{v} \hat{v}^{(5)}}{160 \hat{v}_{\hat x}}
    \right.\notag\\
  &\qquad
    -\frac{73 \hat{v}^2 \hat{v}^{(4)}\hat{v}^{(3)} }{1440 \hat{v}_{\hat x}^3}
    +\frac{17 \hat{v}^2 \hat{v}^{(4)}\hat{v}_{\hat x\hat x}^2}{120 \hat{v}_{\hat x}^4}
    -\frac{13 \hat{v} \hat{v}^{(4)}\hat{v}_{\hat x\hat x}}{160 \hat{v}_{\hat x}^2}
    +\frac{19 \hat{v}^2 \left(\hat{v}^{(3)}\right)^2 \hat{v}_{\hat x\hat x}}{96 \hat{v}_{\hat x}^4}
    \notag\\
  &\qquad\left.
    -\frac{\hat{v} \left(\hat{v}^{(3)}\right)^2}{18 \hat{v}_{\hat x}^2}
    -\frac{35 \hat{v}^2 \hat{v}^{(3)} \hat{v}_{\hat x\hat x}^3}{72 \hat{v}_{\hat x}^5}
    +\frac{71 \hat{v} \hat{v}^{(3)} \hat{v}_{\hat x\hat x}^2}{288 \hat{v}_{\hat x}^3}
    +\frac{2 \hat{v}^2 \hat{v}_{\hat x\hat x}^5}{9 \hat{v}_{\hat x}^6}
    -\frac{37 \hat{v} \hat{v}_{\hat x\hat x}^4}{288 \hat{v}_{\hat x}^4}
     \right)\veps^4 + O(\veps^6) \label{full-genera qKdV hatU}
  \end{align}
satisfies the following so-called \textit{extended $q$-deformed KdV hierarchy}
\begin{align}
  \sqrt{2}\,\rmi\,\veps\pfrac{\hat L}{\hat t^{1,p}}
&=\,
  (-1)^{p+1}\frac{2^{3p+2}}{(2p+1)!!}
  \left[
    \left(
      \hat L^{p+\frac12}
    \right)_+ , \hat L
  \right],& p\geq 0,
\label{qKdV-1-241024}\\
  \sqrt{2}\,\rmi\,\veps\pfrac{\hat L}{\hat t^{0,-p}}
&=\,
  -\frac{(p-1)!}{2^{2p}}
  \left[
    \left(
      \hat L^{-p}
    \right)_- , \hat L
  \right],& p\geq 1,
\label{qKdV-2-241024}\\
  \sqrt{2}\,\rmi\,\veps\pfrac{\hat L}{\hat t^{0,p}}
&=\,
  (-1)^{p+1}\frac{2^{2p-1}}{p!}
  \left[
    \left(
      \hat L^{p}\log\hat L
    \right)_+ , \hat L
  \right],& p\geq 0, \label{qKdV-3-241024}
\end{align}
where
\begin{equation}\label{Lax qKdV}
  \hat L = \Lmd^2 + \hat U\Lmd,\qquad
  \Lmd = \rme^{\sqrt{2}\,\rmi\,\veps\p_{\hat x}},
\end{equation}
the definitions of $\hat L^{-1}$, $\log\hat L$ and more details can be found in Sect.\,4 of  \cite{GFM2}.
The first few flows of this hierarchy have the following forms:
\begin{align}
  \sqrt{2}\,\rmi\,\veps
  \pfrac{\hat U}{\hat t^{1,0}}&=
    4\hat U\left(\frac{\Lmd-1}{\Lmd+1}\hat U\right),  \label{hat t 1,0-flow}
\\
  \sqrt{2}\,\rmi\,\veps\pfrac{\hat U}{\hat t^{1,1}}&=\,
    \frac{32}{3}
    \left(
      \hat U\circ\frac{\Lmd-1}{\Lmd+1}\circ \hat U\circ \frac{\Lmd}{\Lmd+1}
    \right)
    \left[
      \left(
        \frac{1}{\Lmd+1}\hat U
      \right)^2
    \right],
\\
  \pfrac{\hat U}{\hat t^{0,0}} &=\hat U_{\hat x},\qquad
  \sqrt{2}\,\rmi\,\veps\pfrac{\hat U}{\hat t^{0,-1}}  =
  \frac14\left(\frac{1}{\hat U^+}-\frac{1}{\hat U^-}\right),
\end{align}
where $\hat U^{\pm}=\Lambda^{\pm1} \hat U$.

Based on the discussion given in Section \ref{subsection:full genera}, we arrive at the equivalence relation
  \[
    \left(
      \text{KdV Hierarchy} \atop
      \text{with extended flows $\left\{\pp{t^{B,p}}\right\}_{p\in\bbZ}$}
    \right)
    \xlongequal{\text{\eqref{KdV reciproocal}}}
    \left(
      \text{$q$-deformed KdV Hierarchy} \atop
      \text{with extended flows $\left\{\pp{\hat t^{0,p}}\right\}_{p\in\bbZ}$}
    \right)
  \]
  up to a certain Miura-type transformation.
Moreover, we have the following theorem.
\begin{thm}
If $\hat U$ satisfies the extended $q$-deformed KdV hierarchy \eqref{qKdV-1-241024}--\eqref{qKdV-3-241024},
then
\begin{align}
  U&:=\,\frac{8\Lmd}{\Lmd+1}
    \left[
    \left(
      \frac{1}{\Lmd+1}\hat U
    \right)^2
    \right] \notag
\\
  &\phantom:=\,
  \hat U^2+
  \left(
    \hat U \hat U_{\hat x\hat x}
   +\frac{1}{2} \hat U_{\hat x}^2
  \right)\veps^2
 +\left(
   \frac{1}{3} \hat U \hat U^{(4)}
  +\frac{2}{3} \hat U^{(3)} \hat U_{\hat x}
  +\frac{1}{2} \hat U_{\hat x\hat x}^2
 \right)\veps^4 + O(\veps^6)\label{U and hat U}
\end{align}
satisfies the KdV hierarchy \eqref{KdV-hierarchy}
and the topological deformations of the Legendre flows $\left\{\pp{t^{B,p}}\right\}_{p\in\bbZ}$,
under the identifications \eqref{KdV reciproocal}.
In particular, $U$ satisfies the KdV equation
\begin{equation}
  \pfrac{U}{\hat t^{1,1}}
=
  U\pfrac{U}{\hat t^{1,0}}
+\frac{\veps^2}{12}
 \frac{\p^3 U}{\left(\p \hat t^{1,0}\right)^3}.
\end{equation}
\end{thm}

\begin{proof}
  It suffices to verify that the functions $U$ and $\hat U$ defined in \eqref{full-genera KdV U} and \eqref{full-genera qKdV hatU} respectively are related by \eqref{U and hat U}.
  From Theorem \ref{thm: sol to loop eqn} we know that
  the topological deformations of the
  Legendre-extended Principal Hierarchies of $M$ and $\hat M$ share the same tau function, therefore
  \[
    U = \veps^2\frac{\p^2\log\tau}{\p x\p t^{1,0}}
    =
      \veps^2\frac{\p^2\log\hat\tau}{\p \hat t^{1,0}\p\hat t^{1,0}}.
  \]
  By the relationships between the $q$-deformed KdV hierarchy and the \textit{fractional Volterra hierarchy} (FVH) together with its
tau structure (see details in \cite{GFM2,Hodge-FVH,FVH}),
we can verify that
  \[
    \veps^2\frac{\p^2\log\hat\tau}{\p \hat t^{1,0}\p\hat t^{1,0}}
  =
    -8\mathrm{Res}
     \left(
       \Lmd\hat L^{\frac12}
     \right)
  =
    \frac{8\Lmd}{\Lmd+1}
    \left[
    \left(
      \frac{1}{\Lmd+1}\hat U
    \right)^2
    \right],
  \]
thus we arrive at a proof of the theorem.
\end{proof}

\begin{rmk}
  The relation \eqref{U and hat U} given in the above theorem can also be derived from spectral problems.
  Note that the $\hat t^{1,0}$-flow given in \eqref{qKdV-1-241024}
  is the compatibility condition of the following spectral problem and the evolution of the wave function:
  \begin{align}
    \hat L\hat\psi &=\, \hat\lmd\hat\psi, \\
    \sqrt{2}\,\rmi\,\veps\hat\psi_{\hat t^{1,0}} &=\, -4
    \left[
      \Lmd + \left(\frac{1}{1+\Lmd}\hat U\right)
    \right]\hat\psi,
  \end{align}
here $\hat L$ is the Lax operator \eqref{Lax qKdV},
$\hat\lmd$ is the spectral parameter, $\hat\psi$ is the wave function.
Denote
\[
  \lmd = -4\hat\lmd,\qquad \psi = \Lmd\hat\psi,
\]
then by using the following identity
\begin{align}
    &\,
      (\Lmd-1)
      \left[
        \left(\frac{1}{\Lmd+1}\hat U\right)^2
      \right]
    =
      \left(
        \frac{\Lmd}{\Lmd+1}\hat U
      \right)^2
     -\left(
       \frac{1}{\Lmd+1}\hat U
      \right)^2 \notag\\
    =&\,
      \left(
        \hat U - \frac{1}{\Lmd+1}\hat U
      \right)^2
     -\left(
       \frac{1}{\Lmd+1}\hat U
      \right)^2
    =
      \hat U^2 - \hat U\left(\frac{2}{\Lmd+1}\hat U\right)
    =
      \hat U\left(\frac{\Lmd-1}{\Lmd+1}\hat U\right), \label{241118-1357}
    \end{align}
we can show that
\[
  \left[
    \frac{\veps^2}{2}\p_x^2+
    \frac{8\Lmd}{\Lmd+1}
    \left(
      \frac{1}{\Lmd+1}\hat U
    \right)^2
  \right]\psi = \lmd\psi,
\]
which is the spectral problem of the KdV hierarchy \eqref{Lax KdV} with $U$ given by \eqref{U and hat U},
here $x=\hat{t}^{1,0}$.
\end{rmk}

\begin{cor}
  Suppose $U$ satisfies the KdV hierarchy \eqref{KdV-hierarchy}, then the following assertions hold true.
  \begin{enumerate}
    \item There exists a function
    \[
      W=U+\sum_{k=1}^{\infty}U_{[k]}\veps^k
    \]
    uniquely determined by the equation
    \begin{equation}\label{241025-1702}
      W-U = \rmi\,\veps\left(\sqrt{W+U}\right)_x,
    \end{equation}
    here $\{U_{[k]}\}_{k\geq 1}$ are differential polynomials of $U$.
    \item The above-mentioned function $W$ yields a symmetry $U\mapsto W$ of the KdV hierarchy.
    Moreover, this symmetry is generated by the topological deformation of the Legendre flow $\pp{t^{B,0}}$ given in
    \eqref{full-genera Legendre t^B,0 flow}, i.e.
    \begin{equation}\label{241025-1709}
      W = \exp\left(\sqrt{2}\,\rmi\,\veps\pp{t^{B,0}}\right)U.
    \end{equation}
  \end{enumerate}
\end{cor}

  \begin{proof}
It is easy to see that we can solve $U_{[k]}$ recursively from \eqref{241025-1702}, so the first assertion of the corollary holds true.
In order to prove the second assertion of the corollary, it suffices to verify that the function $W$ defined by \eqref{241025-1709} solves the equation \eqref{241025-1702}. Indeed, from \eqref{KdV reciproocal} it follows that
    \[
      \exp\left(\sqrt{2}\,\rmi\,\veps\pp{t^{B,0}}\right)
    =
      \rme^{\sqrt{2}\,\rmi\,\veps\p_{\hat x}} = \Lmd.
    \]
We denote $U^{\pm}=\Lmd^{\pm}U$, then
  \[
    \hat U = \frac{\Lmd+1}{2\sqrt{2}}\sqrt{U+U^-}
  \]
 satisfies the relation \eqref{U and hat U}, so it is a solution
of the extended $q$-deformed KdV hierarchy.
 Therefore from \eqref{hat t 1,0-flow} and \eqref{241118-1357} we arrive at
  \begin{align*}
    (\Lmd-1)U
  &=\,
    \frac{8\Lmd}{\Lmd+1}(\Lmd-1)
    \left[
      \left(
        \frac{1}{\Lmd+1}\hat U
      \right)^2
    \right]
  =
    \left(
      \frac{8\Lmd}{\Lmd+1}\circ\hat U\circ\frac{\Lmd-1}{\Lmd+1}
    \right)\hat U
  \\
  &=\,
    \frac{2\sqrt{2}\,\rmi\,\veps\Lmd}{\Lmd+1}
    \pfrac{\hat U}{\hat t^{1,0}}
  =
    \frac{2\sqrt{2}\,\rmi\,\veps\Lmd}{\Lmd+1}
    \frac{\Lmd+1}{2\sqrt{2}}
    \left(
      \sqrt{U+U^-}
    \right)_x
  =
    \rmi\,\veps
    \left(
      \sqrt{U^++U}
    \right)_x,
  \end{align*}
i.e. $W:= U^+$ solves the equation \eqref{241025-1702}.
The corollary is proved.
  \end{proof}

In fact, the solution $W$ to the equation \eqref{241025-1702} has the form
\begin{align*}
  W &=\, U + \frac{U_x}{2U^{\frac12}} \epsilon
    +
    \left(
      \frac{U_{xx}}{8 U}-\frac{U_x^2}{8 U^2}
    \right)\epsilon ^2
   +
     \left(
       \frac{U^{(3)}}{32 U^{\frac 32}}
      +\frac{25 U_x^3}{256 U^{\frac 72}}
      -\frac{U_{xx}U_x}{8 U^{\frac 52}}
     \right) \epsilon^3 \\
   &\quad +
   \left(
     \frac{U^{(4)}}{128 U^2}
    -\frac{13 U^{(3)} U_x}{256 U^3}
    -\frac{9 U_{xx}^2}{256 U^3}
    +\frac{99 U_{xx} U_x^2}{512 U^4}
    -\frac{15 U_x^4}{128 U^5}
   \right)\epsilon ^4 + O(\veps^5),
\end{align*}
where $\epsilon = \sqrt{2}\,\rmi\,\veps$.

\section{Example: the Toda and the Ablowitz-Ladik hierarchies}

\subsection{Generalized Frobenius manifolds and their complete data}

Consider the 2-dimensional generalized Frobenius manifold $M$ with potential
\begin{equation} \label{AL-GFM-potential}
  F=\frac12 v^2u+v\rme^u+\frac12 v^2\log v,
\end{equation}
where $(v,u):=(v^1,v^2)$ are the flat coordinates with respect to the flat metric
\[
  \eta = \td v\otimes\td u + \td u\otimes \td v.
\]
The unity $e$, Euler vector field $E$ and monodromy data $\mu, R$ of $M$ are given by
\begin{equation} \label{e, E of AL}
  e= \frac{v\p_v-\p_u}{v-\rme^u},\quad
  E=v\p_v+\p_u,\quad
  \mu = \begin{pmatrix}
          -\frac12 & 0 \\
          0 & \frac12
        \end{pmatrix},
  \quad
  R=R_1=\begin{pmatrix}
          0 & 0 \\
          2 & 0
        \end{pmatrix}.
\end{equation}
This generalized Frobenius manifold and its relation to the \textit{Ablowitz-Ladik hierarchy}
are presented in \cite{Brini,AL-triham,GFM1,ExtendedAL}, and it is shown in \cite{GFM1,ExtendedAL} that the flows of the Ablowitz-Ladik hierarchy are contained in the topological deformation of the Principal Hierarchy of $M$.

Consider the vector field
\begin{equation} \label{AL-GFM-legendre field}
  B=\p_v
\end{equation}
on $M$,
which is an invertible quasi-homogeneous Legendre field such that
\[
  \mu_B=\mu_1=-\frac12.
\]
It is clear that the flows $\pp{t^{B,p}}$ in the Legendre-extended Principal Hierarchy
coincide with the flows $\pp{t^{1,p}}$ for $p\in\bbZ$, and the
Legendre-extended calibrations $\xi_{B,p}$, the Hamiltonian densities $\theta_{B,p}$
and the 2-point functions $\Omg_{B,p;i,q}$ can be chosen as
\begin{equation}\label{240701-1101}
  \xi_{B,p}=\xi_{1,p},\qquad
  \theta_{B,p}=\theta_{1,p},\qquad
  \Omg_{B,p;i,q}=\Omg_{1,p;i,q}
\end{equation}
for $p,q\in\bbZ$ and $i\in\{0,B,1,2\}$.
As shown in \cite{AL-Legendre},
this Legendre field $B$ transforms $M$ to the Frobenius manifold $\hat M$ with potential
\begin{equation}\label{F-Toda}
  \hat F = \frac12\hat v^2\hat u+\rme^{\hat u},
\end{equation}
where $(\hat v,\hat u):=(\hat v^1, \hat v^2)$ are the flat coordinates with respect to the new flat matric
\[
  \hat \eta = \td\hat v\otimes\td\hat u+\td\hat u\otimes\td\hat v,
\]
and they are related with $(v,u)$ by
\begin{equation}\label{AL-Toda coordinates}
     \hat v = v+\rme^u, \quad
    \hat u = \log v + u.
 \end{equation}
The relation of the Frobenius manifold $\hat M$ with the \textit{extended Toda hierarchy}
and the Gromov--Witten invariants of $\bbC\bbP^1$ are studied in
\cite{Extended Toda-1,normal-form,Extended Toda-2,Getzler,CP1}.
In the coordinates $\hat v, \hat u$, the unit vector field $e$ and the Euler vector field $E$ given in \eqref{e, E of AL} have the form
\begin{equation}
  e = \p_{\hat v},\qquad
  E = \hat v\p_{\hat v} + 2\p_{\hat u}.
\end{equation}
Note that the unit vector field $e=\p_{\hat v}$ is flat with respect to the new metric $\hat\eta$,
therefore the $\pp{\hat t^{0,p}}$-flows coincide with the $\pp{\hat t^{1,p}}$-flows in the Principal Hierarchy of $\hat M$.
From \eqref{hat v-afa2} and \eqref{e, E of AL} it follows that the inversion $\hat B:=B^{-1}$ has the form
\begin{align}\label{Toda Legendre field}
  \hat B &=\, \hat B\cdot e
         = \frac{v(\hat B\cdot\p_v) - (\hat B\cdot \p_u)}{v-\rme^u} =
    \frac{v\p_{\hat v}-\p_{\hat u}}{v-\rme^u}.
\end{align}
According to \eqref{240609-2127-1}, \eqref{240701-1058}, \eqref{240701-1059} and \eqref{240701-1101},
the Hamiltonian densities $\hat\theta_{\hat B,p}$ of the flows $\pp{\hat t^{\hat B,p}}$
of the Legendre-extended Principal Hierarchy of $(\hat M,\hat B)$ have the form
\begin{equation}
  \hat\theta_{\hat B,p} = \hat\Omg_{\hat B,p;0,0}
  =\Omg_{0,p;B,0} = \Omg_{0,p;1,0}
  =\pair{\xi_{0,p+1}}{\xi_{1,0}}
  =\pfrac{\theta_{0,p+1}}{v},\quad p\in\bbZ,
\end{equation}
where the first few terms and explicit formulas of $\theta_{0,p}$ of $M$
can be found in \cite{GFM1} and \cite{ExtendedAL} respectively.
For example, we have
\begin{align}
 & \hat\theta_{\hat B,-3} =
      -\frac{2 \left(v^2+4 v \rme^u+\rme^{2 u}\right)}{\left(v-\rme^u\right)^5},\quad
  \hat\theta_{\hat B,-2} = \frac{v+\rme^u}{\left(v-\rme^u\right)^3}, \notag\\
 & \hat\theta_{\hat B,-1}=
      -\frac{1}{v-\rme^u}, \notag\quad
  \hat\theta_{\hat B,0} = u,\quad
  \hat\theta_{\hat B,1} = v(u+1) +\rme^u (u-1), \\
&  \hat\theta_{\hat B,2}=
      \frac{1}{4} \left( v^2(2 u+3)+4 v\rme^u (2 u-1)+\rme^{2 u} (2 u-3)\right), \notag\\
  &\hat\theta_{\hat B,3} =
      \frac{1}{36} \left(v^3(6 u+11)+9 v^2\rme^u (6 u-1)+9 v \rme^{2 u} (6 u-7)+\rme^{3 u} (6 u-11)\right),\notag
\end{align}
here $v,u$ and $\hat v,\hat u$ are related by \eqref{AL-Toda coordinates}.
The Hamiltonian densities $\hat\theta_{0,p}=\hat\theta_{1,p}$ and $\hat\theta_{2,p}$
of the Frobenius manifold $\hat M$ with potential \eqref{F-Toda} are given in \cite{normal-form,Extended Toda-2} as follows:
\begin{align}
  \hat\theta_1(z)&=
    \sum_{p\geq 0}\theta_{1,p}z^p
  =
    -2\rme^{z\hat v}
    \left(
      \rmK_0(2z\rme^{\frac12\hat u})
     +(\log z+\gamma)\rmI_0(2z\rme^{\frac12\hat u})
    \right), \\
  \hat\theta_2(z)&=
    \sum_{p\geq 0}\theta_{2,p}z^p
  =
    \frac 1z
    \left(
      \rme^{z\hat v}\rmI_0(2z\rme^{\frac12\hat u})-1
    \right),
\end{align}
where $\rmI_0$, $\rmK_0$ are the modified Bessel functions, and $\gamma=0.577\cdots$
is the Euler's constant.
Then from \eqref{240701-1433} and \eqref{AL-Toda coordinates} we obtain
\begin{align*}
  \pfrac{\theta_1(z)}{u}
&=
  2zv^{\frac12}\rme^u
  \rme^{z(v+\rme^u)}
  \left(
    \rmK_1(2zv^{\frac12}\rme^u)
   -(\log z+\gamma)\rmI_1(2zv^{\frac12}\rme^u)
  \right),
\\
  \pfrac{\theta_1(z)}{v}
&=
  -2z\rme^{z(v+\rme^u)}
  \left(
    \rmK_0(2zv^{\frac12}\rme^u)
   +(\log z+\gamma)\rmI_0(2zv^{\frac12}\rme^u)
  \right),
\\
  \pfrac{\theta_2(z)}{u}
&=
  v^{\frac12}\rme^{\frac u2}
  \rme^{z(v+\rme^u)}
  \rmI_1(2zv^{\frac12}\rme^u),
\\
  \pfrac{\theta_2(z)}{v}
&=
  \rme^{z(v+\rme^u)}
  \rmI_0(2zv^{\frac12}\rme^u),
\end{align*}
which provide an alternative formulation of the Hamiltonian densities $\{\theta_{i,p}\}_{i=1,2,\,p\geq 0}$
of the dispersionless extended Ablowitz-Ladik hierarchy \cite{Brini,ExtendedAL}.

Now let us compute the complete data $\widetilde{\mu}_B, \widetilde{R}_B$ of $(M,B)$.
It is not difficult to see that some of the 2-point functions of $(M,B)$ are given by
\begin{align*}
  \Omg_{0,0;0,0} &=\, \theta_{0,0} = u-\log(v-\rme^u), \\
  \Omg_{0,0;B,0} &=\, \theta_{B,0} = \theta_{1,0} = u, \\
  \Omg_{B,0;B,0} &=\, \Omg_{1,0;1,0} = u+\log v,
\end{align*}
therefore $\mcalL_E\Omg_{0,0;0,0}=0$, $
  \mcalL_E\Omg_{0,0;B,0}=1$ and
  $\mcalL_E\Omg_{B,0;B,0}=2$.
Then from \eqref{p1Omg} and other basic properties of the complete data it follows that
\[
  \widetilde{\mu}_B
=
  \left(
    \begin{array}{cccccc}
      -\frac12 &&&&&\\
      &-\frac12&&&&\\
      &&-\frac12&&&\\
      &&&\frac12&&\\
      &&&&\frac12&\\
      &&&&&\frac12
    \end{array}
  \right),
\quad
  \widetilde{R}_B = \widetilde{R}_{B;1}
=
    \left(
    \begin{array}{cc|cc|cc}
      0 &0&&&&\\
      0&0&&&&\\
      \hline
      0&0&0&0&&\\
      1&2&2&0&&\\
      \hline
      0&1&1&0&0&0\\
      1&2&2&0&0&0
    \end{array}
  \right).
\]
On the other hand, the complete data $\hat{\widetilde{\mu}}_{\hat B},
\hat{\widetilde{R}}_{\hat B}$ of $(\hat M,\hat B)$ can be obtained from the above $\widetilde{\mu}_B, \widetilde{R}_B$
by using Theorem \ref{thm:tau correspendence}.

\subsection{Solutions to the loop equations}
In this subsection, we compare the solutions to the loop equations of $\hat M$ and $M$.
The loop equation of the Frobenius manifold $\hat M$ can be found in Sect.\,5 of \cite{Extended Toda-2}
(see also in Example 3.10.27 of \cite{normal-form}).
The genus-one solution $\hat\mcalF_1$ to this loop equation has the expression
\begin{equation}\label{Toda F1}
  \hat\mcalF_1 =
  \frac{1}{24}\log\left(\hat v^2_{\hat x}-\rme^{\hat u}\hat u^2_{\hat x}\right)-\frac{1}{24}\hat u.
\end{equation}
Under the linear reciprocal transformation \eqref{reciprocal}--\eqref{240701-1622} and \eqref{AL-Toda coordinates},
the relation between jet variables $\hat v_{\hat x}, \hat u_{\hat x}$ and $v_x, u_x$ are given by
\begin{align*}
  \hat v_{\hat x} &= \frac 1v\left((v+\rme^u)v_x + 2v\rme^uu_x\right), \\
  \hat u_{\hat x} &= \frac 1v\left(2v_x+(v+\rme^u)u_x\right).
\end{align*}
Then $\hat\mcalF_1$ given in \eqref{Toda F1} can be rewritten as
\begin{equation}
  \hat\mcalF_1
=
  \frac{1}{24}\log\left(v_x^2-v\rme^uu_x^2\right)
 +\frac{1}{12}\log(v-\rme^u)
 -\frac{1}{8}\log v
 -\frac{1}{24}u,
\end{equation}
which coincides with the genus-one solution $\mcalF_1$ to the loop equation of generalized Frobenius manifold $M$,
see Example 11.2 of \cite{GFM1}.

As it is shown in \cite{Extended Toda-2},
the topological deformation of the $\pp{\hat t^{i,p}}$-flows $(i=1,2)$
of $\hat M$ yields the so-called extended Toda hierarchy \cite{Extended Toda-1} up to a certain Miura-type transformation,
and the additional Legendre flows $\pp{\hat t^{\hat B,p}}$ give a further extension of the extended Toda hierarchy.
On the other hand, it is proved in \cite{ExtendedAL} that the topological deformation of the Legendre-extended Principal Hierarchy of $(M,B)$
coincides with the so-called extended Ablowitz-Ladik hierarchy up to a certain Miura-type transformation.
Based on the discussion given in Section \ref{subsection:full genera}
and the fact that Legendre flows $\pp{t^{B,p}}=\pp{t^{1,p}}$,
we know that the above-mentioned extension of the extended Toda hierarchy
and the extended Ablowitz-Ladik hierarchy are related by the linear reciprocal transformation \eqref{reciprocal}
up to a certain Miura-type transformation.

\subsection{The extended Toda hierarchy}
In this subsection, let us recall the definition and some basic properties of the extended Toda hierarchy.
For more details, see \cite{Extended Toda-1, Extended Toda-2}.
Introduce the Lax operator
\begin{equation} \label{Toda Lax Operator}
  \hat L = \hat\Lmd+\hat V+\rme^{\hat U}\hat\Lmd^{-1},
\end{equation}
where $\hat V, \hat U$ are unknown functions, and the shift operator
\begin{equation}
  \hat\Lmd=\rme^{\veps\p_{\hat x}}.
\end{equation}
The \textit{extended Toda hierarchy} consists of the following flows
\begin{equation}\label{Extended Tdoa Hierarchy}
  \veps\pfrac{\hat L}{\hat t^{\afa,p}}=[A_{\afa,p}, \hat L],\quad \afa=1,2,\, p\geq 0,
\end{equation}
where
\begin{align*}
  A_{1,p} = \frac{2}{p!}
  \left(
    \hat L^p
    (
      \log \hat L - c_p
    )
  \right)_+, \quad
  A_{2,p} =\frac{1}{(p+1)!}\left(\hat L^{p+1}\right)_+,
\end{align*}
and $c_p=1+\frac12+\cdots+\frac 1p$ for $p\geq 1$, $c_0=0$.
Note that for an operator of the form $X=\sum\limits_{k\in\bbZ}X_k\hat\Lmd^k$,
the positive part $X_+$ is defined as $\sum\limits_{k\geq 0}X_k\hat\Lmd^k$, and
$\res X = X_0$.
The definition of $\log \hat L$ can be found in \cite{Extended Toda-1} and is omitted here.

The $\pp{\hat t^{2,p}}$-flows for $p \geq 0$, known as the positive flows,
constitute the usual Toda hierarchy,
while the logarithmic flows $\pp{\hat t^{1,p}}$ serve as an extension of this hierarchy.
The first flows of the extended Toda hierarchy \eqref{Extended Tdoa Hierarchy} are
$\pp{\hat t^{1,0}}=\p_{\hat x}$ and
\begin{equation}\label{first flow of ETH}
\veps\pfrac{\hat V}{\hat t^{2,0}} = \rme^{\hat U^{\oplus}}-\rme^{\hat U},\qquad
\veps\pfrac{\hat U}{\hat t^{2,0}} = \hat V-\hat V^{\ominus},
\end{equation}
here we introduce the following short notation: for a function $f$ we denote
\begin{equation}
  f^{\oplus} = \hat\Lmd f,\qquad f^{\ominus} = \hat\Lmd^{-1} f.
\end{equation}

The extended Toda hierarchy \eqref{Extended Tdoa Hierarchy} possesses a tau structure.
For any solution of the extended Toda hierarchy there exists a function $\hat\tau$ such that
\begin{align}
  \hat V&=\,\veps(\hat\Lmd-1)\pfrac{\log\hat\tau}{\hat t^{2,0}}, \label{Toda-tau-1}\\
  \hat U&=\, (\hat\Lmd-1)(1-\hat\Lmd^{-1})\log\hat\tau,  \label{Toda-tau-2}
\end{align}
\begin{equation}\label{tau structure of ETH-3}
  \veps(\hat\Lmd-1)
  \frac{\p^2\log\hat\tau}{\p\hat t^{\afa,p}\p\hat t^{\beta,q}}
=
  \pp{\hat t^{\beta,q}}\res A_{\afa,p}
\end{equation}
for $\afa,\beta=1,2$ and $p,q\geq 0$.
For example, we have the following two-point functions
\begin{align}
  \veps^2\frac{\p^2\log\hat\tau}{\p\hat t^{2,0}\p\hat t^{2,0}} &=\, \rme^{\hat U}, \label{Toda-2p-1}\\
  \veps^2\frac{\p^2\log\hat\tau}{\p\hat t^{2,0}\p\hat t^{2,1}}
  &=\, \frac12 (\hat V+\hat V^\ominus )\rme^{\hat U},   \label{Toda-2p-2}
\end{align}
which can be derived from \eqref{first flow of ETH} and \eqref{tau structure of ETH-3} directly.

The extended Toda hierarchy is the topological deformation of the Principal Hierarchy of $\hat M$ \eqref{F-Toda}
up to a certain Miura-type transformation.
Precisely speaking, let
\[\Delta\hat\mcalF=\sum\limits_{g\geq 1}\veps^{2g-2}\hat\mcalF_g\]
be the solution to the loop equation of $\hat M$,
and suppose $\hat v,\hat u$ satisfy the Principal Hierarchy of $\hat M$, then
\begin{align*}
  \hat V &=\, \frac{\hat\Lmd-1}{\veps\p_{\hat x}}
  \left(
    \hat v+\veps^2\frac{\p^2\Delta\hat\mcalF}{\p\hat x\p\hat t^{2,0}}
  \right), \\
  \hat U &=\, \frac{(\hat\Lmd-1)(1-\hat\Lmd^{-1})}{\veps^2\p_{\hat x}^2}
  \left(
    \hat u + \veps^2\frac{\p^2\Delta\hat\mcalF}{\p\hat x\p\hat t^{1,0}}
  \right)
\end{align*}
satisfy the extended Toda hierarchy \eqref{Extended Tdoa Hierarchy},
and the corresponding tau function satisfies
\[
  \log\hat\tau = \veps^{-2}\hat f + \Delta\hat\mcalF,
\]
where $\hat f$ be the genus-zero free energy of the tau-cover \eqref{tau cover}.

We remark that in the Principal Hierarchy of $\hat M$, we have
\[
  \pp{\hat t^{0,p}} = \pp{\hat t^{1,p}}\qquad\text{for $p\geq 0$},
\]
and $\pp{\hat t^{0,p}}=0$ for $p\leq -1$.
The Legendre field $\hat B$  of $\hat M$ defined in \eqref{Toda Legendre field}
yields a family of extended flows $\left\{\pp{\hat t^{\hat B,p}}\right\}_{p\in\bbZ}$,
the topological deformation of which gives a further extension of the extended Toda hierarchy.

\subsection{The extended Ablowitz-Ladik hierarchy}
In this subsection, let us recall the definition and some basic properties of the extended Ablowitz-Ladik hierarchy.
For more details, see \cite{AL-triham, ExtendedAL} and Sect.\,11 of \cite{GFM1}.
Introduce the Lax operator
\begin{align}
  L &=\, (1-Q\Lmd^{-1})^{-1}(\Lmd-P) \notag\\
  &=\,
    \Lmd + (Q-P) + Q(Q^--P^-)\Lmd^{-1} + \cdots, \label{AL-Lax operator}
\end{align}
where $P,Q$ are unknown functions, and the shift operator
\begin{equation}
  \Lmd = \rme^{\veps\p_x}.
\end{equation}
Here we introduce the short notation
\begin{equation}
  f^+=\Lmd f,\qquad f^-=\Lmd^{-1}f
\end{equation}
for any function $f$. The Ablowitz-Ladik hierarchy consists of following flows:
\begin{align}
  \veps\pfrac{L}{t^{2,p}} &=\, \frac{1}{(p+1)!}\left[\left(L^{p+1}\right)_+,L\right], \label{AL-positive flows}\\
  \veps\pfrac{L}{t^{0,-p-1}} &=\, (-1)^pp!\left[\left(L^{-p-1}\right)_-,L\right]\label{AL-negative flows}
\end{align}
for $p\geq 0$, which are known as the positive and negative flows respectively.
For example, the first positive flow has the form
\begin{equation}\label{AL-t20flow}
  \veps\pfrac{P}{t^{2,0}} = P(Q^+-Q), \quad
  \veps\pfrac{Q}{t^{2,0}} = Q(Q^+-Q^--P+P^-).
\end{equation}

Introduce new coordinates
\begin{equation}\label{AL-VU coordinate}
  V=Q-P,\qquad U=\log Q,
\end{equation}
then the $t^{2,0}$-flow \eqref{AL-t20flow} can be rewritten as
\begin{equation}\label{AL-t20flow-2}
  \veps\pfrac{V}{t^{2,0}} = V\rme^{U^+}-V^-\rme^U, \quad
  \veps\pfrac{U}{t^{2,0}} = \rme^{U^+}-\rme^U+V-V^{-}.
\end{equation}
In general, the dispersionless limit of Ablowitz-Ladik hierarchy \eqref{AL-positive flows}--\eqref{AL-negative flows}
coincides with the $t^{2,p}$- and $t^{0,-p-1}$-flows
in the Principal Hierarchy of the generalized Frobenius manifold $M$ with potential \eqref{AL-GFM-potential}
under the change of variables $V\mapsto v$ and $U\mapsto u$.

The Ablowitz-Ladik hierarchy \eqref{AL-positive flows}--\eqref{AL-negative flows} admits two families of extended flows
\begin{equation}\label{AL-log flow}
  \pp{t^{0,p}},\ \pp{t^{1,p}},\quad p\geq 0
\end{equation}
with
\begin{equation*}
  \pp{t^{0,0}} = \p_x,
\end{equation*}
and their dispersionless limits coincide which the $t^{0,p}$- and $t^{1,p}$-flows
of the Principal Hierarchy of $M$ under the change variables $V\mapsto v$ and $U\mapsto u$.
The Ablowitz-Ladik hierarchy \eqref{AL-positive flows}--\eqref{AL-negative flows} together with extended flows \eqref{AL-log flow}
is known as the \textit{extended Ablowitz-Ladik hierarchy} \cite{ExtendedAL}.
Although we do not know at the moment the explicit formulations of the extended flows \eqref{AL-log flow} in terms of Lax equations,
we have the following result.

\begin{prop}(Theorem 3.1 of \cite{ExtendedAL}).
The extended Ablowitz-Ladik hierarchy admits the following discrete symmetry
  \begin{align}
    P&\mapsto\, P^\oplus = \frac{P(Q^{++}-P^+)}{Q^+-P}, \label{discrete-sym of PQ-1}\\
    Q&\mapsto\, Q^\oplus = \frac{Q^+(Q^{++}-P^+)}{Q^+-P}. \label{discrete-sym of PQ-2}
  \end{align}
Moreover, the above symmetry is generated by the extended flow $\pp{t^{1,0}}$, i.e.
\begin{equation}
  P^\oplus = \rme^{\veps\pp{t^{1,0}}}P,\quad
  Q^\oplus = \rme^{\veps\pp{t^{1,0}}}Q.
\end{equation}
\end{prop}

We introduce the shift operator
\begin{equation}
  \hat\Lmd = \rme^{\veps\pp{t^{1,0}}},
\end{equation}
and we denote
\[
  f^\oplus = \hat\Lmd f,\qquad f^\ominus = \hat\Lmd^{-1}f.
\]
Here we use the same notations $\hat\Lmd$, $\oplus$ and $\ominus$ as in previous subsection, this is because
the flow $\pp{t^{1,0}}$ turns out to be identified with the flow $\pp{\hat t^{1,0}}=\pp{\hat t^{0,0}}=\p_{\hat x}$
of the extended Toda hierarchy via the linear reciprocal transformation \eqref{reciprocal}.

From \eqref{discrete-sym of PQ-1}--\eqref{discrete-sym of PQ-2} it follows that
\[
  P^\ominus = \frac{P(Q^--P^-)}{Q-P},\quad
  Q^\ominus = \frac{Q^-(Q^{--}-P^{--})}{Q^--P^-},
\]
so in the coordinates $V, U$ defined by \eqref{AL-VU coordinate}, the above discrete symmetry of the extended Ablowitz-Ladik hierarchy can be reformulated as
\begin{align}
  V^\oplus &=\, V^+ + \rme^{U^{++}}-\rme^{U^+}, \\
  \rme^{U^\oplus} &=\, \rme^{U^+}\frac{V^+ + \rme^{U^{++}}-\rme^{U^+}}{V + \rme^{U^{+}}-\rme^{U}}, \\
  V^\ominus &=\, V^- + \rme^{U^-}\frac{V^{--}}{V^-}-\rme^U\frac{V^-}{V}, \label{V ominus}\\
  \rme^{U^\ominus} &=\, \rme^{U^-}\frac{V^{--}}{V^-}.  \label{U ominus}
\end{align}

The extended Ablowitz-Ladik hierarchy admits a tau structure.
For any solution of the extended Ablowitz-Ladik hierarchy there exists a function $\tau$ such that
\begin{align}
  V&=\, \veps(\Lmd-1)\pfrac{\log\tau}{t^{2,0}}, \label{AL-tau-1}\\
  U&=\, (1-\Lmd^{-1})(\hat\Lmd-1)\log\tau, \label{AL-tau-2}\\
\label{AL-tau structure}
  \veps(\Lmd-1)\frac{\p^2\log\tau}{\p t^{2,p}\p t^{2,q}}
&=
\pp{t^{2,q}}
  \left(
    \frac{1}{(p+1)!}
    \res L^{p+1}
  \right),\quad p,q\geq 0.
\end{align}
For example, we have the following two-point functions
\begin{align}
  \veps^2\frac{\p^2\log\tau}{\p t^{2,0}\p t^{2,0}}
  &=\, (Q^--P^-)Q = V^-\rme^U,  \label{AL-2p-1}\\
  \veps^2\frac{\p^2\log\tau}{\p t^{2,0}\p t^{2,1}}
  &=\,
    \frac12\Big(
      Q^+Q(Q^--P^-)+QQ^-(Q^{--}-P^{--})
  \notag\\
  &\qquad
    +Q(Q-P)(Q^--P^-)+Q(Q^--P^-)^2\Big)  \notag\\
  &=\,\frac12\Big(
    \rme^{U+U^+}V^- + \rme^{U^-+U}V^{--} + (V^-+V)V^-\rme^U
  \Big),  \label{AL-2p-2}
\end{align}
which can be derived from \eqref{AL-t20flow} and \eqref{AL-tau structure} directly.

The extended Ablowitz-Ladik hierarchy is the topological deformation of the Principal Hierarchy of $M$ with potential \eqref{AL-GFM-potential}
up to a certain Miura-type transformation.
More precisely, let
\[\Delta\mcalF=\sum\limits_{g\geq 1}\veps^{2g-2}\mcalF_g\]
be the solution to the loop equation of $M$,
and suppose $v,u$ satisfy the Principal Hierarchy of $M$, then
it was conjectured in \cite{GFM1} and was proved in \cite{ExtendedAL} that
\begin{align*}
  V&=\, \frac{\Lmd-1}{\veps\p_x}
  \left(
    v+\veps^2\frac{\p^2\Delta\mcalF}{\p x\p t^{2,0}}
  \right), \\
  U&=\, \frac{1-\Lmd^{-1}}{\veps\p_x}
        \frac{\hat\Lmd-1}{\veps\p_{t^{1,0}}}
        \left(
          u+\veps^2\frac{\p^2\Delta\mcalF}{\p x\p t^{1,0}}
        \right)
\end{align*}
satisfy the extended Ablowitz-Ladik hierarchy under the change of variables \eqref{AL-VU coordinate},
and the corresponding tau function satisfies
\begin{equation}
  \log\tau = \veps^{-2}f + \Delta\mcalF,
\end{equation}
where $f$ is the genus-zero free energy of the tau-cover \eqref{tau cover}.

We remark that the Legendre-extended flows $\pp{t^{B,p}}$ of $M$
associated with the Legendre field $B=\p_v$ \eqref{AL-GFM-legendre field} coincide with $\pp{t^{1,p}}$ for $p\geq 0$, and
$\pp{t^{B,p}}=0$ for $p\leq -1$.

\subsection{Linear reciprocal transformation between hierarchies}

Since the generalized Frobenius manifold $M$ and the Frobenius manifold $\hat M$ with potentials \eqref{AL-GFM-potential} and \eqref{F-Toda} respectively are related by a generalized Legendre transformation,
Theorem \ref{thm: sol to loop eqn} asserts 
that the topological deformations of the Legendre-extended Principal Hierarchies of them
are related by the linear reciprocal transformation
\begin{equation}\label{AL-Toda reciprocal}
  \hat t^{1,p}=t^{1,p},\quad
  \hat t^{2,p}=t^{2,p},\quad
  \hat t^{\hat B,p}=t^{0,p},\quad
  \hat t^{\hat B,-p-1}=t^{0,-p-1},\quad p\ge 0,
\end{equation}
therefore,
the extended Toda hierarchy and the extended Ablowitz-Ladik hierarchy are related by this linear reciprocal transformation.

\begin{thm}
Suppose $V, U$ satisfy the extended Ablowitz-Ladik hierarchy
\eqref{AL-positive flows}, \eqref{AL-negative flows} and \eqref{AL-log flow} under the change of variables \eqref{AL-VU coordinate}, then
\begin{equation}\label{AL-Toda full-genera coord transf}
  \hat V = V+\rme^{U^+},\qquad
  \hat U = \log V^- + U
\end{equation}
satisfy the extended Toda hierarchy \eqref{Extended Tdoa Hierarchy} and its Legendre-extended flows $\pp{\hat t^{\hat B,p}}$
via the linear reciprocal transformation \eqref{AL-Toda reciprocal}.
In particular, $\hat V, \hat U$ satisfy the Toda equation \eqref{first flow of ETH}.
\end{thm}
\begin{proof}
  According to Theorem \ref{thm: sol to loop eqn},
  the topological deformations of the Legendre-extended Principal Hierarchies of
  $(M,B)$ and $(\hat M,\hat B)$ share the same tau function
  \[
    \tau = \hat\tau.
  \]
  Therefore, it suffices to verify that the $V, U$ and $\hat V,\hat U$ defined as in
  \eqref{AL-tau-1}--\eqref{AL-tau-2} and
  \eqref{Toda-tau-1}--\eqref{Toda-tau-2} are related by \eqref{AL-Toda full-genera coord transf}.
  Comparing their two-point functions
  \eqref{Toda-2p-1}--\eqref{Toda-2p-2} and
  \eqref{AL-2p-1}--\eqref{AL-2p-2}, we arrive at
  \begin{align*}
    \rme^{\hat U} &=\, V^-\rme^U, \\
    \hat V+\hat V^{\ominus} &=\,
    \rme^{U^+}+\rme^{U^-}\frac{V^{--}}{V^-}+V^-+V,
  \end{align*}
  therefore $\hat U = \log V^- + U$, and from \eqref{V ominus}--\eqref{U ominus} we obtain
  \[
    \hat V = \frac{1}{1+\hat\Lmd^{-1}}
    \left(
      \rme^{U^+}+\rme^{U^-}\frac{V^{--}}{V^-}+V^-+V
    \right)
  =
    V+\rme^{U^+}.
  \]
  The theorem is proved.
\end{proof}

\begin{rmk}
  The relation \eqref{AL-Toda full-genera coord transf} given in the above theorem can also be derived from spectral problems.
  Note that the Lax operator $L$ of the Ablowitz-Ladik hierarchy given in \eqref{AL-Lax operator} can also be represented in the form
  \[
    L=(\Lmd-\tilde P)(1-\tilde Q\Lmd^{-1})^{-1},
  \]
  where
  \[\tilde P=\frac{P^+(Q-P)}{Q^+-P^+},\quad
  \tilde Q=\frac{Q(Q^--P^-)}{Q-P}.\]
  Then it can be verified that the discrete symmetry
  \eqref{discrete-sym of PQ-1}--\eqref{discrete-sym of PQ-2} is the compatibility condition
  of the following spectral problem and the discrete evolution of the wave function
  \begin{align}
    L\psi &=\, \lmd\psi, \label{AL-spec problem}\\
    \hat\Lmd\psi &=\, (\Lmd-Q^+)\psi,
  \end{align}
  where $\lmd, \psi$ are spectral parameter and wave function respectively.
  The evolution equation of the wave function can be rewritten as
  \begin{equation} \label{241120-2107-1}
    \psi^+=\psi^\oplus + Q^+\psi.
  \end{equation}
  By applying $\Lmd$ on both sides of the above equation
  and by using \eqref{discrete-sym of PQ-1}--\eqref{discrete-sym of PQ-2}, we obtain
  \begin{equation} \label{241120-2107-2}
    \psi^{++} = \psi^{\oplus\oplus}
    + (Q^{\oplus +} + Q^{++})\psi^\oplus
    + Q^{++}Q^+\psi.
  \end{equation}
  The spectral problem \eqref{AL-spec problem} implies
  \[
    \psi^{++}-(\lmd+P^+)\psi^++\lmd Q^+ \psi = 0.
  \]
By substituting \eqref{241120-2107-1}--\eqref{241120-2107-2} into the above equation,
  and by using \eqref{discrete-sym of PQ-1}--\eqref{discrete-sym of PQ-2}, we obtain
  \[
    \psi^{\oplus\oplus}
    +\left(Q^++Q-P\right)^\oplus\psi^\oplus
    +\left[\left(Q^--P^-\right)Q\right]^\oplus\psi = \lmd\psi^\oplus,
  \]
  therefore
  \[
    \left(
      \hat\Lmd + (Q^++Q-P) + (Q^--P^-)Q\hat\Lmd^{-1}
    \right)\psi = \lmd\psi,
  \]
  which is the spectral problem \eqref{Toda Lax Operator},
  \eqref{AL-VU coordinate}, \eqref{AL-Toda full-genera coord transf} of the Toda hierarchy.
\end{rmk}

\section{Conclusion}
In this paper we establish the relationship between the Principal Hierarchies, the tau structures,
the Virasoro operators and the full-genera free energies
of two semisimple generalized Frobenius manifolds $M$ and $\hat M$ that are related by a generalized Legendre transformation.
An important step in the establishing of this relationship is the construction of the hierarchies of flows associated with the Legendre fields. We call these flows the Legendre flows, and we add them to the Principal Hierarchies of the generalized Frobenius manifolds to obtain the Legendre extended Principal Hierarchies. We show that the Legendre extended Principal Hierarchies and their topological deformations associated with the generalized Frobenius manifolds $M$ and $\hat M$ are related by a linear reciprocal transformation, and we also give two examples to illustrate such a relationship between integrable hierarchies associated with generalized Frobenius manifolds.

We remark that a generic generalized Frobenius manifold $M$ with non-flat unity can be transformed to a usual Frobenius manifold $\hat M$ with flat unity by a certain generalized Legendre transformation.
To see this, let $\{v^1,\dots v^n\}$ be a system of flat coordinates of $M$ such that its Euler vector field takes the form \eqref{DDE=0},
then all the vector fields $\pp{v^\afa}$ are quasi-homogeneous and Legendre; under the assumption that $\pp{v^\afa}$ is invertible for a certain $\afa\in\{1,\dots, n\}$, the Legendre transformation associated with
$B=\pp{v^\afa}$ transforms $M$ to a Frobenius manifold $\hat M$ with flat unity $e = \pp{\hat v^\afa}$ with respect to the flat metric of $\hat M$.

Note that the quasi-Miura transformation
 \eqref{quasi-Miura transf} yielded by the solution to the loop equation \eqref{loop equation-2308}
 transforms the bihamiltonian structure \eqref{Bihamiltonian structure} of the Legendre-extended Principal Hierarchy of a semisimple generalized Frobenius manifold $M$
 to that of its topological deformation.
 An important problem is whether this deformed bihamiltonian structure possesses the polynomiality property, i.e., whether the deformation terms of the bihamiltonian structure can be represented in terms of differential polynomials of the unknown functions of the integrable hierarchy.
 In the case of a usual semisimple Frobenius manifold with flat unity,
 the polynomiality of its deformed bihamiltonian structure is studied in  \cite{shadrin-1, shadrin-2, normal-form, Iglesias-Shadrin}, and it is proved recently in \cite{LWZ1}. The fact that a generic generalized Frobenius manifold with non-flat unity can be transformed to a usual Frobenius manifold with flat unity by a generalized Legendre transformation leads us to conjecture that
 the polynomiality property also holds true for
 the deformed bihamiltonian structure of a semisimple generalized Frobenius manifold with non-flat unity.
We will study this topic in subsequent publications.
\vskip 0.3cm
\noindent \textbf{Acknowledgements.}
This work is supported by
National Key R\&D Program of China (Grant No.\,2020YFE0204200),
NSFC No.\,12171268 and No.\,11725104. We would like to thank Di Yang and Zhe Wang for very helpful comments and discussions on this work.

\end{document}